\documentclass[11pt, reqno]{amsart}
\usepackage[margin=1in]{geometry}
\usepackage[svgnames]{xcolor}
\usepackage[pagebackref, colorlinks = true, linkcolor = DarkBlue, urlcolor  = DarkBlue, citecolor = DarkRed]{hyperref}
\usepackage{fontawesome}
\usepackage{amsthm}
\usepackage{amssymb}
\usepackage{physics}
\usepackage{graphicx}
\usepackage{stmaryrd}
\usepackage{cancel}
\usepackage[baseline]{euflag}
\usepackage[skip=0.5\baselineskip]{caption}

\usepackage{multirow}
\usepackage{tabularx}
\DeclareMathOperator*{\bigtimes}{\vartimes}
\renewcommand{\epsilon}{\varepsilon}

\def\evidence{\noindent{\bf Evidence:}\ }

\newtheorem{thm}{Theorem}[section]
\newtheorem*{thm*}{Theorem}
\newtheorem{lem}[thm]{Lemma}
\newtheorem{cor}[thm]{Corollary}
\newtheorem{ex}[thm]{Example}
\newtheorem{defi}[thm]{Definition}
\newtheorem{prop}[thm]{Proposition}
\newtheorem{conj}[thm]{Conjecture}
\newtheorem*{conj*}{Conjecture}
\newtheorem{remark}[thm]{Remark}
\newtheorem{question}[thm]{Question}

\title[Inclusion constants for free spectrahedra and applications]{Inclusion constants for free spectrahedra with applications to quantum incompatibility}
\date{\today}

\author{Andreas Bluhm}
\address{Univ.\ Grenoble Alpes, CNRS, Grenoble INP, LIG, 38000 Grenoble, France}
\email{andreas.bluhm@univ-grenoble-alpes.fr}

\author{Eric Evert}
\address{Department of Mathematics, University of Florida, Gainesville, FL, United States}
\email{ericevert@ufl.edu}

\author{Igor Klep}
\address{University of Ljubljana, Faculty of Mathematics and Physics \& University of Primorska, Famnit, Koper
\& Institute of Mathematics, Physics and Mechanics, Ljubljana, Slovenia}
\email{igor.klep@fmf.uni-lj.si}

\author{victor Magron}
\address{LAAS CNRS \& Institute of Mathematics, Toulouse, France}
\email{vmagron@laas.fr}

\author{Ion Nechita}
\address{Univ.~Toulouse, CNRS, Laboratoire de Physique Th\'eorique, France}
\email{nechita@irsamc.ups-tlse.fr}

\newcommand{\jewel}[0]{\text{\faDiamond}}
\newcommand{\cD}{ {\mathcal D} }
\newcommand{\cH}{ {\mathcal H} }
\newcommand{\cW}{ {\mathcal W} }
\newcommand{\cWmin}{\mathcal{W}^{\operatorname{min}}}
\newcommand{\cWmax}{\mathcal{W}^{\operatorname{max}}}
\newcommand{\R}{ {\mathbb R} }
\newcommand{\C}{ {\mathbb C} }
\newcommand{\deleuc}{ {\partial^{\operatorname{Euc}}} } 
\newcommand{\delfree}{ {\partial^{\operatorname{free}}} } 
\newcommand{\delmat}{ {\partial^{\operatorname{mat}}} }

\newcommand{\matco}{\operatorname{mat}^{\operatorname{co}}}
\newcommand{\MinBallSDP}[2]{ {\mathfrak{D}_{#1,#2}} }

\newcommand{\gamk}[1]{\frac{2 #1}{#1-1+\sqrt{1+#1}}}
\begin{document}
\begin{abstract}
Building on the matrix cube problem, inclusions of free spectrahedra have been used successfully to obtain relaxations of hard spectrahedral inclusion problems. The quality of such a relaxation is quantified by the inclusion constant associated with each free spectrahedron. While optimal values of inclusion constants were known in certain highly symmetric cases, no general method for computing them was available. In this work, we show that inclusion constants for Cartesian products of free simplices can be computed using methods from non-commutative polynomial optimization, together with a detailed analysis of the extreme points of the associated free spectrahedra. This analysis also yields new closed-form analytic expressions for these constants. As an application to quantum information theory, we prove new bounds on the amount of white noise that incompatible measurements can tolerate before they become compatible. In particular, we study the case of one dichotomic and one $k$-outcome measurement, as well as the case of four dichotomic qubit measurements.    
\end{abstract}

\maketitle

\tableofcontents

\section{Introduction}

Spectrahedra are the solution sets of linear matrix inequalities (LMIs). Under mild conditions,  a spectrahedron can be written as
\begin{equation} \label{eq:intro-spectrahedron}
    \mathcal D_A(1) :=\left\{x \in \mathbb R^g: I_d - \sum_{i=1}^g x_i A_i \succeq 0 \right\} \,,
\end{equation}
where $(A_1, \ldots, A_g)$ is a $g$-tuple of Hermitian matrices of size $d \in \mathbb{N}$ and $I_d$ is the identity matrix. A particular spectrahedron is the cube $[-1, 1]^g$. The matrix cube problem of Ben-Tal and Nemirovski \cite{Nemirovski2007} is to determine, given a spectrahedron $\mathcal D_A(1)$, whether 
\begin{equation*}
    [-1, 1]^g \subseteq \mathcal D_A(1) \, .
\end{equation*}
The matrix cube problem arises in robust control when checking system stability under uncertainty. Moreover, it includes maximizing a positive definite quadratic form over a unit cube which is a central problem in combinatorial optimization. Unfortunately, this problem is NP-hard in general \cite{ben-tal2002tractable, kellner2013containment}. In \cite{ben-tal2002tractable}, the authors gave a computable relaxation of the original problem, including a quantification of the error made in this relaxation. 
In \cite{helton_matricial_2013}, it was proposed that, to verify whether the spectrahedron $\mathcal D_A(1)$ is contained in the spectrahedron $\mathcal D_B(1)$, one should instead consider the inclusion of the corresponding free spectrahedra.
Given a $g$-tuple of Hermitian matrices $(A_1, \ldots, A_g)$ of size $d$, the corresponding free spectrahedron is defined as $\mathcal D_A := \bigsqcup_{n \in \mathbb N} \mathcal D_A(n)$, where
\begin{equation} \label{eq:intro-free-spectrahedron}
    \mathcal D_A(n) :=\left\{X \in SM_n(\mathbb C)^g: I_{dn} - \sum_{i=1}^g A_i \otimes X_i \succeq 0 \right\} \,.
\end{equation}
Here, by $SM_n(\mathbb C)$ we denote the Hermitian matrices of size $n$ and the tensor product is the usual Kronecker tensor product of matrices. Thus, where the spectrahedron contains all vectors $x$ that fulfill the LMI defined by $A$, the free spectrahedron contains all $g$-tuples of Hermitian matrices of arbitrary size that fulfill the same LMI. Thus, for $n=1$, equation \eqref{eq:intro-free-spectrahedron} coincides with equation \eqref{eq:intro-spectrahedron}, which explains our notation. 
Free spectrahedra are important examples within the more general class of matrix convex sets \cite{Effros1997, kriel2019intro}.

We say that $\mathcal D_A \subseteq \mathcal D_B$ if the inclusion holds on all levels, i.e., if $\mathcal D_A(n) \subseteq \mathcal D_B(n)$ for all $n \in \mathbb N$. Thus, clearly the inclusion of free spectrahedra $\mathcal D_A \subseteq \mathcal D_B$ implies the inclusion of the corresponding spectrahedra $\mathcal D_A(1) \subseteq \mathcal D_B(1)$. It was shown in \cite{helton_matricial_2013} that checking the inclusion of free spectrahedra is equivalent to a feasibility semidefinite program (SDP), which is efficiently solvable under mild assumptions. 
SDPs are a central class of convex optimization problems that can model a wide range of tasks in control theory, combinatorial optimization, and beyond (see, e.g., \cite{wolkowicz2012handbook}).

If $\mathcal D_A$ is not included in $\mathcal D_B$, we cannot conclude in general that $\mathcal D_A(1)$ is not contained in $\mathcal D_B(1)$. However, there exist so-called \emph{inclusion constants} $s$ for which $s\cdot \mathcal D_A \not \subseteq \mathcal D_B \implies \mathcal D_A(1) \not \subseteq \mathcal D_B(1)$. More precisely, $s\in[0,1]$ is an \emph{inclusion constant at level $m$} if the implication $\mathcal D_A(1) \not \subseteq \mathcal D_B(1) \implies s\cdot \mathcal D_A \subseteq \mathcal D_B$ holds for all $g$-tuples $B$ of Hermitian matrices of dimension $m$. In some cases, bounds on the inclusion constants are known \cite{helton2019dilations, davidson2016dilations}, in particular for free spectrahedra coming from unit balls of norms  
\cite{passer2018minimal}. 
However, a general method for computing these inclusion constants has so far been lacking. 
Our article takes a step toward closing this gap by showing that methods from non-commutative polynomial optimization 
\cite{burgdorf2016optimization,wang2021exploiting}
can be used to compute bounds on inclusion constants for free spectrahedra, such as those arising in quantum information theory.

As an application of our theory, we consider the incompatibility of quantum measurements, a basic notion in quantum information theory. 
A set of measurements is compatible if they can all be replaced by a single measurement. It has been known since the early days of quantum mechanics that there exist incompatible measurements, with position and momentum of a particle arguably the best known example \cite{Heisenberg1927, Bohr1928}. We refer the reader to \cite{Heinosaari2016} for an introduction to measurement incompatibility. In this work, we model measurements in quantum mechanics by positive operator-valued measures (POVMs). Incompatibility is of relevance for the development of quantum technologies, since compatible measurements can never violate Bell inequalities \cite{Fine1982}, which is needed for many quantum technologies such as device-independent quantum key distribution \cite{pironio2009device}. See \cite{Brunner2014} for more applications of Bell inequality violations. It has been shown in \cite{bluhm2018joint, bluhm2020compatibility} that the compatibility of measurements can be formulated as an inclusion problem of free spectrahedra: for dichotomic measurements, i.e., measurements with two outcomes, compatibility can be formulated as the inclusion problem of a universal free spectrahedron, called the matrix diamond (see \cite{davidson2016dilations} and Example \ref{ex:MatrixDiamond}), inside a free spectrahedron defined by the observables. The inclusion constants of the matrix diamond quantify how much incompatibility there is maximally in a setup with fixed dimension of the quantum system, a fixed number of measurements, and a fixed number of outcomes. The matrix diamond is dual to the matrix cube in the sense that while the latter is a matricial relaxation of the $\ell_\infty$ unit ball, the former is a matricial relaxation of the $\ell_1$ unit ball. Thus, knowing the inclusion constants of the matrix diamond and its generalization, the matrix jewel (see \cite{bluhm2020compatibility} and Example \ref{ex:MatrixJewel})
 helps us understand how noisy quantum measurement devices can be before any measurement incompatibility vanishes. 

To solve the convex optimization problems we encounter, it is important to know the extreme points of the convex sets we are optimizing over, as for example the set of optimizers of a linear functional over a convex set contains at least one of its extreme points. For a free spectrahedron $\mathcal D_A$, each level $\mathcal D_A(n)$ is a convex set, so it can be described by its extreme points, which are called \emph{Euclidean extreme points} to distinguish them from other types of extreme points we shall discuss next.  However, the various levels of the free spectrahedron are closely related and we can generalize the notion of a convex combination to allow for matrix tuples of different sizes. This leads to two further notions of extreme points which are called \emph{matrix extreme} and \emph{free extreme points}, respectively, see, e.g.,  \cite{evert2018extreme} and Section \ref{sec:extreme-points} of the present paper. Intuitively, a tuple $X$ is a matrix extreme point of a matrix convex set $K$ if it cannot be written as a nontrivial matrix convex combination of elements of $K$ \emph{with size less than or equal} to that of $X$. On the other hand,  $X$ is free extreme if it cannot be written as a nontrivial matrix convex combination of \emph{any} other elements of $K$. Thus, these different sets of extreme points are related as follows:
\begin{equation*}
    \mathrm{free~extreme}\; \subseteq\; \mathrm{matrix~extreme} \;\subseteq\; \mathrm{Euclidean~extreme}\,.
\end{equation*}
Extreme points have been studied extensively in the literature of matrix convex sets, e.g., see \cite{farenick2000extremal, kleski2014boundary,Davidson2019Noncommutative,klep2022facial}. We refer the reader to the recent survey of matrix convex sets and their extreme points \cite{evert2024extreme}.

\subsection{Guide to the paper}
The article is organized as follows: In Section \ref{sec:main-results}, we discuss the main results of this article, before reviewing in Section \ref{sec:preliminaries} relevant results from matrix convex sets, measurement incompatibility, and polynomial optimization. In Section \ref{sec:extreme-points}, we focus on extreme points of matrix convex sets. In Section \ref{sec:CartesianExtreme}, we investigate the relation between extreme points of two free spectrahedra and the extreme points of their Cartesian product. In Section \ref{sec:optimizers-direct-sums}, we use the results of the previous section in order to characterize the optimizers for some optimization problems inspired by measurement incompatibility which involve Cartesian products of free simplices. In Section \ref{sec:opti-problems-from-qit}, we formulate  the task of computing the maximal noise robustness of incompatible measurements as a polynomial optimization problem, which we can simplify using the characterization of optimizers derived in the previous section. Finally, in Section  \ref{sec:conclusion}, we provide some numerics and discuss the results of this article, concluding with some open questions.

\subsection{Main results} \label{sec:main-results}

The main contribution of this article is the computation of optimal inclusion constants of free spectrahedra that arise as direct sums of free simplices. For free spectrahedra $\cD_A$ and $\cD_B$, their direct sum is defined to be the free spectrahedron $\cD_{(A \otimes I, I \otimes B)}$. A \emph{free polytope} is a free spectrahedron $\cD_A$ as in equation \eqref{eq:intro-free-spectrahedron} where the matrices $A_i$ defining $\cD_A$ are diagonal in the same basis. A \emph{free simplex} is a bounded free polytope where the $A_i$ have dimension $g+1$. The \emph{matrix cube} in $g$-variables is the free polytope obtained by taking the union over $n$ of all $g$-tuples $(X_1,\dots,X_g)$ of $n \times n$ self-adjoint matrices that satisfy $X_i^2 \preceq I_n$ for all $i=1,\dots,g$. In the case $g=2$, we commonly call the matrix cube the matrix square. If $g=1$, then we call the matrix cube the matrix interval. Our results follow from our study of the different kinds of extreme points of such free spectrahedra and from using tools from non-commutative polynomial optimization. The new optimal inclusion constants shed light on the noise robustness of measurement incompatibility in quantum information theory.

Our main contribution is the optimal inclusion constant
at level two of the direct sum of two real free simplices, namely the free line (the free simplex in $1$ variable) and a particular free simplex in $k$ variables which is of interest in quantum information. The following is a restatement of Theorem \ref{thm:kSimplexPlusLineOptimum}:
\begin{thm*}
Fix an integer $k \geq 1$ and define 
\begin{equation} 
B_{j}(k) = I_2 \otimes \big(I_{k+1}-(k+1) E_{j}(k)\big) \qquad \mathrm{for} \qquad j=1,\dots,k
\end{equation}
and
\begin{equation} \label{eq:line-plus-simplex_dual-FS-2}
B_{k+1}(k) = \operatorname{diag} (1,-1) \otimes I_{k+1},
\end{equation}
so that $\mathcal D_{B(k)}$ is the direct sum of the real matrix interval and a real free simplex in $k$ variables. Let $s_A(2)$ be its optimal inclusion constant at level $2$. Then,
\begin{equation}
    s_A(2) = \frac{k-1+\sqrt{1+k}}{2k} \, . \label{eq:line+simplex-main}
\end{equation}
\end{thm*}
The free spectrahedron $\cD_{B(k)}$ is closely related to the matrix jewel \cite{bluhm2020compatibility}, see Example \ref{ex:MatrixJewel}. 

To prove the above theorem, we show in Theorem \ref{thm:MaxEigIsMinBallInclusionConst} that for general free polytopes, an upper bound for the corresponding inclusion constant can be found via an optimization problem involving both the free polytope and its dual free polytope, which is the maximal matrix convex set for the polar dual of $\mathcal D_A(1)$. We then construct explicit elements of the free polytopes that achieve $s_A(2)$ in this optimization problem. On the other hand, we show that $s_A(2)$ is a lower bound for the inclusion constant by constructing an exact solution to the corresponding  feasibility SDP of \cite{helton_matricial_2013}. 

En route to solving this feasibility SDP, we study in detail the extreme points of Cartesian products of a pair of free spectrahedra, with a particular focus on the case of Cartesian products of free simplices, such as the dual free polytope for the matrix jewel. In Section \ref{sec:CartesianExtreme}, we specialize known kernel-based classifications of extreme points of free spectrahedra to the Cartesian product setting. This leads to new insights into matrix extreme points of free spectrahedra that are not free extreme, points which were known to exist, but were considered difficult to construct \cite{epperly2024matex}. In particular, in Proposition \ref{prop:MatrixCartesianEuclidean} we give a simple construction that for a Cartesian product of free spectrahedra yields a tuple that is matrix extreme but not free extreme. 

The fact that we can prove the optimal form of $s_A(2)$ in the real setting is due to differences in the structure of the extreme points between the real and the complex setting. The following result is an immediate consequence of Theorem~\ref{thm:SimplexXIntervalRealMatEx} for the first assertion and Theorem~\ref{thm:SimpleXIntervalComplexMatEx} for the second:\looseness=-1
\begin{thm*}
    Let $\mathcal D_A=\mathcal D_{S_1} \times \mathcal D_{S_2}$ be the Cartesian product of two free simplices and let $(X,Y)$ be a real matrix extreme point in level two, i.e., in $\cD_A(2)$. Then, $X$ and $Y$ are free extreme in $\mathcal D_{S_1}$ and $\mathcal D_{S_2}$, respectively. However, in the complex setting there exist matrix extreme points at level two of a Cartesian product of two free simplices that are not free extreme.
\end{thm*}

In addition to the direct sum of a real line and a real $k$-simplex, we also study the inclusion constant (over the complexes) for level two of the matrix cube in four variables. The following is an informal restatement of Proposition \ref{prop:extreme-points-four-lines}:
\begin{thm*}
Let $\mathcal D_B$ be the direct sum of $4$ real lines and let $s_{B}(2)$ be its optimal inclusion constant at level $2$. Then,
\begin{equation*}
    s_{B}(2) \leq \frac{2}{\sqrt{13}} \,.
\end{equation*}
\end{thm*}
The direct sum of $4$ real lines is known as the matrix diamond \cite{davidson2016dilations} and is the free dual polytope to the matrix cube.  We prove the theorem by finding exact extreme points of the matrix cube that achieve this bound in the relevant optimization problem. 

It was shown in \cite{bluhm2018joint, bluhm2020compatibility, bluhm2022tensor} that optimal inclusion constants at level $d$ for certain direct sums of free simplices correspond to the minimal amount of noise that is necessary to make any measurements on $d$-dimensional quantum systems with a fixed number of outcomes compatible. Therefore, the  previous results can be applied to study the robustness of measurement incompatibility. At the heart of our results on quantum information theory lies the realization that we can obtain the relevant optimal inclusion constants from non-commutative polynomial optimization problems, while it was previously unknown how to obtain inclusion constants in a systematic way. We show in Theorem \ref{thm:lambda-max-formulation} that  computing the minimum compatibility degree of measurements with fixed outcomes in dimension $d$, i.e., the minimal amount of noise that is necessary to make any measurements on $d$-dimensional quantum systems with a fixed number of outcomes compatible, is equivalent to solving a polynomial optimization problem.
We can thus use tools from non-commutative polynomial optimization such as convergent hierarchies of SDPs to find lower bounds on the maximal robustness of measurement incompatibility to noise. We are especially interested in two measurement settings, namely the minimum compatibility degree of one dichotomic and one measurement with $k+1$ outcomes and the minimum compatibility degree of $4$ dichotomic qubit measurements (i.e., on quantum systems of dimension $2$). We write the compatibility degree of $g$ measurements in dimension $d$ with outcomes $\mathbf{k}=(k_1, \ldots, k_g)$ as $s(d,g,\mathbf{k})$, see Definition \ref{def:incompatibility-degree}.

For one dichotomic and one measurement with $k+1$ outcomes, from \eqref{eq:line+simplex-main} we obtain the maximal noise robustness of incompatibility if we restrict to real quantum mechanics, i.e., to quantum states and measurements that can be described by real-valued matrices. The following is a restatement of Theorem \ref{thm:results-for-2-plus-k}: 
\begin{thm*}
    The minimum compatibility degree for one dichotomic and one measurement with $k+1$ outcomes in dimension $2$ when restricted to real quantum mechanics is 
    \begin{equation}
    s_{\mathbb R}(2,2,(2,k+1)) = \frac{k-1+\sqrt{k+1}}{2k} \, .
\end{equation}
   A pair of maximally incompatible measurements is a computational basis measurement together with a Hadamard basis measurement that is padded with zeroes.
\end{thm*}
This result implies an upper bound for complex quantum mechanics, i.e., 
\begin{equation*}
    s_{\mathbb C}(2,2,(2,k+1)) \leq \frac{k-1+\sqrt{k+1}}{2k} \, .
\end{equation*}
In fact, we conjecture that the upper bound is optimal and holds for any $d \geq 2$:
\begin{conj*}
    We conjecture the minimum compatibility degree for one dichotomic and one measurement with $k+1$ outcomes in dimension $d$ to be 
    \begin{equation}
    s_{\mathbb C}(d,2,(2,k+1)) = \frac{k-1+\sqrt{k+1}}{2k} \, .
\end{equation}
\end{conj*}
For small $k$, we support this conjecture by solving the corresponding non-commutative polynomial optimization problem using the NPA hierarchy (see Section \ref{sec:line_simplex_num}). A crucial ingredient here are our characterizations of extreme points of Cartesian products of free simplices in Section \ref{sec:CartesianExtreme}.

For the compatibility degree of $4$ dichotomic qubit measurements, we obtain the following result, which is a restatement of Theorem \ref{thm:results-for-four-qubits}:
\begin{thm*}
    The minimum compatibility degree of $4$ dichotomic qubit measurements is bounded from above by
    \begin{equation*}
        s_\C(2,4, (2,2,2,2)) \leq \frac{2}{\sqrt{13}} \, .
    \end{equation*}
    The upper bound is the compatibility degree of a computational basis measurement and three projective measurements that correspond to equiangular lines in the $X-Y$ plane in the Bloch representation, see Fig.~\ref{fig:plot-4-meas-Bloch}.
\end{thm*}
We conjecture this bound to be optimal, i.e., that the measurements we found are among the most incompatible ones.
\begin{conj*}
    The minimum compatibility degree of $4$ dichotomic qubit measurements is
    \begin{equation*}
        s_\C(2,4, (2,2,2,2)) = \frac{2}{\sqrt{13}} \, .
    \end{equation*}
\end{conj*}
We support this conjecture by numerical lower bounds on the corresponding non-commutative polynomial optimization problem, obtained using the Lasserre hierarchy (see Section \ref{sec:four-qubits-numerics}).

In summary, our results demonstrate the usefulness of characterizing free extreme points for non-commutative polynomial optimization problems arising from quantum information theory.

\section{Preliminaries} \label{sec:preliminaries}

\subsection{Notation}
For simplicity, we will write $[n] := \{1, \ldots, n\}$. For $m$, $n \in \mathbb N$, the set of $n \times m$ matrices with entries from the field $\mathbb F = \mathbb {R}$ or $\mathbb{C}$ will be denoted by $M_{n, m}(\mathbb F)$ or $M_{n}(\mathbb F)$ if $m = n$. We will denote the symmetric $n \times n$ matrices by $SM_n(\mathbb R)$ and the Hermitian $n \times n$ matrices by $SM_n(\mathbb C)$. It will later be useful to write $M(\mathbb F) = \cup_{n \in \mathbb N} M_n(\mathbb F)$ and $SM(\mathbb F) = \cup_{n \in \mathbb N} SM_n(\mathbb F)$. Elements $X$, $Y \in M_n(\mathbb F)^g$  are unitarily equivalent if there exists a unitary $U \in M_n(\mathbb F)$ such that $(U A_1 U^\ast, \ldots, U A_g U^\ast) = (B_1, \ldots, B_g)$. We say that $A \in M(\mathbb F)^g$ is reducible over $\mathbb F$ if there exists $B$, $C \in M(\mathbb F)^g$ such that $(A_i)_{i \in [g]}$ is unitarily equivalent to $(B_i \oplus C_i)_{i \in [g]}$. Otherwise, $A$ is called irreducible over $\mathbb F$.

By writing $A \succeq 0$ for $A \in M_{n}(\mathbb F)$, we will mean that $A$ is positive semi-definite. For the identity matrix in dimension $n$, we will write $I_n$, where we will sometimes drop the index $n$. The Pauli matrices are
\begin{equation*}
    \sigma_X = \begin{pmatrix}
        0&1\\1&0
    \end{pmatrix}\,, \quad \sigma_Y = \begin{pmatrix}
        0&-i\\i&0
    \end{pmatrix}\,, \quad
\sigma_Z = \begin{pmatrix}
        1&0\\0&-1
    \end{pmatrix}\,.
    \end{equation*}
For a vector $x \in \mathbb F^n$, $\operatorname{diag}(x) \in M_n(\mathbb F)$ will be the diagonal matrix with $x$ on the diagonal. We write $\lambda_{\max}(A)$ for the largest eigenvalue of $A$.

\subsection{Matrix convex sets}

Matrix convex sets are special kinds of \textbf{graded sets} of the form $(S(n))_{n \in \mathbb N} \subseteq SM(\mathbb F)^g$ with $S(n) \subseteq SM_n(\mathbb F)^g$ for all $n \in \mathbb N$. To define them we first need a suitable generalization of the convex hull, see \cite{evert2018extreme,passer2018minimal}:
\begin{defi}[Matrix convex hull]
    Let $(S(n))_{n \in \mathbb N} \subseteq SM(\mathbb F)^g$ be a graded set. For all $i \in [s]$ and some $s \in \mathbb N$, let $X^{(i)} := (X^{(i)}_1, \ldots,  X^{(i)}_g) \in S_{n_i}$, i.e., $X_j^{(i)} \in SM_{n_i}(\mathbb F)$. Then, a \textbf{matrix convex combination} of $X^{(1)}, \ldots, X^{(s)}$ is an expression of the form 
    \begin{equation}\label{eq:matrix-convex-combination}
        \sum_{i\in [s]} V_i^\ast X^{(i)} V_i\, ,
    \end{equation}
    where $V_i \in M_{n_i, n}(\mathbb F)$ are matrices such that $\sum_{i \in [s]} V_i^\ast V_i = I_n$ and $V_i^\ast X^{(i)} V_i:= (V_i^\ast X^{(i)}_j V_i)_{j \in [g]}$. The \textbf{matrix convex hull} $\matco(S)$ of a graded set $S$ is the set of all matrix convex combinations of elements of $S$.
\end{defi}
With this, we can define matrix convex sets:
\begin{defi}
    A graded set $(S(n))_{n \in \mathbb N} \subseteq SM(\mathbb F)^g$ is \textbf{matrix convex} if it is closed under matrix convex combinations. 
\end{defi}
In particular, $\matco(S)$ is the smallest matrix convex set containing $S$. Alternatively, $S$ is a matrix convex set if and only if it is closed under direct sums and conjugations by isometries. A matrix convex set $S$ is called compact if and only if $S(n)$ is compact for all $n \in \mathbb N$. 

\subsubsection{Matrix convex sets closed under complex conjugation}
Say a matrix convex set $K \subset SM(\C)^g$ is \textbf{closed under complex conjugation} if for all $X \in K$, one has that $\overline{X} = (\overline{X}_1,\dots,\overline{X}_g) \in K$. This is equivalent to the assumption that $K$ is closed under transposition, i.e., $X \in K$ if and only if $X^T = (X_1^T,\dots,X_g^T) \in K$. The property of being closed under complex conjugation plays an important role in matrix convexity and leads to a cleaner theory. For example, as we shall see in Section \ref{sec:extreme-points}, closure under complex conjugation yields a cleaner theory of extreme points. 

Given any matrix convex set $K^\C$, one naturally obtains a real matrix convex set $(K^\C)_r$ by setting $(K^\C)_r: = K^\C \cap SM(\R)^g$. On the other hand, given a real matrix convex set $K^\R \subset SM(\R)^g$, the complexification of $K^\R$ is defined by
\[
(K^\R)_c:= \left\{X+iY \in SM(\C)^g: \ X,Y \in M(\R)^g \ \ \mathrm{and} \ \begin{pmatrix}
    X & -Y \\
    Y & X
\end{pmatrix} \in K^\R\right\}.
\]
As a consequence of \cite[Theorem 3.1]{blecher2025real}, $(K^\R)_c$ is a complex matrix convex set. Furthermore, $(K^\R)_c$ is equal to the complex matrix convex hull of $K^\R$. It is easy to see that $(K^\R)_c$ is closed under complex conjugation. 

\begin{lem}
\label{lem:RealPartComplexified}
    Let $K^\C \subset SM(\C)^g$ be a complex matrix convex set. Then $Z \in ((K^\C)_r)_c$ if and only if $Z,\overline{Z} \in K^\C$. As an immediate consequence, $((K^\C)_r)_c = K^\C \cap \overline{K^\C}$. Moreover, $((K^\C)_r)_c =K^\C$ 
    if and only if $K^\C$ is closed under complex conjugation. 
\end{lem}
\begin{proof}
Let $Z = X+iY \in K^\C$ where $X$ and $Y$ are tuples of real matrices. Then by definition $Z \in ((K^\C)_r)_c$ if and only if 
\[
W:= \begin{pmatrix}
    X & -Y \\
    Y & X
\end{pmatrix} \in (K^\C)_r. 
\]
As $(K^\C)_r = K^\C \cap SM(\R)^g$, this 
is in turn equivalent to $W \in K^\C$. However, $W$ is unitarily equivalent to $Z \oplus \overline{Z}$, so it follows that $Z \in ((K^\C)_r)_c$ if and only if $Z,\overline{Z} \in K^\C$. 
\end{proof}

Throughout the article, all of the matrix convex sets considered will be closed under complex conjugation. Thus, Lemma \ref{lem:RealPartComplexified} allows us to naturally pass between the real part and the complexification. Under the assumption of closure under complex conjugation, much of the theory easily transfers from the complex to the real setting, or vice versa, via arguments along these lines. However, some aspects of the theory break down when passing between reals and complexes. See \cite{blecher2025real,blecher2025realII,evert2025matrix} for further discussion on the relationship between real and complex matrix convex sets. 

\subsubsection{Free spectrahedra and their inclusion constants} 
A special class of matrix convex sets that will be relevant for this article are the \textbf{free spectrahedra}. They are matricial relaxations of ordinary spectrahedra that are solution sets to LMIs that appear, for example, in optimization theory \cite{helton_matricial_2013}. 
\begin{defi}
    Let $g$, $d \in \mathbb N$,  $A \in SM_d(\mathbb F)^g$. Then, the \textbf{$n$-th level of the free spectrahedron defined by $A$} is defined as 
    \begin{equation*}
        \mathcal D^{\mathbb F}_A(n) := \left\{X \in SM_n(\mathbb F)^g : \sum_{i \in [g]} A_i \otimes X_i \preceq I_{nd} \right\} \, .
    \end{equation*}
    The \textbf{free spectrahedron defined by $A$} is the disjoint union of all levels, i.e.,
    \begin{equation*}
        \mathcal D^{\mathbb F}_A := \bigsqcup_{n \in \mathbb N} \mathcal D^{\mathbb F}_A(n) \, .
    \end{equation*}
    To keep notation light, we will drop the superscript ${\mathbb F}$ if the field is unimportant or clear from the context.
 \end{defi}
Thus, $\mathcal D_A(1)$ is the spectrahedron defined by $A$. If all $A_i$ are diagonal in the same basis, $\mathcal D_A$ is a \textbf{free polytope}. Furthermore, if $\mathcal{D}_A$ is a bounded free polytope and $d=g+1$, then $\mathcal{D}_A$ is a \textbf{free simplex}. We note that the defining pencil of a free polytope can be taken to be real. Therefore, every free polytope is closed under complex conjugation. That is, if $\cD_A$ is a free polytope and $X \in \cD_A$, then $\overline{X} \in \cD_A$. 

Of particular interest in this article are Cartesian products and direct sums of free spectrahedra. Given free spectrahedra $\cD_A \subset SM(\mathbb{F})^g$ and $\cD_B \subset SM(\mathbb{F})^h$, the $n$-th level of the Cartesian product of $\cD_A$ and $\cD_B$ is defined as
\[
(\cD_A \times \cD_B) (n):= \{ (X,Y) \in SM_n (\mathbb{F})^{g+h} : \ X \in \cD_A (n) \ \mathrm{and} \ Y \in \cD_B (n)\}.
\]
We then have $\cD_A \times \cD_B = \cup_n (\cD_A \times \cD_B) (n)$. As one may expect, the Cartesian product of free spectrahedra is a free spectrahedron, see Lemma \ref{lem:CartesianIsFreeSpec}. We define the direct sum of $\cD_A$ and $\cD_B$ to be the free spectrahedron $\cD_{(A \otimes I, I \otimes B)}$. In the case that $\cD_A$ and $\cD_B$ are free polytopes, the Cartesian product and direct sum of $\cD_A$ and $\cD_B$ are related to each other in a natural dual sense, see Proposition \ref{prop:CartesianPolyDualIsDirect}.

We frequently use the notation
\[
\Lambda_A (X) := \sum_{i \in [g]} A_i \otimes X_i \qquad \mathrm{and} \qquad L_A(X) := I-\Lambda_A(X).
\]
With this notation, $\mathcal{D}_A = \{X \in SM(\mathbb F)^g: L_A (X) \succeq 0\}$. We say $A$ is a \textbf{minimal defining tuple} for $\cD_A$ if the dimension $d$ of $A$ is as small as possible. That is, $A \in SM_d(\mathbb F)^g$  is minimal if and only if $B \in SM_{d_1}(\mathbb F)^g$ and $\cD_A = \cD_B$ implies $d \leq d_1$. See \cite{helton_matricial_2013,zalar2017operator} for an in-depth discussion of minimal defining tuples. 

Sometimes, it will be useful to look at a homogeneous version of a free spectrahedron:

\begin{defi}
    Let $g$, $d \in \mathbb N$,  $A \in SM_d(\mathbb F)^g$. Then, the \textbf{$n$-th level of the homogeneous free spectrahedron defined by $A$} is defined as 
    \begin{equation*}
        \mathcal H_{(I,A)}(n) := \left\{X \in SM_n(\mathbb F)^{g+1} : I_d \otimes X_0 + \Lambda_A (X) \succeq 0 \right\} \, .
    \end{equation*}
    The \textbf{homogeneous free spectrahedron defined by $A$} is the disjoint union of all levels, i.e.,
    \begin{equation*}
        \mathcal H_{(I,A)} := \bigsqcup_{n \in \mathbb N} \mathcal H_{(I,A)}(n) \, .
    \end{equation*}
\end{defi}
Note that due to the sign conventions used above, the free spectrahedron $\mathcal{D}_A$ is naturally associated to the homogeneous free spectrahedron $\cH_{(I,-A)}$.

Before we proceed, we give some examples of free spectrahedra.

\begin{ex}[{Matrix diamond \cite{davidson2016dilations}}] \label{ex:MatrixDiamond}
    The \textbf{matrix diamond} is the disjoint union $\mathcal D^{\mathbb F}_{\diamond, g} := \bigsqcup_{n \in \mathbb N} \mathcal D^{\mathbb F}_{\diamond, g}(n)$, where
\begin{equation*}
    \mathcal D^{\mathbb F}_{\diamond, g}(n) := \left\{X \in SM_n(\mathbb F)^g:\sum_{i = 1}^g \epsilon_i X_i \preceq I_n~\forall \epsilon \in \{\pm 1\}^g\right\}.
\end{equation*}
Note that we could have written $\mathcal D^{\mathbb F}_{\diamond, g}$ as $\mathcal D_A^{\mathbb F}$ with $A_i = \left(\bigotimes_{j \in [i-1]} I_2\right) \otimes \operatorname{diag}(+1, -1) \otimes \left(\bigotimes_{j = i+1}^g I_2 \right)$.
\end{ex}

\begin{ex}[Matrix jewel \cite{bluhm2020compatibility}]
\label{ex:MatrixJewel}
    The \textbf{matrix jewel} generalizes the matrix diamond. It is defined as $\mathcal D^{\mathbb F}_{\jewel, \mathbf k} := \bigsqcup_{n \in \mathbb N} \mathcal D^{\mathbb F}_{\jewel, \mathbf k}(n)$, where $\mathbf k = (k_1, \ldots, k_g) \in \mathbb N^g$ and 
\begin{align*}
   & \mathcal D^{\mathbb F}_{\jewel, \mathbf k}(n) := \\ &\left\{X \in SM_n(\mathbb F)^{(\sum_i (k_i-1) )}: \sum_{j = 1}^g \sum_{i = 1}^{k_j-1} \left[ \left(\bigotimes_{\ell \in [j-1]} I_{k_\ell} \right) \otimes \operatorname{diag}(v_i^{(k_j)}) \otimes \left(\bigotimes_{\ell = j+1}^g I_{k_{\ell}}\right) \right] \otimes X_{ij} \preceq I_{(k_1 \cdot \ldots \cdot k_g)n}\right\}\, .
\end{align*}
The vectors $v_i^{(k)} \in \mathbb R^{k}$ are defined for $i \in [k -1]$ as 
\begin{equation*}
    v^{(k)}_i(\ell) := - \frac{2}{k} + 2\delta_{\ell, i} \qquad \forall \ell \in [k]\, , 
\end{equation*}
where $\delta_{i,j}$ is the Kronecker delta. For $k = 2$, we get $v_1 = (1, -1)$. Therefore,
\begin{equation*}
    \mathcal D^{\mathbb F}_{\jewel, 2^{\times g}} = \mathcal D^{\mathbb F}_{\diamond, g}.
\end{equation*}
In general, $\mathcal D^{\mathbb F}_{\jewel, \mathbf{k}}$ is a direct sum of free simplices.
\end{ex}

\subsubsection{Inclusion constants}
For two free spectrahedra $\mathcal D_A$ and $\mathcal D_B$, we define the inclusion $\mathcal D_A \subseteq \mathcal D_B$ to mean $\mathcal D_A(n) \subseteq \mathcal D_B(n)$ for all $n \in \mathbb N$. Obviously, the implication $\mathcal D_A \subseteq \mathcal D_B$ thus implies $\mathcal D_A(1) \subseteq \mathcal D_B(1)$. The reverse implication is usually not true, but we can make it true by shrinking $\mathcal D_A$ sufficiently. Therefore, we define the set of inclusion constants as follows:
\begin{defi}[Inclusion constants]
Let $g \in \mathbb N$ and $A \in SM(\mathbb F)^g$, Then, the set of \textbf{inclusion constants} is defined as
    \begin{equation*}
    \Delta^{\mathbb F}_A(m) := \{s \in [0,1]: \mathcal D_A(1) \subseteq \mathcal D_B(1) \implies  s \cdot \mathcal D_A \subseteq \mathcal D_B ~\forall B \in SM_m(\mathbb F)^g\}.
\end{equation*}
In particular, we are interested in the largest element of this set, i.e., 
\begin{equation*}
    s^{\mathbb F}_A(m) := \sup \{s \in \Delta_A^{\mathbb F}(m)\},
\end{equation*}
which we will call the \textbf{maximal inclusion constant}. We will again drop the superscript ${\mathbb F}$ if the field is unimportant or clear from the context.
\end{defi}
For more about inclusion constants, we refer the reader to \cite{helton2019dilations}. 

For fixed $B \in \mathcal SM_D(\mathbb F)^g$ and $A \in SM_d(\mathbb F)^g$ such that $\mathcal D_A(1)$ is bounded, whether $\mathcal D_A \subseteq \mathcal D_B$ holds can be checked with a semidefinite program \cite{helton_matricial_2013}. In fact, it is sufficient to consider the level $D$ inclusion $\mathcal D_A(D) \subseteq \mathcal D_B(D)$, not all levels. We need to check whether the following SDP is feasible:
\begin{align} \label{eq:inclusion-free-spectra-SDP}
    &C:=(C_{pq})^d_{p,q =1} \succeq 0\, , \nonumber\\
    &\sum_{p = 1}^d C_{pp} = I_D\, , \\
    &\sum_{p,q=1}^d (A_i)_{pq} C_{pq} = B_i \qquad \forall i \in [g]\, ,\nonumber\\
    &C_{pq} \in SM_D(\mathbb F) \qquad \forall p, q \in [d] \, .\nonumber
\end{align}
The SDP amounts to checking whether there is a unital completely positive map which maps $A_i \mapsto B_i$. Here, $C$ is the Choi matrix corresponding to this completely positive map. The first line checks complete positivity, the second unitality, and the third whether indeed  $A_i \mapsto B_i$. One of the main results of \cite{helton_matricial_2013} is that inclusion of free spectrahedra is in one-to-one correspondence with the existence of such completely positive maps. We refer the reader to \cite{convexPsatz}\footnote{The paper \cite{helton_matricial_2013} was written and posted to arXiv \url{https://arxiv.org/abs/1003.0908} earlier than \cite{convexPsatz} \url{https://arxiv.org/abs/1102.4859} but got published later.} for an extension to unbounded $\cD_A$.

\subsubsection{Minimal and maximal matrix convex sets} Given a closed convex set $\mathcal{C} \subset \R^g$, there typically exist many matrix convex sets $K \subset SM(\mathbb F)^g$ such that $K(1) = \mathcal{C}$. Among these, the two most notable ones are the {\bf minimal} and {\bf maximal} matrix convex sets associated to $\mathcal{C}$. These are respectively the smallest and largest matrix convex sets whose first level is equal to $\mathcal{C}$. 

More precisely, the {\bf minimal matrix convex set} associated to $\mathcal{C}$, denoted $\cWmin (\mathcal{C})$, is simply the matrix convex hull of $\mathcal{C}$. That is,
\[
\cWmin_{\mathbb F}(\mathcal{C}) := \matco(\mathcal{C}).
\]
On the other hand, the {\bf maximal matrix convex set} associated to $\mathcal{C}$, denoted $\cWmax_{\mathbb F} (\mathcal{C})$ is the set of all matrix tuples $X$ that satisfy the linear inequalities satisfied by $\mathcal{C}$. That is,
\[
\cWmax_{\mathbb F}(\mathcal{C}) := \{ X \in SM (\mathbb F)^g : \Lambda_a(X) \preceq a_0 I \mathrm{ \ whenever \ } \Lambda_a(x) \leq a_0 \mathrm{\ for \ all \ } x \in \mathcal{C}\}. 
\]
In the case where $0$ is in the interior of $\mathcal{C}$, one sees that $a_0$ must be strictly positive, hence the above definition simplifies to
\[
\cWmax_{\mathbb F}(\mathcal{C}) := \{ X \in SM (\mathbb F)^g : L_a(X) \succeq 0 \mathrm{ \ whenever \ } L_a(x) \succeq 0 \mathrm{\ for \ all \ } x \in \mathcal{C}\}.
\]
As a consequence, if $\cD_A$ if a free polytope, then $\cWmax_{\mathbb F}(\cD^{\mathbb F}_A(1))= \cD^{\mathbb F}_A$. Again, we will most often drop the $\mathbb F$ for readability. It is straightforward to check that minimal and maximal matrix convex sets over a convex set $\mathcal{C} \subset \R^g$ are closed under complex conjugation.

Finally, we can also define inclusion constants for polytopes, namely as how much you have to shrink the maximal matrix convex set to make it fit inside the minimal matrix convex set.
\begin{defi}
    Let $\mathcal P \subset \mathbb R^g$. Then, the set of \textbf{inclusion constants} is defined as 
    \begin{equation*}
        \Delta^{\mathbb F}_{\mathcal P}(m) := \{s \in [0,1]: s\cdot \mathcal W_{\mathbb F}^{\min}(\mathcal P)(m) \subseteq \mathcal W_{\mathbb F}^{\max}(\mathcal P)(m)\} \, .
    \end{equation*}
    The \textbf{maximal inclusion constant} is defined as $s^{\mathbb F}_{\mathcal P}(m) := \sup\{s \in \Delta^{\mathbb F}_{\mathcal P}(m)\}$. We will again drop the superscript ${\mathbb F}$ if the field is unimportant or clear from the context.
\end{defi}

See \cite{davidson2016dilations,passer2018minimal} for further discussion of minimal and maximal convex sets. We mention that simplices are of particular note with respect to inclusion constants. Namely, \cite[Theorem 4.1]{passer2018minimal} shows that if $K \subset \mathbb{R}^g$ is a convex set, then $\mathcal W^{\min}( K) = \mathcal W^{\max}( K)$ if and only if $K$ is a simplex. Thus, the maximal inclusion constant for a simplex is $1$. 

\subsubsection{Polar duals}

Given a closed convex set $\mathcal{C}$, we let $\mathcal{C}^\bullet$ denote the (classical) {\bf polar dual} of $\mathcal{C}$. That is 
\[
\mathcal{C}^\bullet = \{y \in \R^g : \langle x,y \rangle \leq 1 \mathrm{\ for \ all \ } x \in \mathcal{C}\}.
\]
Also, for a closed matrix convex set $K$, we let $K^\circ$ denote the {\bf free polar dual} of $K$ \cite{helton2017tracial}, defined by
\[
K^\circ = \{Y \in SM(\mathbb F)^g : \Lambda_Y (X) \preceq I \mathrm{\ for \ all \ } X \in K\} = \cap_{X \in K} \cD_X.
\]
A straightforward check shows that $K(1)^\bullet = (K^\circ) (1)$. Additionally, if $0$ is in the interior of $K$, then $(K^\circ)^\circ = K$. 

From \cite[Theorem 4.7]{davidson2016dilations}, for a closed convex set $\mathcal{C}$, one has
\[
\cWmin(\mathcal{C})^\circ = \cWmax(\mathcal{C}^\bullet).
\]
Furthermore, if $0 \in \mathcal{C}$, then 
\begin{equation*}
    \cWmax(\mathcal{C})^\circ = \cWmin(\mathcal{C}^\bullet)\, .
\end{equation*}
In particular, if $\cD_A$ is a free polytope then we obtain
\begin{equation}
\label{eq:PolytopeDual}
\cD_A^\circ = \cWmin (\cD_A(1)^\bullet).
\end{equation}

Since \cite{davidson2016dilations} is originally proved over the complexes and plays an important role in our theory, we show that the result also holds over the reals. To do this, we will make use of the following lightly modified version of \cite[Proposition 4.3 (5,6)]{helton2017tracial} (see also \cite{Effros1997}).

\begin{lem}
\label{lemma:RealBipolar}
Let $K^\mathbb{F} \subset SM(\mathbb{F})^g$ be a closed matrix convex set. Then $(K^\mathbb{F})^{\circ \circ} = \matco (K^\mathbb{F} \cup \{0\})$
\end{lem}
\begin{proof}
    There are only two small differences in our statement compared to \cite[Proposition 4.3 (5)]{helton2017tracial}. The first is that \cite[Proposition 4.3 (5)]{helton2017tracial} works over $\C$, while our result is stated over $\mathbb{F}$. The main tool in the proof of \cite[Proposition 4.3 (5)]{helton2017tracial} is an appeal to the Effros-Winkler matricial Hahn-Banach theorem \cite{Effros1997}, which while proved over $\C$, is easily shown to hold over $\R$, so the choice of $\R$ or $\C$ does not matter.
    
    The second difference is that we do not require the closure of the matrix convex hull. However, since we assume that $K^\mathbb{F}$ is closed and since the matrix convex hull of $\{0\}$ is evidently a closed matrix convex set, \cite[Theorem 3.14]{evert2025matrix} shows that $\matco (K^\mathbb{F} \cup \{0\})$ is closed. Thus, the closure is not needed in our setting.
\end{proof}

Now, as a consequence of \cite[Lemma 6.6, Proposition 6.8]{blecher2025real}, the complex matrix convex set $\cWmax_\C(\mathcal{C})$ is the complexification of the real matrix convex set $\cWmax_\R(\mathcal{C})$, and similarly $\cWmin_\C(\mathcal{C})= (\cWmin_\R(\mathcal{C}))_c$. Using this together with the following lemma shows that \cite[Theorem 4.7]{davidson2016dilations} holds in the real setting.

\begin{lem}
    Let $K_1^\C$ and $K_2^\C$ be closed complex matrix convex sets and assume that $K_2^\C$ is closed under complex conjugation. Let $K_1^\R:= (K_1^\C)_r$ and $K_2^\R:=(K_2^\C)_r$ denote the real parts of $K_1^\C$ and $K_2^\C$, respectively. Then $K_1^\C = (K_2^\C)^\circ$ if and only if $K_1^\R = (K_2^\R)^\circ$. 
\end{lem}
\begin{proof}
    We first show that $K_1^\C = (K_2^\C)^\circ$ implies $K_1^\R = (K_2^\R)^\circ$. By definition $K_1^\C = (K_2^\C)^\circ$ means $K_1^\C = \cap_{Z \in K_2^\C} \ \cD^\C_Z$. It follows that
    \[
    K_1^\R = \left(\cap_{Z \in K_2^\C} \ \cD^\C_Z\right) \cap SM(\R)^g = \cap_{Z \in K_2^\C} \left(\cD^\C_Z\cap SM(\R)^g \right)  = \cap_{Z \in K_2^\C} \ \cD^\R_Z. 
    \]
    Now, since $K_2^\mathbb{C}$ is closed under complex conjugation, $Z \in K_2^\C$ if and only if $Z \oplus \overline{Z} \in K_2^\C$, so 
    \[
    \cap_{Z \in K_2^\C} \ \cD^\R_Z =  \cap_{Z \in K_2^\C} \ \cD^\R_{Z \oplus \overline{Z}}.
    \]
    Now, writing $Z = X+iY$ with $X,Y$ real, we have that $Z$ is unitarily equivalent to 
    \[
    W = \begin{pmatrix}
        X & -Y \\
        Y & X
    \end{pmatrix}, 
    \]
    It follows that $Z \in K_2^\C$ if and only if $W \in K_2^\R$ and also that $\cD_{Z \oplus \overline{Z}} ^\R = \cD_W^\R$. We now obtain that
    \[
    \cap_{Z \in K_2^\C} \ \cD^\R_{Z \oplus \overline{Z}} = \cap_{W \in K_2^\R} \ \cD_W^\R = (K_2^\R)^\circ, 
    \]
    which completes the forward implication.

     For the reverse implication, assume that $K_1^\R = (K_2^\R)^\circ$. Since $0$ is always an element of the free polar dual, Lemma \ref{lemma:RealBipolar} shows that this is equivalent to $K_1^\R = (\matco(K_2^\R \cup\{0\}))^\circ$. Furthermore, we have $0 \in K_1^\R \subset K_1^\C$, so Lemma \ref{lemma:RealBipolar} shows that $K_1^\C = (K_1^\C)^{\circ\circ}$. Set $K_3^\C := (K_1^\C)^\circ$, so that $K_1^\C = (K_3^\C)^\circ$. By the first part of the proof, $K_1^\R = (K_3^\R)^\circ$. 
    
    Again using Lemma \ref{lemma:RealBipolar}, we obtain $K_3^\R = \matco(K_2^\R \cup\{0\})$. Since the complexification of a matrix convex set can be obtained by taking the matrix convex hull of the set over $\C$, it is now straightforward to check that $K_3^\C = \matco(K_2^\C \cup\{0\})$. It follows that $(K_3^\C)^\circ = (K_2^\C)^\circ,$ hence  $(K_2^\C)^\circ= K_1^\C$. 
\end{proof}

Additionally, if $\cD_A$ is a free polytope, then $\cWmax(\cD_A(1)^\bullet)$ is a free polytope which we call the {\bf dual free polytope} of $\cD_A$. We let $\cD_A^\square$ denote the dual free polytope of $\cD_A$. 

We record a small observation that will be useful later:
\begin{prop} \label{prop:spectrahedron-mcset-equivalent}
    Let $\mathcal D_A \subset SM(\mathbb F)^g$ be a free polytope and $X \in SM(\mathbb F)^g$. Then,
    \begin{enumerate}
        \item $\mathcal D_A(1) \subseteq \mathcal D_X(1) \iff X \in \mathcal W^{\max}(\mathcal D_A(1)^\bullet)=\cD_A^\square$.
        \item $\mathcal D_A \subseteq \mathcal D_X \iff X \in \mathcal W^{\min}(\mathcal D_A(1)^\bullet)$.
    \end{enumerate}
\end{prop}
\begin{proof}
    Since $\mathcal D_A$ is a free spectrahedron, it is easy to see that the polytope $\mathcal D_A(1)$ is closed and contains $0$. Thus, by the bipolar theorem, $\mathcal D_A(1)^{\bullet \bullet} = \mathcal D_A(1)$. The relation $\mathcal D_A(1) \subseteq \mathcal D_X(1)$ is equivalent to 
    \begin{equation*}
        \sum_{i \in [g]} y_i X_i \leq I \qquad \forall y \in \mathcal D_A(1)\, .
    \end{equation*}
    Using $\mathcal D_A(1)^{\bullet \bullet} = \mathcal D_A(1)$ and $\cD_A(1)^{\bullet \bullet} = \{y \in \mathbb R^g : L_x(y) \succeq 0~ \forall x \in \cD_A(1)^{\bullet}\}$, we can likewise write 
    \begin{equation*}
\cWmax(\mathcal D_A(1)^\bullet) = \{ X \in SM (\mathbb F)^g : L_y(X) \succeq 0 ~\forall y \in \mathcal D_A(1)\}.
    \end{equation*}
This proves the first assertion. For the second assertion, it is enough to realize that $\mathcal D_A \subseteq \mathcal D_X$ is equivalent to $X \in (\mathcal D_A)^\circ$.  The assertion then follows from equation \eqref{eq:PolytopeDual}.
\end{proof}

\subsection{Measurement incompatibility} \label{sec:incompatibility}
The problem from quantum information theory we will be concerned with in this article is the incompatibility of measurements. First, we will need some notation from quantum information theory. We refer the reader to \cite{Heinosaari2011} for an introduction to the mathematical description of quantum systems. For quantum systems of dimension $d \in \mathbb N$, the set of quantum states is given as
\begin{equation*}
    \mathcal S(\mathbb C^d) = \left\{\rho \in M_d(\mathbb C) : \rho \succeq 0, \tr(\rho) = 1\right\}.
\end{equation*}
Thus, quantum states can be thought of as a non-commutative generalization of probability distributions. 

Measurements on these quantum states will be described as follows: A measurement with outcomes in a finite set $\Sigma$ is a tuple of matrices $(E_i )_{i \in \Sigma} \subseteq M_d(\mathbb C)$ such that
\begin{enumerate}
    \item $E_i \succeq 0$ for all $i \in \Sigma$;
    \item $\sum_{i \in \Sigma} E_i = I_d$. 
\end{enumerate}
Sets of matrices that satisfy the above are known as \textbf{positive operator-valued measures (POVMs)}. Without loss of generality, we will set $\Sigma = [n]$ for $n = |\Sigma|$, which can be thought of as a labeling of the outcomes of the measurement.

Now, we consider a set of measurements, i.e.,  $(E_{i|x})_{i \in [k_x]}$ is a POVM for every $x \in [g]$. Here, $g \in \mathbb N$ is the number of measurements and $k_x \in \mathbb N$ is the number of outcomes for the measurement with label $x$. 
\begin{defi}[Compatible measurements]
    Let $g \in \mathbb N$, $d \in \mathbb N$, and $k_x \in \mathbb N$ for all $x \in [g]$. Let $(E_{i|x})_{i \in [k_x]}$, $x \in [g]$ be a collection of $g$ $d$-dimensional POVMs. These measurements are \textbf{compatible} if there exists another $d$-dimensional POVM $(J_{i_1, \ldots, i_g})_{i_1 \in [k_1], \ldots, i_g \in [k_g]}$ such that 
    \begin{equation*}
        E_{i|x} = \sum_{\substack{i_y \in [k_y]\\y \neq x}} J_{i_1, \ldots, i_{x-1}, i, i_{x+1}, \ldots, i_{g}} \qquad \forall i \in [k_x], \forall x \in [g] \, .
    \end{equation*}
    If the measurements are not compatible, they are called \textbf{incompatible}.
\end{defi}
An equivalent way to write this condition is that there exist a finite set $\Lambda$, a POVM $(M_\lambda)_{\lambda \in \Lambda}$ and conditional probability distributions $(p_{i|\lambda, x})_{i \in [k_x]}$ for all $x \in [g]$ such that 
\begin{equation}\label{eq:joint-measurement}
    E_{i|x} = \sum_{\lambda \in \Lambda} p_{i|\lambda, x} M_\lambda \, .
\end{equation}

This can be interpreted as follows: The measurements $(E_{i|x})_i$ are compatible if we can obtain the same outcome statistics by performing a single measurement $(M_\lambda)_\lambda$ instead and randomly assigning its outcomes to the measurements $(E_{i|x})_i$ according to the probability distributions $(p_{i|\lambda, x})_{i}$. We call the latter process \textbf{classical post-processing} because it can be done on a normal computer and does not involve any quantum systems. For a general introduction to measurement incompatibility, we refer the reader to \cite{Heinosaari2016}.

In \cite{bluhm2018joint, bluhm2020compatibility}, it was realized that there measurement incompatibility and matrix convex sets are closely related. A key observation is the following:
\begin{prop} \label{prop:incompatibility-spectrahedra}
    Let $g$, $d \in \mathbb N$ and $k_x \in \mathbb N$ for all $x \in [g]$. Moreover, let $E_{i|x} \in SM_d(\mathbb C)$ for all $i \in [k_x]$, $x \in [g]$. We define $A_{i|x} = 2E_{i|x} - \frac{2}{k_i}I_d$ for all $i \in [k_x-1]$, $x \in [g]$ and $A = (A_{i|x})_{i \in [k_x-1], x \in [g]}$. Then, 
  \begin{enumerate}
      \item $\mathcal D^{\mathbb C}_{\jewel, \mathbf{k}}(1) \subseteq \mathcal D^{\mathbb C}_A(1)$ if and only if $(E_{i|x})_{i \in [k_x]}$ is a POVM for all $x \in [g]$.
      \item $\mathcal D^{\mathbb C}_{\jewel, \mathbf{k}} \subseteq \mathcal D^{\mathbb C}_A$ if and only if $(E_{i|x})_{i \in [k_x]}$ are compatible POVMs for all $x \in [g]$.
  \end{enumerate}
  Equivalently,
  \begin{enumerate}
        \item $A \in \mathcal W_{\mathbb C}^{\max}(\mathcal D^{\mathbb C}_{\jewel, \mathbf{k}}(1)^\bullet) $ if and only if $(E_{i|x})_{i \in [k_x]}$ is a POVM for all $x \in [g]$.
      \item $A \in \mathcal W_{\mathbb C}^{\min}(\mathcal D^{\mathbb C}_{\jewel, \mathbf{k}}(1)^\bullet) $ if and only if $(E_{i|x})_{i \in [k_x]}$ are compatible POVMs for all $x \in [g]$.
  \end{enumerate}
\end{prop}
\begin{proof}
    The first two statements can be found in \cite{bluhm2020compatibility}. The last two statements follow from Proposition \ref{prop:spectrahedron-mcset-equivalent} or can directly be found in \cite{bluhm2022tensor}.
\end{proof}

In a similar way to noise destroying entanglement, which is part of what makes building quantum computers challenging, incompatible measurements can be made compatible by adding a sufficient amount of noise. In this article, we will consider uniform noise, also called white noise.
\begin{defi}
    Let $k \in \mathbb N$ and let $(E_i)_{i \in [k]}$ be a POVM. Let $s \in [0,1]$ be a noise parameter. Then, we define the POVM $(E_i(s))_{i \in [k]}$ with
    \begin{equation*}
        E_i(s) := s E_i + (1-s)\frac{I}{k}
    \end{equation*}
    as the noisy version of $(E_i)_{i \in [k]}$.
\end{defi}

Now, fixing the dimension of the quantum system, the number of measurements we want to consider, and their outcomes, we ask for the minimal amount of noise that makes any such measurements compatible. This is captured by our next definition. 
\begin{defi}[Compatibility degree]\label{def:incompatibility-degree}
Given a $g$-tuple of measurements $E=(E_{\cdot|x})_{x \in [g]}$ on a $d$-dimensional Hilbert space, having respectively $k_1, \ldots, k_g$ outcomes, define their \textbf{compatibility degree} as 
$$s_\C(E):=\max\{s \in [0,1]: \{(E_{i|x}(s))_{i \in [k_x]}\}_{x \in [g]}\text{ are compatible}\}.$$

Consider a measurement setting given by positive integers $d$, $g \in \mathbb N$ and $k_x \in \mathbb N$ for all $x \in [g]$. The \textbf{minimum compatibility degree} of this measurement setting is defined as  
\begin{align*}
    s_\C(d,g,(k_1, \ldots, k_g)) := \min&\{ s_\C(E) : \text{ $E$ is a $g$-tuple of measurements}\\
    &\text{on a $d$-dimensional Hilbert space with $k_1, \ldots, k_g$ outcomes}\}.
\end{align*}
    If $k_1 = \ldots = k_g = 2$, we write $s_\C(d,g)$ instead.
\end{defi}

We can use the results in Proposition \ref{prop:incompatibility-spectrahedra} to make a connection between the minimum compatibility degree and the maximal inclusion constant of the matrix jewel. The following has been proved in \cite{bluhm2020compatibility, bluhm2022tensor}:
\begin{prop} \label{eq:inclusion-constant-compatibility-degree}
    Let $d$, $g \in \mathbb N$ and $k_x \in \mathbb N$ for all $x \in [g]$. Moreover, let $A$ be the matrices defining $\mathcal D^{\mathbb C}_{\jewel, \mathbf{k}}(1)$ and $\mathcal P = \mathcal D^{\mathbb C}_{\jewel, \mathbf{k}}(1)^\bullet$.
    Then, 
    \begin{equation*}
        s^\C_A(d) = s_\C(d,g,\mathbf{k}) = s^\C_{\mathcal P}(d) \, .
    \end{equation*}
\end{prop}

In the following, we collect what is known about the minimum compatibility degree for different measurement settings. Some of the results make use of Proposition \ref{eq:inclusion-constant-compatibility-degree}.
\begin{prop} \label{prop:bounds-incompat-degree}
    Let $d$, $g \in \mathbb N$ and $k_x \in \mathbb N$ for all $x \in [g]$. Then,
    \begin{enumerate}
       \item For any permutation $\sigma$ on $g$ elements, $s_\C(d,g,(k_1, \ldots, k_g)) = s_\C(d,g,(k_{\sigma(1)}, \ldots, k_{\sigma(g)}))$.
        \item $s_\C(d,g,(k_1, \ldots, k_g, 1)) = s_\C(d,g,(k_1, \ldots, k_g))$.
        \item For $k^\prime_x \in \mathbb N$ such that $k_x^\prime \geq k_x$ $ \forall x \in [g]$, it holds that $s_\C(d,g,(k_1^\prime, \ldots, k_g^\prime)) \leq s_\C(d,g,(k_1, \ldots, k_g))$.
        \item For $d^\prime \in \mathbb N$ such that $d^\prime \geq d$, it holds that $s_\C(d^\prime,g,(k_1, \ldots, k_g)) \leq s_\C(d,g,(k_1, \ldots, k_g))$.
        \item For $g^\prime \in \mathbb N$ such that $g^\prime \geq g$, it holds that $s_\C(d,g^\prime) \leq s_\C(d,g)$.
        \item $s_\C(d,g) \geq 1/\sqrt{g}$.
        \item $s_\C(d,g) = 1/\sqrt{g}$ for $d \geq 2^{\lceil (g-1)/2\rceil}$.
        \item $s_\C(d,g) \geq \tau(d)$ with $\tau(d) = 4^{-n} \binom{2n}{n}$, $\mathrm{with~}n=\lfloor d/2 \rfloor$. Asymptotically, $\tau(d)$ behaves as $\sqrt{2/(\pi d)}$.
        \item $s_\C(d,g,(k_1, \ldots, k_g))\geq 1/g$.
        \item $s_\C(d,g,(k_1, \ldots, k_g))\geq \frac{g+k_{\max}d}{g(1+k_{\max}d)}$ for $k_{\max} = \max_{i \in [g]} k_i$.
        \item $s_\C(d,2,(k_1, k_2))\geq \frac{1}{2}\left(1+\frac{1}{\sqrt{k_1k_2}+1}\right)$
    \end{enumerate}
\end{prop}
\begin{proof}
    The first two points are straightforward to prove. The third point has been proved in \cite{bluhm2020compatibility}. The fourth point  follows along the lines of \cite[Proposition 3.6]{bluhm2018joint}. The fifth point follows from the first three. The sixth and seventh point were proved in \cite{bluhm2018joint}, based on results by \cite{passer2018minimal}. The eight point was proved in \cite{bluhm2022steering}, using ideas from \cite{ben-tal2002tractable}. The ninth point can be found in \cite{Heinosaari2016}. The tenth point was proved in \cite[Proposition 6.7]{bluhm2020compatibility}, using an idea of \cite{Heinosaari2014}. The last bound has been proved in \cite[Section 3.2.3]{Designolle2019}.
\end{proof}

\begin{remark}
    In \cite{bluhm2025random}, the compatibility degree of random quantum measurements has been considered. It has been shown that, in many natural settings, generic random measurements are almost maximally incompatible: in the limit where the Hilbert space dimension goes to infinity, the compatibility degree of independent random quantum measurements approaches the minimum value $s_\C(d,g,\mathbf k)$ from Definition \ref{def:incompatibility-degree}. 
\end{remark}

\medskip 

We now introduce a real-valued version of measurement compatibility. First, note that real measurements $E_{i|x} \in SM_d(\R)$ are compatible if and only if there exists a joint measurement with \emph{real}-valued effects $M_\lambda$ satisfying equation \eqref{eq:joint-measurement}. Indeed, one can simply take for $M$ the real part of an a priori complex joint measurement. 

In order to obtain a real version of Proposition \ref{prop:incompatibility-spectrahedra} for measurements defined by real matrices, we will need the following result. It shows that for real matrices, the inclusion of the real free spectrahedra is equivalent to the inclusion of the  complex versions. 

\begin{lem} \label{lem:real-complex-inclusion}
    Let $A,B$ be $g$-tuples of real symmetric matrices. Then, for all $n \geq 1$, we have:
    $$\mathcal D_A^\R(2n) \subseteq \mathcal D_B^\R(2n) \implies \mathcal D^\C_A(n) \subseteq \mathcal D^\C_B(n) \implies \mathcal D_A^\R(n) \subseteq \mathcal D_B^\R(n).$$
    In particular,
    $$\mathcal D^\C_A \subseteq \mathcal D^\C_B \iff \mathcal D_A^\R \subseteq \mathcal D_B^\R.$$
    
\end{lem}
\begin{proof}
    The only non-trivial implication follows from the following observation: given a selfadjoint matrix $X$, we have
    $$X \succeq 0 \iff \begin{bmatrix}
        \Re X & \Im X\\
        -\Im X & \Re X
    \end{bmatrix} \succeq 0,$$
    where $\Re X$ and $\Im X$ are the real, resp.~imaginary parts of the matrix $X$.
\end{proof}
Using this lemma, we immediately obtain a real version of Proposition \ref{prop:incompatibility-spectrahedra}:
\begin{cor} \label{cor:incompatibility-spectrahedra-real}
       Let $g$, $d \in \mathbb N$ and $k_x \in \mathbb N$ for all $x \in [g]$. Moreover, let $E_{i|x} \in SM_d(\mathbb R)$ for all $i \in [k_x]$, $x \in [g]$. We define $A_{i|x} = 2E_{i|x} - \frac{2}{k_i}I_d$ for all $i \in [k_x-1]$, $x \in [g]$ and $A = (A_{i|x})_{i \in [k_x-1], x \in [g]}$. Then, 
  \begin{enumerate}
      \item $\mathcal D^\R_{\jewel, \mathbf{k}}(1) \subseteq \mathcal D^\R_A(1)$ if and only if $(E_{i|x})_{i \in [k_x]}$ is a POVM for all $x \in [g]$.
      \item $\mathcal D^\R_{\jewel, \mathbf{k}} \subseteq \mathcal D^\R_A$ if and only if $(E_{i|x})_{i \in [k_x]}$ are compatible POVMs for all $x \in [g]$.
  \end{enumerate}
\end{cor}

We also introduce the real version of the minimum compatibility degree from Definition \ref{def:incompatibility-degree}:

\begin{align*}
        s_\R(d,g,(k_1, \ldots, k_g)) := \min&\{s \in [0,1]: \{(E_{i|x}(s))_{i \in [k_x]}\}_{x \in [g]}\text{ are compatible for all sets}\\ &~ \text{ of $g$ \emph{real valued} $d$-dimensional POVMs } \{(E_{i|x})_{i \in [k_x]}\}_{x \in [g]}\} \, .
    \end{align*}

\subsection{Non-commutative polynomial optimization}

In this section, we briefly present different hierarchies of semidefinite programs (SDPs) to solve problems in commutative and non-commutative polynomial optimization. We will focus on the problems they are designed to solve and refer the reader to the corresponding papers for details about the actual hierarchies.

\subsubsection{Lasserre hierarchy}
\label{sec:lasserre}
Let $p: \mathbb R^n \to \mathbb R$ and $g_i: \mathbb R^n \to \mathbb R$ be polynomials for all $i\in [r]$, $r \in \mathbb N$. Then, we are interested in the following optimization problem:
\begin{align}
    \mathrm{minimize} \qquad & \qquad p(x)\nonumber\\
    \mathrm{such~that} \qquad & \qquad g_i(x) \geq 0 \quad \forall i \in [r] \label{eq:lasserre} \\
    \qquad & \qquad x \in \mathbb R^n\nonumber
\end{align}
Here, the set $K$ defined by the polynomial inequalities $g_i(x)\geq0$ is assumed to be compact. Moreover, the technical Assumption 4.1 from \cite{lasserre2001global} is assumed to hold.

It was shown in \cite{lasserre2001global} that there exists a sequence of lower bounds $p_1 \leq p_2 \leq \ldots$ such that each $p_i$ can be computed by an SDP and such that $\lim_{i \to \infty} p_i = p^\ast$, where $p^\ast$ is the value of \eqref{eq:lasserre}. Moreover, \cite{lasserre2001global} gave optimality conditions of  Karush–Kuhn–Tucker (KKT) type. The underlying idea of the hierarchy is to express $p(x)-p^\ast $ as a sum of squares.

\subsubsection{NPA hierarchy}

For non-commutative polynomials, we are interested in the following problem: 
Let $p$, $g_i$ be non-commutative polynomials  in $n$ variables for all $i\in [r]$, $r \in \mathbb N$, $\mathcal H$ be a complex Hilbert space, 
$\mathcal{B}(\mathcal H)$ be the set of bounded linear operators on $\mathcal H$, 
$\psi$ be a unit vector in that Hilbert space and $(X_1, \ldots, X_n) \in \mathcal{B}(\mathcal H)^n$ be a vector of self-adjoint operators. Then, we want to solve:
\begin{align}
        \mathrm{minimize} \qquad & \qquad \langle \psi, p(X) \psi \rangle \nonumber\\
    \mathrm{such~that} \qquad & \qquad g_i(X) \succeq 0 \quad \forall i \in [r] \label{eq:npa} 
\end{align}
where the optimization runs over all $\mathcal H$, $\psi$ and $X$. Here, the polynomials $g_i$ have to be such that the positivity domain $K$, i.e., the set of  $X$ such that $g_i(X)\succeq 0$ for all $i \in [r]$, has an associated Archimedean quadratic module. 

In \cite{HM04,pironio2010convergent}, it was shown that again there exists a sequence of lower bounds $p_1 \leq p_2 \leq \ldots$ such that each $p_i$ can be computed by an SDP and such that $\lim_{i \to \infty} p_i = p^\ast$, where $p^\ast$ is the value of \eqref{eq:npa}. The SDP hierarchy is referred to as NPA hierarchy after Navascués, Pironio, and Acín, the authors of \cite{pironio2010convergent}. The authors also gave a 
sufficient criterion for the $i$-th relaxation to be optimal, i.e., $p_i = p^\ast$. It is based on the associated moment matrix having a flat extension. Recently, optimality conditions of KKT type for this hierarchy were given in \cite{araujo2311first}.

\if{
\subsubsection{Trace polynomial hierarchy}
A trace polynomial in symmetric non-commutative variables $(X_1, \ldots, X_n)$ is polynomial in the symmetric variables themselves and normalized traces of their products. A trace polynomial is pure if it consists only of traces. Let $p$ be a symmetric pure trace polynomial in $n$ symmetric non-commutative variables and $g_i$ be symmetric trace polynomials in $n$ symmetric non-commutative variables. Then, the optimization problem to solve is:
\begin{align}
        \mathrm{minimize} \qquad & \qquad \langle \psi, p(X) \rangle \nonumber\\
    \mathrm{such~that} \qquad & \qquad g_i(X) \succeq 0 \quad \forall i \in [r] \label{eq:trace-poly} 
\end{align}
Here, the optimization runs over finite von Neumann algebras with given tracial state and their elements $X$. The positivity domain is again assumed to have an associated Archimedean quadratic module. 

In \cite{klep2022optimization} two of the present authors together with a collaborator showed that again there exists a series of lower bounds $p_1 \leq p_2 \leq \ldots$ such that each $p_i$ can be computed by an SDP and such that $\lim_{p \to \infty} = p^\ast$, where $p^\ast$ is the value of \eqref{eq:trace-poly}. The underlying idea is to prove a cyclic Positivstellensatz for trace polynomials.
}\fi

\section{Extreme points of matrix convex sets} \label{sec:extreme-points}

We now discuss in detail the extreme points of matrix convex sets in addition to two optimization problems over the extreme points of matrix convex sets. 

\subsection{General results}
\subsubsection{Types of extreme points in matrix convex sets}

Matrix convex sets have several different notions of extreme points, which all have different uses. We will be interested in this article in Euclidean extreme points, matrix extreme points, Arveson extreme points, and free extreme points. We will start by reviewing their definitions. For more background on the different types of extreme points, we refer the reader to \cite{evert2018extreme, epperly2024matex} and the recent survey \cite{evert2024extreme}.

\begin{defi}[Euclidean extreme points]
    Let $(S(n))_{n \in \mathbb N}$ be a matrix convex set. Then, $X \in S(n)$ is called an \textbf{Euclidean extreme point} if it is an extreme point of the convex set $S(n)$. We write $\deleuc S$ for the Euclidean extreme points of $S$ and $\deleuc S(n)$ for those in $S(n)$.
\end{defi}
Before we can go on to define matrix convex points, we need to define a special class of matrix convex combinations.
\begin{defi}[Proper matrix convex combination]
    A matrix convex combination as in equation \eqref{eq:matrix-convex-combination} is \textbf{proper} if all the matrices $V_i \in M_{n_i, n}(\mathbb F)$ correspond to surjective linear transformations.
\end{defi}

\begin{defi}[Matrix extreme points]
     Let $(S(n))_{n \in \mathbb N} \subset SM(\mathbb F)^g$ be a matrix convex set. Then, $X \in S(n)$ is called a \textbf{matrix extreme point} if from any expression of $X$ as a proper matrix convex combination of points $X^{(i)} \in S(n_i)$, $i \in [s]$ for some $s \in \mathbb N$, it follows that $n_i=n$ and each $X^{(i)}$ unitarily equivalent to $X$ for all $i \in [s]$. We write $\delmat S$ for the matrix extreme points of $S$ and $\delmat S (n)$ for those in $S(n)$.
\end{defi}

\begin{defi}[Arveson extreme points]
Let $(S(n))_{n \in \mathbb N} \subset SM(\mathbb F)^g$ be a matrix convex set. Then, $X \in S(n)$ is called an \textbf{Arveson extreme point} if 
\begin{equation*}
    Y = \begin{pmatrix}
    X & \beta \\ \beta^\ast & \gamma
    \end{pmatrix} \in S_{n+m}
\end{equation*}
    implies $\beta =0$ for $\beta \in M_{n,m}(\mathbb F)^g$ and $\gamma \in SM_{m}(\mathbb F)^g$.
\end{defi}

In the above definition, the tuple $Y$ is called a \textbf{dilation} of $X$. 

\begin{defi}[Free extreme points]
    Let $(S(n))_{n \in \mathbb N}$ be a matrix convex set. Then, $X \in S(n)$ is called a \textbf{free extreme point} if from any expression of $X$ as a matrix convex combination of points $X^{(i)} \in S_{n_i}$, $i \in [s]$ for some $s \in \mathbb N$, it follows for each $i \in [s]$ that either $n_i = n$ and $X^{(i)}$ unitarily equivalent to $X$ or $n_i > n$ and there exists a $Z^{(i)} \in S$ such that $X^{(i)}$ is unitarily equivalent to $X \oplus Z^{(i)}$, where the direct sum is understood entry-wise. We write $\delfree S $ for the free extreme points of $S$ and $\delfree S(n)$ for those in $S(n)$.
\end{defi}
In \cite{evert2018extreme}, free extreme points are called absolute extreme points. From \cite[Theorem 1.1]{evert2018extreme} it follows that for a matrix convex set $(S(n))_{n \in \mathbb N}$, a tuple $X \in  (S(n))_{n \in \mathbb N}$ is free extreme if and only if it is an irreducible Arveson extreme point of $(S(n))_{n \in \mathbb N}$.

It can be seen that free extreme points are in particular matrix extreme points and matrix extreme points are in particular Euclidean extreme points \cite{evert2018extreme}.

\begin{remark} \label{rem:extreme-points-free-simplex}
    The extreme points of the free simplex $\mathcal D_A$ are simple to characterize: As a consequence of \cite[Theorem 6.5]{evert2018extreme} we have
    \[
    \delfree \mathcal D_A = \delmat \mathcal D_A = \deleuc \mathcal D_A(1)
    \]
    That is, all matrix and free extreme points of $\cD_A$ lie at level $1$ of $\cD_A$. Additionally, a tuple $X \in \cD_A$ is an Arveson extreme point of $\cD_A$ if and only if it is unitarily equivalent to a direct sum of points in $\deleuc \mathcal D_A(1)$. In particular, all Arveson extreme points of $\cD_A$ are tuples of commuting matrices. 
\end{remark}

\subsection{Kernel conditions for extreme points}

In the upcoming Section \ref{sec:CartesianExtreme}, we shall be studying in greater detail the extreme points of a Cartesian product of free spectrahedra. A key tool that we will use is a kernel-based method for classifying matrix extreme points of free spectrahedra.  

Given $g+1$ tuples of matrices $(A_0,A_1,\dots,A_g) = (A_0,A)$ and $(X_0,X_1,\dots,X_g)$, let 
\[
\Lambda_{(A_0,A)} (X_0,X) := A_0 \otimes X_0+A_1\otimes X_1 +\dots A_g \otimes X_g = A_0 \otimes X_0 + \Lambda_A (X). 
\]
We also extend this notation to tuples $(A_0,A,B)$ and $(X_0,X,Y)$ of compatible size in the natural way. That is,
\[
\Lambda_{(A_0,A,B)} (X_0,X,Y) = A_0 \otimes X_0 + \Lambda_A (X) + \Lambda_B (Y).
\]
With this notation observe that
\[
L_A (X) = I \otimes I - \Lambda_A (X) = \Lambda_{(I,-A)} (I,X). 
\]

Given a free spectrahedron $\cD_A$, recall that $\cH_{(I,-A)}$ is the corresponding homogeneous free spectrahedron which is defined by
\[
\cH_{(I,-A)} = \{(H,W) \in SM(\C)^{g+1}: \ \Lambda_{(I,-A)} (H,W) \succeq 0. 
\]
A tuple $(H,W)$ is said to be an extreme ray of $\cH_{(I,-A)}$ if $(H,W)$ cannot be expressed as a nontrivial conic combination of elements of $\cH_{(I,-A)}$. A key insight of \cite{kriel2019intro} is that matrix convex combinations of elements of $\cD_A(n)$ correspond to conic combinations of elements of $\cH_{(I,-A)}(n)$. A consequence of this insight is that the matrix extreme points of $\cD_A$ correspond to extreme rays of $\cH_{(I,-A)}$.

\begin{prop} \cite[Proposition 6.5 (c)]{kriel2019intro}
    Let $\cD_A \subset SM(\mathbb F)^g$ be a free spectrahedron and let $\cH_{(I,-A)}$ be the corresponding homogeneous free spectrahedron. Then a tuple $X \in SM(\mathbb F)^g$ is matrix extreme in $\cD_A$ if and only if $(I,X)$ is an extreme ray of $\cH_{(I,-A)}$.
\end{prop}

Combining the above with \cite[Corollary 4]{ramana1995geometric} leads to the following kernel-based classification of the matrix extreme points of a free spectrahedron.

\begin{prop}
\cite[Theorem 2.6 (3)]{epperly2024matex}
    \label{prop:MatExker}
    Let $A \in SM_d (\C)^g$ and let $X \in \cD_A(n)$. Also, assume that $\cD_A$ is bounded. Then $X$ is a matrix extreme point of $\cD_A$ if and only if all tuples $(H,W) = (H,W_1,\dots,W_g) \in SM_n (\C)^{g+1}$ that satisfy
        \[
        \ker L_A (X) \subseteq \ker \Lambda_{(I,-A)} (H,W)
        \]
        are of the form $(H,W) = \alpha (I,X)$ for some $\alpha \in \R$.
\end{prop}

There are also kernel-based classifications of Arveson and Euclidean extreme points of free spectrahedra, see \cite[Lemma 2.1 (3)]{evert2019span} and \cite[Corollary 2.3]{evert2018extreme}, respectively. We do not make use of these as the Arveson and Euclidean extreme points of a Cartesian product of free spectrahedra are straightforward to classify, see the upcoming Lemma \ref{lem:ExtremeOfCartesian} \eqref{it:ArvEucCartesian}. However, combining \cite[Corollary 2.3]{evert2018extreme} with Proposition \ref{prop:MatExker} leads to another useful corollary.

\begin{cor}
\label{cor:EucMatKer}
    Let $A \in SM_d (\C)^g$ and let $X \in \cD_A(n)$. Also assume that $\cD_A$ is bounded and that $X$ is an Euclidean extreme point of $\cD_A$. Then $X$ is a matrix extreme point of $\cD_A$ if and only if all tuples $(H,W) = (H,W_1,\dots,W_g) \in SM_n (\C)^{g+1}$ that satisfy
        \[
        \ker L_A (X) \subseteq \ker \Lambda_{(I,-A)} (H,W),
        \]
        also satisfy $H = \alpha I$ for some $\alpha \in \R$.
\end{cor}

\begin{proof}
    Suppose $X$ is an Euclidean extreme point of $\cD_A$ and that $(\alpha I,W) \in SM_n (\C)^{g+1}$ satisfies
    \[
    \ker L_A (X) \subseteq \ker \Lambda_{(I,-A)} (\alpha I,W),
    \]
    for some $\alpha \in \R$. If $\alpha = 0$, then is equivalent to 
    \[
    \ker L_A (X) \subseteq \ker \Lambda_{(A)} (W),
    \]
    in which case \cite[Corollary 2.3 (ii)]{evert2018extreme} shows that $W = 0$ since $X$ is Euclidean extreme. Similarly, if $\alpha \neq 0$, then we have
    \[
    \ker L_A (X) \subseteq \ker \Lambda_{(I,-A)} (I,W/\alpha) = \ker L_A (W/\alpha).
    \]
    In this case, the proof of \cite[Corollary 2.3 (iii)]{evert2018extreme} shows that $W/\alpha = X$. In either case we find that $(\alpha I,W)  = (\alpha I, \alpha X)$. The result then follows from Proposition \ref{prop:MatExker}.
\end{proof}

\subsection{Two equivalent viewpoints on optimization problems of free polytopes}

In this section, we will use extreme points of matrix convex sets to find an alternative formulation of an optimization problem that we will encounter in Section \ref{sec:opti-prob-for-any-number} in our study of compatibility.

Given a free polytope $\cD_A$, a major goal in this paper is to understand the optimization problem
\begin{equation}
\label{eq:SumEigMax}
  \begin{split}
   \sup \quad & \lambda_{\max} \left(\sum X_i \otimes Y_i\right) \\
    \mathrm{subject~to} \quad
    & X \in \cD_A, \quad Y \in \cD_A^\square.
    \end{split}
\end{equation}
In addition, we also consider a level fixed version of the above problem, where the dimensions of the matrices $X_i$ and $Y_i$ are fixed. That is,
\begin{equation}
\label{eq:SumEigMaxFixedn}
   \begin{split}
   \sup \quad & \lambda_{\max} \left(\sum X_i \otimes Y_i\right)  \\
    \mathrm{subject~to} \quad
    & X \in \cD_A(n), \quad Y \in \cD_A^\square (n),
    \end{split}
\end{equation}
where $n \in \mathbb{N}$ is a fixed positive integer. We can now prove a statement that is similar in spirit to Proposition \ref{prop:spectrahedron-mcset-equivalent}.

\begin{thm}
\label{thm:MaxEigIsMinBallInclusionConst}
    Let $\cD_A$ be a bounded free polytope. The objective value of the optimization problem \eqref{eq:SumEigMax} is  equal to the objective value of 
    \begin{equation}
    \label{eq:MinBallContainment}
      \begin{split}
   \mathrm{minimize} \quad & \gamma \\
    \mathrm{subject~to} \quad
    & \delfree\cD_A \subseteq \gamma \cWmin(\cD_A(1)).
    \end{split}
    \end{equation}
    Similarly, if $n \in \mathbb{N}$ is a fixed positive integer, then the objective value of \eqref{eq:SumEigMaxFixedn} is equal to the objective value of 
      \begin{equation}
      \label{eq:MinBallContainmentFixedn}
      \begin{split}
   \mathrm{minimize} \quad & \gamma \\
    \mathrm{subject~to} \quad
    & \delmat\cD_A (j) \subseteq \gamma \cWmin(\cD_A(1)) \qquad \mathrm{for \ all\ }  j \leq n.
    \end{split}
    \end{equation}

Furthermore, the objective value of \eqref{eq:SumEigMaxFixedn} is attained for $X$ and $Y$ that are a direct sum of matrix extreme points of $\cD_A$ and $\cD_A^\square$, respectively.
    On the other hand, in the optimization problem \eqref{eq:SumEigMax}, it is sufficient to consider $X \in \delfree \cD_A$ and $Y \in \delfree\mathcal D_A^\square$.
\end{thm}

Before giving the proof, we expand on a technical point in the statement of Theorem \ref{thm:MaxEigIsMinBallInclusionConst}. As a consequence of \cite[Theorem 2.9]{hartz2021dilation} (see also \cite{webster1999krienmilman,kriel2019intro}) we have that $\cD_A = \matco (\delmat \cD_A),$ and moreover that 
\[
\cD_A (n) = \left(\matco (\cup_{j=1}^n \delmat \cD_A(j))\right)(n).
\] 
However, for $n \geq 2$, it may not be the case that $\cD_A (n) = \matco (\delmat \cD_A(n))$. In fact, $\delmat \cD_A (n)$ can be empty. For example, since a free simplex is the matrix convex hull of its free extreme points at level $1$, a free simplex has no matrix extreme points at levels greater than $1$. Another example is the matrix square in two variables, which has matrix extreme points at levels one and two but does not have matrix extreme points at any level greater than two. 

Intuitively, the issue is that matrix extreme points are required to be irreducible \cite{evert2018extreme}, hence a direct sum of matrix extreme points is not matrix extreme. Allowing direct sums of matrix extreme points in the optimization resolves the issue. 

We most commonly consider optimization problem \eqref{eq:MinBallContainmentFixedn} in the case where $n=2$. Of course, the optimal constant for $j=1$ is $\gamma=1$, so in this case it is sufficient to restrict to matrix extreme points at level $2$.

\begin{proof}[Proof of Theorem \ref{thm:MaxEigIsMinBallInclusionConst}]
First note that bounded free polytopes are closed under complex conjugation since their defining tuple can be taken to be real. It then follows from \cite[Theorem 1.1]{evert2019span} that $\matco(\delfree \cD_A)) = \cD_A$, hence $\delfree \cD_A \subseteq \gamma \cWmin(\cD_A (1))$ if and only if $\cD_A \subseteq \gamma \cWmin(\cD_A (1))$.
Next observe that if one exchanges the roles of $\cD_A$ and its dual in equation \eqref{eq:PolytopeDual} then one can obtain
\begin{equation}
    (\cD_A^\square)^\circ = \cWmin (\cD_A (1)). 
\end{equation}
We then have that $\delfree \cD_A \subseteq \gamma(\cD_A^\square)^\circ$ if and only if for all $X \in \cD_A$ and all $Y \in \cD_A^\square$, one has
\[
\Lambda_Y (X) = \sum Y_i \otimes X_i \preceq \gamma I. 
\]
This is in turn equivalent to stating that the maximum achieved by \eqref{eq:SumEigMax} is at most $\gamma$, from which the equivalence of \eqref{eq:SumEigMax} and \eqref{eq:MinBallContainment} follows. 

We now argue that the objective value of \eqref{eq:SumEigMaxFixedn} is equal to that of \eqref{eq:MinBallContainmentFixedn}. To this end, fix $X \in \cD_A (n)$. Then taking polar duals shows that $X \in \gamma \cW^{\min} (\cD_A (1))$ if and only if $\cD_A^\square \subset \gamma \cD_X$. Moreover, since $X$ is a tuple of $n \times n$ matrices, as a consequence of \cite[Theorem 3.5]{helton_matricial_2013}, we have $\cD_A^\square \subset \gamma \cD_X$ if and only if $\cD_A^\square (n)\subset \gamma \cD_X (n)$. This is in turn equivalent to 
\[
\Lambda_X (Y) = \sum X_i \otimes Y_i \preceq \gamma I \qquad \mathrm{for\ all\ } Y \in \cD_A^\square (n).
\]
We conclude that $\cD_A(n) \subset \gamma \cW^{\min} (\cD_A (1))$ if and only if the objective value of \eqref{eq:SumEigMaxFixedn} at most $\gamma$. The fact that it is sufficient to test on matrix extreme points in \eqref{eq:MinBallContainmentFixedn} is immediate from \cite[Theorem 2.9]{hartz2021dilation} which shows that $\cD_A(n) = \left(\matco(\cup_{j=1}^n \delmat \cD_A(j))\right)(n)$.

The fact that one need only test on free extreme points and matrix extreme points in optimization problems \eqref{eq:SumEigMax} and \eqref{eq:SumEigMaxFixedn}, respectively, follows easily as a consequence of the above arguments and the equivalence of these problems to \eqref{eq:MinBallContainment} and \eqref{eq:MinBallContainmentFixedn}. The fact that the maximum is attained in optimization problem \eqref{eq:SumEigMaxFixedn} is immediate from the compactness of $\cD_A (n)$ and $\cD_A^\square (n)$.
\end{proof}

\begin{cor}
    Let $\cD_A$ be a bounded free polytope and let $\cD_A^\square$ be its dual free polytope. The optimal inclusion constants for $\cD_A$ and $\cD_A^\square$ in optimization problem \eqref{eq:MinBallContainment} are equal. The same is true for optimization problem \eqref{eq:MinBallContainmentFixedn}.
\end{cor}
\begin{proof}
    If $\cD_A$ is bounded, then it is straightforward to show that $(\cD_A^\square)^\square = \cD_A$. The proof is then immediate from Theorem \ref{thm:MaxEigIsMinBallInclusionConst}.
\end{proof}

\subsubsection{Optimality of optimization problem \eqref{eq:SumEigMax}}
    To the authors' knowledge, it is an open question if optimization problem \eqref{eq:SumEigMax} attains its maximum. A challenge there is in general no bound on the level one must check in order to test the containment of a free spectrahedron in a general matrix convex set.

    \begin{prop}
    \label{prop:DAnotinKj}
        Let $\mathcal D_{\square, 3} \subset SM(\C)^3$  be the matrix cube in three variables. Then for each $j \in \mathbb{N}$, there exists a matrix convex set $K_j \subset SM(\C)^3$ such that $\mathcal D_{\square, 3} (j) \subseteq K_j$ but $\mathcal D_{\square, 3}(j+1) \not\subseteq K_j$. 
    \end{prop}
    \begin{proof}
        Fix $j \in \mathbb{N}$ and set $K_j = \matco(\mathcal D_{\square, 3} (j))$. Then it is immediate that $\mathcal D_{\square, 3}(j) \subset K_j$. However, $\mathcal D_{\square, 3}$ has free extreme points at all levels, see \cite[Example 5.9]{evert2024extreme}. Let $X \in \mathcal D_{\square, 3}(j+1)$ be a free extreme point of $\mathcal D_{\square, 3}$. Then from the definition of a free extreme point, $X \notin \matco(\mathcal D_{\square, 3} (j)) = K_j$, hence $\mathcal D_{\square, 3}(j+1) \not\subseteq K_j$. 
    \end{proof}

    The barrier to testing containment in this way is that a $j$-positive map on a subspace of a matrix algebra may fail to have a completely positive extension. We direct the reader to \cite{Paulsen2002} for definitions related to completely positive maps. 

    \begin{prop}
    \label{prop:NoCPExtension}
    There exists a subspace $\mathcal{S} \subset M_6 (\C)$ such that for each $j \in \mathbb{N}$, there exists an $n_j \in \mathbb{N}$ and a $j$-positive map $\tau_j: \mathcal{S} \to M_{n_j} (\C)$ such that $\tau_j$ is not $(j+1)$-positive. As a consequence, if $j \geq 6$, then $\tau_j$ cannot have a $j$-positive extension to all of $M_6(\C)$.
    \end{prop}
    \begin{proof}
        As in Lemma \ref{prop:DAnotinKj}, let $\mathcal D_{\square, 3}$ be the matrix cube in three variables, and let $A \in SM_6 (\R)^3$ be a minimal defining tuple for $\mathcal D_{\square, 3}$, so $\mathcal D_{\square, 3}=\cD_A$. Also fix a $j \in \mathbb{N}$ and let $K_j =  \matco(\cD_A (j))$. As $K_j$ is a closed bounded matrix convex set with $0$ in the interior of $K(1),$ using the Effros-Winkler matricial Hahn-Banach Theorem \cite{Effros1997,HMannals} shows that $K_j$ is a (possibly infinite) intersection of free spectrahedra. That is, there exists tuples $\{B^{(\alpha)}\}_\alpha \subset SM (\C)^g$ such that $K_j = \cap_\alpha \cD_{B^{(\alpha)}}$. 
        
       Since $\cD_A(j) \subset K_j(j)$ for each $j$, we must have $\cD_A(j) \subseteq \cD_{B^{(\alpha)}}(j)$ for each ${B^{(\alpha)}}$ in the collection. However, as $\cD_A(j+1) \not\subseteq K_j(j+1)$, there must exist some $\alpha_0$ such that if one sets $B := B^{{(\alpha_0)}}$, then $\cD_A(j+1) \not\subseteq \cD_{B}(j+1).$ Using \cite[Theorem 3.5]{helton_matricial_2013} now shows that the map $\tau_j: M_6(\C) \to M_{n_j} (\C)$ defined by
       \[
        \tau_j (I_6) = I_{n_j} \qquad \mathrm{and} \qquad \tau_{j} (A_\ell) = B_\ell \qquad \mathrm{for\ } \ell=1,2,3
       \]
       is $j$-positive but not $(j+1)$-positive. Here $n_j$ denotes the size of the matrix tuple $B \in SM_{n_j} (\R)^3$. 

       To complete the proof, note that if $j \geq 6$ and $\tau_j$ had a $j$-positive extension to all of $M_6(\C)$, then as a consequence of Choi's Theorem, e.g., see \cite[Theorem 3.14]{Paulsen2002}, the map $\tau_j$ would be completely positive, hence $\tau_j$ would be $(j+1)$-positive. 
    \end{proof}

    \begin{remark} It is not difficult to show that since $\mathcal D_{\square, 3}$ is bounded and closed under complex conjugation, each $\cD_{B^{(\alpha)}}$ can be assumed to be bounded and closed under complex conjugation. In particular, each $\cD_{B^{(\alpha)}}$ being closed under complex conjugation implies that each $B^{(\alpha)}$ can be taken to be a real matrix tuple. 

    Additionally, for context we note that it is well known that a positive map need not be completely positive, see \cite[Example 2.13]{Paulsen2002}. However, If $\tau$ is $n$-positive and maps into $M_n(\C)$ or if $\tau$ maps from $M_n(\C)$ and is defined on all of $M_n(\C)$, then $\tau$ necessarily is completely positive, see \cite[Theorem 3.7 and Theorem 3.14]{Paulsen2002}. The above proposition is intended to highlight that if one maps from a strict subspace of a finite-dimensional space and does not restrict the dimension of the target space, then $j$-positivity cannot be used to guarantee complete positivity. 
    \end{remark}

    While Propositions \ref{prop:DAnotinKj} and \ref{prop:NoCPExtension} are negative results toward optimality of optimization problem \eqref{eq:SumEigMax}, they do not rule out the possibility of optimality. They only highlight that for a given free spectrahedron $\cD_A$, there is no fixed $j$ such that for all matrix convex sets $K$, if $\cD_A(j) \subset K(j)$ then $\cD_A \subset K$. However it could be the case that for each matrix convex set $K$, there does exist some $j_K$ depending on $K$ such that if $\cD_A(j_K) \subset K(j_K)$ then $\cD_A \subset K$, which would imply that optimization problem \eqref{eq:SumEigMax} achieves its maximum. 

    Moreover, we are not considering containment of $\cD_A$ into a general matrix convex set when we study optimization problem \eqref{eq:MinBallContainment}. We only need consider $\cD_A \subset \gamma\cWmin (\cD_A (1))$. In this special case, we have numerical evidence that suggests that the containment can be checked by checking a sufficiently high level whose size depends only on the number of variables in $A$. Furthermore, we observe that, roughly speaking, the contribution of a free simplex to the level bound is very limited.

    \begin{question}
        Let $A \in SM(\R)^g$ and assume that $\cD_A$ is a bounded free spectrahedron. Additionally assume that 
        \[
        \cD_A (2^{g-1}) \subset \gamma\cWmin (\cD_A (1)).
        \]
        Does it hold that $\cD_A \subset \gamma\cWmin (\cD_A (1))$?

        Additionally, suppose 
        \[
        \cD_A = \cD_{A^1} \times \cD_{A^2} \times \dots \times \cD_{A^M}
        \]
        where each $\cD_{A^j}$ is a free simplex and assume that 
        \[
        \cD_A (2^{M-1}) \subset \gamma\cWmin (\cD_A (1)).
        \]
        Does it hold that $\cD_A \subset \gamma\cWmin (\cD_A (1))$?
    \end{question}  

We mention that the level bound $2^{g-1}$ is known to be optimal in the case $\cD_A$ is a matrix cube in $g$-variables, see \cite[Theorem 6.6]{passer2018minimal}.

\section{Extreme points of Cartesian products of free spectrahedra}\
\label{sec:CartesianExtreme}
We now explore the extreme points of Cartesian products of free spectrahedra. First note that the Cartesian product of free spectrahedra is indeed a free spectrahedron.

\begin{lem}
\label{lem:CartesianIsFreeSpec}
Let $A \in SM_{d_1}(\C)^g$ and let $B \in SM_{d_2}(\C)^{h}$. Then $\cD_A \times \cD_B$ is the free spectrahedron
\[
\cD_A \times \cD_B= \cD_{(\tilde{A},\tilde{B})}  \subset SM(\C)^{g+h}
\]
where $\tilde{A} = A \oplus 0\in SM_{d_1+d_2}(\C)^g $ and $\tilde{B} = 0 \oplus B \in SM_{d_1+d_2}(\C)^h$. Furthermore, $\cD_{A} \times \cD_{B}$ is bounded if and only if both $\cD_A$ and $\cD_B$ are bounded.
\end{lem}
\begin{proof}
    Straightforward.
\end{proof}

As mentioned previously, the Cartesian product of free polytopes is related to the direct sum of free polytopes via the free polyhedral dual.

\begin{prop}
\label{prop:CartesianPolyDualIsDirect}
    Let $\cD_A \times \cD_B$ be a Cartesian product of bounded free polyhedra. Then the dual free polytope $(\cD_A \times \cD_B)^\square$ is the direct sum of the dual free polytopes $\cD_A^\square$ and $\cD_B^\square$. As an immediate consequence, the dual free polytope of a direct sum of free polyhedra is the Cartesian product of the corresponding dual free polytopes. 
\end{prop}
\begin{proof}
    For a general bounded free polytope $K$, it is straightforward to show that the dual free polytope $K^\square$ has defining pencil $\mathcal{E} = \oplus_{E \in \delfree \cD_A (1)} E$. That is, if $\mathcal{E}$ is a direct sum of the extreme points of $K (1)$, then $K^\square = \cD_\mathcal{E}$. 

    On the other hand, the extreme points at level $1$ of a Cartesian product of free spectrahedra are given by pairs of extreme points at level one of the respective free spectrahedra. That is $(X,Y) \in \delfree (\cD_A \times \cD_B) (1)$ if and only if $X \in \delfree \cD_A (1)$ and $Y \in \delfree \cD_B(1)$. Then $\mathcal{E}_A$ and $\mathcal{E}_B$ denote diagonal tuples equal to direct sums given by direct sums of the free extreme points of $\cD_A (1)$ and $\cD_B(1)$, respectively. It is then straightforward to check that 
    \[
    \mathcal{E}_{A \otimes B}:=(\mathcal{E}_A \otimes I, I \otimes \mathcal{E}_B)
    \]
    is a tuple that is given by a direct sum of the free extreme points of $(\cD_A \times \cD_B) (1)$. We conclude that
\[
(\cD_A \times \cD_B)^\square = \cD_{\mathcal{E}_{A \times B}} 
\]
is a direct sum of free polytopes. 
\end{proof}
\begin{remark}
    As we will soon see, the free extreme points of a Cartesian product of free spectrahedra are essentially Cartesian products of free extreme points, making them easy to work with. However, the situation for direct sums of free spectrahedra is less straightforward. At level one, a tuple $(X,Y)$ is a (free) extreme point of the direct sum of $\cD_A$ and $\cD_B$ if and only if either $X = 0$ and $Y$ is extreme in $\cD_B(1)$ or $Y=0$ and $X$ is extreme in $\cD_A(1)$. This does not extend to higher levels of direct sums of free spectrahedra. 
\end{remark}

We next give a kernel containment based classification of extreme points of a Cartesian product of free spectrahedra.

\begin{lem}
    \label{lem:ExtremeOfCartesian}
    Let $\cD_A$ and $\cD_B$ be bounded free spectrahedra and let $\cD_A \times \cD_B$ be the Cartesian product of $\cD_A$ and $\cD_B$. Then we have the following.
    \begin{enumerate}
        \item \label{it:ArvEucCartesian} $(X,Y) \in \cD_A \times \cD_B$ is an Arveson (Euclidean) extreme point of $\cD_A \times \cD_B$ if and only if $X$ and $Y$ are Arveson (Euclidean) extreme points of $\cD_A$ and $\cD_B$, respectively. 

        \item \label{it:MatExCartesian} $(X,Y)\in \cD_A \times \cD_B$ is a matrix extreme point of $\cD_A \times \cD_B$ if and only if all solutions $(H,W,Z)$ to the kernel containment
        \begin{equation}
\label{eq:CartesianMatExEquation}
\ker L_{A} (X) \subset \ker \Lambda_{(I,-A)} (H,W) \qquad \mathrm{and} \qquad \ker L_{B} (Y) \subset \ker \Lambda_{(I,-B)} (H,Z)
\end{equation}
have the form $(H,W,Z) = \alpha(I,X,Y)$ for some $\alpha \in \R$.
    \end{enumerate}
\end{lem}
\begin{proof}
 The proof of Item \eqref{it:ArvEucCartesian}  quickly follows using Lemma \ref{lem:CartesianIsFreeSpec} together with \cite[Theorem 2.6]{epperly2024matex} (see also \cite[Theorem 1.1]{evert2018extreme} and \cite[Proposition 6.5 (c)]{kriel2019intro}). 

Since we are most concerned with matrix extreme points, details are given for Item \eqref{it:MatExCartesian}. Using \cite[Theorem 2.6 (c)]{epperly2024matex}, we have that $(X,Y) \in \cD_A \times \cD_B = \cD_{(\tilde{A}, \tilde{B})}$ is a matrix extreme point if and only if all solutions to the kernel containment
\begin{equation}
    \label{eq:MatExKerEquation}
\ker L_{(\tilde{A}, \tilde{B})} (X,Y) \subset \ker \Lambda_{(I,-\tilde{A}, -\tilde{B})} (H,W,Z)
\end{equation}
have the form $(H,W,Z) = \alpha(I,X,Y)$. Set $\hat{A} = (A,0)$ and $\hat{B} = (0,B)$ so that 
$
(\tilde{A}, \tilde{B}) = \hat{A} \oplus \hat{B}$.
Then, by using canonical shuffles, we can construct a unitary $U$ such that for any $(H,W,Z)$ we have
\[
U^*\left(\Lambda_{(I, -\hat{A} \oplus \hat{B})} (H,W,Z)\right)U =\Lambda_{(I,- \hat{A})} (H,W,Z) \oplus \Lambda_{(I, -\hat{B})} (H,W,Z).
\] 
However, since $\hat{A} = (A,0)$ and $\hat{B} = (0,B)$ we have 
\[
\Lambda_{(I, -\hat{A})} (H,W,Z) =\Lambda_{(I,-A)} (H,W) \qquad \mathrm{and} \qquad \Lambda_{(I, -\hat{B})} (H,W,Z)=\Lambda_{(I,-B)} (H,Z).
\] 
We conclude that $(H,W,Z)$ is a solution to equation \eqref{eq:MatExKerEquation} if and only if $(H,W,Z)$ is a solution to equation \eqref{eq:CartesianMatExEquation}, from which the result follows.
\end{proof}

\begin{remark}
    Note that if $(X,Y) \in \cD_A \times \cD_B$ is free extreme, then it is also Arveson extreme, so the preceding lemma shows that $X$ and $Y$ are Arveson extreme in $\cD_A$ and $\cD_B$, respectively. However, it may not be the case that $X$ and $Y$ are free extreme points of $\cD_A$ and $\cD_B$. This is because the tuple $(X,Y)$ may be irreducible while $X$ and $Y$ themselves are reducible. 

    As an example the tuple
    \[
    (X,Y) = \left(\begin{pmatrix} 1 & 0 \\ 0 & -1 \end{pmatrix},
    \begin{pmatrix} 0 & 1 \\ 1 & 0 \end{pmatrix}\right) 
    \]
    is a free extreme point of the matrix square. However, neither $X$ nor $Y$ are free extreme points of the matrix interval, since they are reducible. 
\end{remark}

Lemma \ref{lem:ExtremeOfCartesian} shows that, up to reducibility, the behavior of free and Euclidean extreme points under Cartesian products mirrors the classical setting. However, as suggested by the lemma, the situation for matrix extreme points is more complicated. Indeed, a matrix extreme point in a Cartesian product need not be a Cartesian product of matrix extreme points.

\begin{prop}
\label{prop:MatrixCartesianEuclidean}
Let $\mathcal{D}_A$ and $\mathcal{D}_B$ be bounded free spectrahedra. Also let $X$ be an Euclidean extreme point of $\mathcal{D}_A$ and let $Y$ be a matrix extreme point of $\mathcal{D}_B$. Then $(X,Y)$ is a matrix extreme point of $\mathcal{D}_A \times \mathcal{D}_B$. 
\end{prop}

\begin{proof}
Let $X$ be an Euclidean extreme point of $\cD_A$ and let $Y$ be a matrix extreme point of $\cD_B$. Then Lemma \ref{lem:ExtremeOfCartesian} \eqref{it:ArvEucCartesian} shows that $(X,Y)$ is an Euclidean extreme point of $\mathcal{D}_A \times \mathcal{D}_B$. Therefore, using Lemma \ref{lem:ExtremeOfCartesian} \eqref{it:MatExCartesian}, it is sufficient to show that all solutions 
to equation \eqref{eq:CartesianMatExEquation} are of the form $(H,W,Z) = \alpha (I,X,Y)$. 

Since $Y$ is a matrix extreme point in $\cD_B$, the only solutions to
\[
L_B (Y) \subset \ker \Lambda_{(I,-B)} (H,Z)
\]
are of the form $(H,Z) = \alpha (I,Y)$. Therefore, if $(H,W,Z)$ satisfies equation \eqref{eq:CartesianMatExEquation}, then we must have $(H,W,Z) = (\alpha I, W, \alpha Y)$ for some $\alpha \in \R.$ Next consider solutions to
\begin{equation}
\label{eq:EucPartOfMatEuc}
\ker L_A (X) \subset \ker \Lambda_{(I,-A)} (\alpha I,W).
\end{equation}
Using the same argument as in the proof of Corollary \ref{cor:EucMatKer} shows that $W = \alpha X$. We conclude that $(H,W,Z)$ must be equal to $\alpha(I,X,Y)$ for some $\alpha$ from which it follows that $(X,Y)$ is matrix extreme in $\cD_A \times \cD_B$. 
\end{proof}

\begin{remark}
    Considering the additional restrictions placed on free extreme points compared to matrix extreme points, it is natural to expect that matrix extreme points that are not free extreme are plentiful; however until recently no such examples were known for real free spectrahedra. Such points are a focus of \cite{epperly2024matex} in which a small number of exact examples of matrix extreme points that are not free extreme are constructed. \cite{epperly2024matex} also provides numerical evidence that matrix extreme points that are not free extreme are plentiful. In the article, it is illustrated that there can be computational challenges in constructing such tuples with exact arithmetic. 

    For the special case of Cartesian products of free spectrahedra, Proposition \ref{prop:MatrixCartesianEuclidean} provides new insight to this problem and essentially reduces the task to finding Euclidean extreme points that are not Arveson extreme.

\end{remark}

\begin{ex}
\label{ex:SimplexXSquareMatNotArv}
    Let $\cD_A$ be the free simplex defined by
    \[
    A = \big(\mathrm{diag}(1,0,-1),\mathrm{diag}(0,1,-1) \big)
    \]
    and let $\cD_B$ be the free matrix square. Set 
    \[
    X = \left(\begin{pmatrix} 1 & 0 \\ 0 & 0 \end{pmatrix},
    \begin{pmatrix}
        \frac{1}{2} & \sqrt{\frac{5}{6}} \\ 
        \sqrt{\frac{5}{6}} & -\frac{2}{3}
    \end{pmatrix}\right) \qquad \mathrm{and} \qquad     Y = \left(\begin{pmatrix} 1 & 0 \\ 0 & -1 \end{pmatrix},
    \begin{pmatrix} 0 & 1 \\ 1 & 0 \end{pmatrix}\right).
    \]
Then $(X,Y)$ is a matrix extreme point of $\cD_A \times \cD_B$ but not an Arveson extreme point of $\cD_A \times \cD_B$, which follows from Proposition \ref{prop:MatrixCartesianEuclidean} with Lemma \ref{lem:ExtremeOfCartesian} \eqref{it:ArvEucCartesian}. In particular, by checking the kernel containment conditions of \cite[Theorem 2.6]{epperly2024matex}, one can show that $X$ is an Euclidean extreme point of $\cD_A$ but not an Arveson extreme point of $\cD_A$. Also, $Y$ is a free extreme point, hence a matrix extreme point, of $\cD_B$.

Furthermore, as discussed in Remark \ref{rem:extreme-points-free-simplex}, all matrix extreme points of $\cD_A$ are elements of $\cD_A(1)$, so $X$ is not matrix extreme in $\cD_A$. It follows that $(X,Y)$ is a matrix extreme point of $\cD_A$ that is not a Cartesian product of matrix extreme points. Note that $X$ and $Y$ are both irreducible in this example, so this behavior is not caused by a simple failure of irreducibility. 
\end{ex}

\subsection{Real extreme points of Cartesian products with a free simplex}

We now focus on the special case of the Cartesian product where one of the free spectrahedra in the product is a free simplex. Throughout this subsection we let $\cD_S$ denote a bounded free simplex in $g$-variables.

 \begin{thm}
\label{thm:SimplexXIntervalRealMatEx}
Let $\mathcal{D}_S$ be a free simplex, let $\mathcal{D}_B$ be any free spectrahedron, and let $(X,Y) \in \left(\mathcal{D}_S \times \mathcal{D}_B\right) (2)$. Also assume that $(X,Y)$ is real valued.

\begin{enumerate}
\item
\label{it:SimpleXSpecRealMatEx}
If $(X,Y)$ is a matrix extreme point of $\cD_S \times \cD_B$ and $Y$ is not matrix extreme in $\cD_B$, then $X$ must be an Arveson extreme point of $\cD_S$. Equivalently, if $X$ is not Arveson extreme in $\cD_S$, then $Y$ must be matrix extreme in $\cD_B$. 

\item
\label{it:SimpleXSimplexRealMatEx}
If $\cD_B$ is also a free simplex, then $(X,Y)$ is a matrix extreme point of $\cD_S \times \cD_B(2)$ if and only if $(X,Y)$ is free extreme in $\cD_S \times \cD_B(2)$. That is, if $(X,Y)$ is a matrix extreme point of $(\cD_S \times \cD_B )(2)$, then $(X,Y)$ is irreducible and $X$ and $Y$ are Arveson extreme in $\cD_S$ and $\cD_B$, respectively. 
\end{enumerate}
\end{thm}

We are particularly interested in the case that $\cD_B$ is the matrix interval. Note that the matrix interval is itself a free simplex, and therefore, the real matrix extreme points at level two of the Cartesian product of a free simplex with the matrix interval are completely characterized by Theorem \ref{thm:SimplexXIntervalRealMatEx} \eqref{it:SimpleXSimplexRealMatEx}.

We will need several lemmas before proving Theorem \ref{thm:SimplexXIntervalRealMatEx}. We begin by classifying the extreme rays of homogeneous free simplices. 

\begin{prop}
\label{prop:ExtremeRaysVsMatEx}
 Let $\cD_A$ be a bounded free spectrahedron in $g$ variables and let $\mathcal{H}_{(I,-A)}$ be the corresponding homogeneous free spectrahedron. Then $X=(X_0,X_1,\dots,X_g) \in \mathcal{H}_{(I,-A)} (n)$ is a nonzero extreme ray of $\mathcal{H}_{(I,-A)}$ if and only if there is a unitary $U$ such that
 \[
    X = U^* (W \oplus 0) U
 \]
 where $W=(W_0,W_1\dots,W_g) \in \mathcal{H}_{(I,-A)} (m)$ for 
 some $m \leq n$ with $W_0 \succ 0$ and where the tuple
 \[
W_0^{-1/2}(W_1,\dots,W_g)W_0^{-1/2}
 \]
 is a matrix extreme point of $\cD_A(m)$. 
\end{prop}
\begin{proof}
    This result essentially follows from \cite[Theorem 6.5]{kriel2019intro} together with \cite[Lemma 2.2]{evert2021quadrilaterals}. More precisely, \cite[Lemma 2.2]{evert2021quadrilaterals} shows that $X_0 \succeq 0$ and gives the kernel containment $\ker X_0 \subseteq \ker X_i$ for each $i=1,\dots,g$, which allows us to write $X = U^* (W \oplus 0) U$ where $W_0 \succ 0$. The remainder of the proof follows from combining \cite[Theorem 6.5 (c),(g)]{kriel2019intro} which together show that 
    \[
    W_0^{-1/2}(W_0,W_1,\dots,W_g)W_0^{-1/2}
    \]
    is an extreme ray in $\mathcal{H}_{(I,-A)}$ and that if $(I,W')$ is an extreme ray in $\mathcal{H}_{(I,-A)}$, then $W' \in \cD_A$ is a matrix extreme point of $\cD_A$.
\end{proof}

\begin{cor}
\label{cor:SimplexExtremeRays}
 Let $\cD_S$ be a free simplex in $g$ variables and let $\mathcal{H}_{(I,-S)}$ be the corresponding homogeneous free simplex. Then $X=(X_0,X_1,\dots,X_g) \in \mathcal{H}_{(I,-S)} (n)$ is a nonzero extreme ray of $\mathcal{H}_{(I,-S)}$ if and only if there is a unitary $U$ such that
 \[
    X = U^* (W \oplus 0) U,
 \]
 where $W=(W_0,W_1\dots,W_g) \in \mathcal{H}_{(I,-S)} (1)$ with $W_0 > 0$ and where the tuple
 \[
(W_1,\dots,W_g)/W_0,
 \]
 is an ordinary extreme point of $\cD_S(1)$. 
\end{cor}

\begin{proof}
As discussed in Remark \ref{rem:extreme-points-free-simplex}, the set of matrix extreme points of a free simplex coincides with the set of ordinary extreme points of level one of the free simplex. The result then follows from Proposition \ref{prop:ExtremeRaysVsMatEx}.
\end{proof}

\begin{lem} \label{lem:exists-W}
    Let $\cD_S$ be a bounded free simplex and let $X \in \cD_S(2)$ be a real Euclidean extreme point of $\cD_S(2)$ that is not Arveson extreme in $\cD_S$. Given any $H \in SM_2(\R)$ there exists a $W \in SM_2 (\R)^g$ such that 
    \[
    \ker L_S (X) \subseteq \ker \Lambda_{(I,-S)} (H,W).
    \]
\end{lem}

\begin{proof}
    From Remark \ref{rem:extreme-points-free-simplex}, we see that if $X \in \cD_S(2)$ is an Euclidean extreme point of $\cD_S(2)$ that is not Arveson extreme, then $X$ is irreducible and is not matrix extreme in $\cD_S$. Using \cite[Theorem 6.5 (c)]{kriel2019intro}, since $X$ is not matrix extreme of $\cD_S$, the tuple $(I,X)$ can be written as a classical convex combination of extreme rays of $\mathcal{H}_{(I,-S)}.$ That is, there exists a finite collection of extreme rays $\{(W_0^\ell,\dots,W_g^\ell)\}_{\ell=1}^k \subset SM_{2}(\R)^{g+1}$ and positive constants $\alpha_\ell > 0$ such that
\begin{equation}
\label{eq:ExtremeRayConvexCombo}
(I,X) = \sum_{\ell=1}^k \alpha_\ell (W_0^\ell,W_1^\ell,\dots,W_g^\ell) \qquad \mathrm{where} \qquad \sum_{\ell=1}^k \alpha_\ell = 1,
\end{equation}
and such that
    \[
    \ker L_S (X) \subseteq \ker \Lambda_{(I,-S)} (W_0^\ell,W_1^\ell,\dots,W_g^\ell) \qquad \mathrm{for \ each} \qquad \ell = 1,\dots,k.
    \]

Of course, linear combinations of solutions to the above kernel containment also satisfy the kernel containment, so it suffices to prove that $\mathrm{span} \{W_0^\ell\} = SM_2(\R)$. Using Corollary \ref{cor:SimplexExtremeRays}, the extreme rays of $\cD_S$ all have the form $W=(W_0,W_1,\dots,W_g)$ where the $W_i$ all have rank at most one and where $\mathrm{ran} (W_i) \subseteq \mathrm{ran} (W_0)$ for each $i$.  Combining this with equation \eqref{eq:ExtremeRayConvexCombo} implies that $\dim \mathrm{span} (\{W_0^\ell\}_{\ell=1}^k) \geq 2$, since $I_2$ is in the span of the rank one matrices ${W_0^\ell}$.

We now argue that if  $\dim \mathrm{span} \{W_0^\ell\} = 2$, then  $(I,X)$ is reducible. To this end assume WLOG that $\{W_0^1,W_0^2\}$ forms a basis for $\mathrm{span} (\{W_0^\ell\}_{\ell=1}^k).$ That is, all remaining $W_0^\ell$ are linear combinations of $W_0^1$ and $W_0^2$. On the other hand, each $W_0^\ell$ is symmetric and has rank $1$, so each $W_0^\ell$ is a constant multiple of either $W_0^1$ or $W_0^2$. 

Now, from $\sum \alpha_\ell W_0^\ell = I_2$, we obtain that there are some constants $\beta_1$ and $\beta_2$ such that $\beta_1 W_0^1 +\beta_2 W_0^2 =I_2$. Since $W_0^1$ and $W_0^2$ are both rank one, this implies that their ranges are orthogonal. That is, up to one's choice of orthonormal basis $\{e_1,e_2\}\subset \R^2$, we have $W_0^1 = \gamma_1 e_1 e_1^T$ and $W_0^2 = \gamma_2 e_2 e_2^T$. Now, as a consequence of the fact that each $W_0^\ell$ is a constant multiple of one either $W_0^1$ or $W_0^2$, we have that each $W_0^\ell = \gamma_\ell e_{k_\ell} e_{k_\ell}^T$, where $k_\ell = 1$ or $2$. Using this with the range containment $\mathrm{ran} (W_i^\ell) \subseteq \mathrm{ran} (W_0^\ell)$ tells us that the spans of $e_1$ and $e_2$ are reducing subspaces for each $(W_0^\ell,W_1^\ell,\dots,W_g^\ell)$. Finally, since the spans of $e_1$ and $e_2$ are reducing for each of these tuples, they are also reducing for a linear combination of them, hence for $(I,X)$. That is, $(I,X)$ is reducible if $\dim \mathrm{span} \{W_0^\ell\} = 2$,

By assumption, $(I,X)$ is irreducible, so we now obtain that $\dim \mathrm{span} (\{W_0^\ell\}_{\ell=1}^k) > 2.$ Since $SM_2 (\R)$ is three-dimensional and since $\mathrm{span} (\{W_0^\ell\}_{\ell=1}^k) \subseteq SM_2(\R)$, we conclude $\mathrm{span} (\{W_0^\ell\}_{\ell=1}^k)= SM_2(\R)$ from which the conclusion follows. 
\end{proof}

\subsubsection{Proof of Theorem \ref{thm:SimplexXIntervalRealMatEx}  }
\begin{proof}
We first prove item \eqref{it:SimpleXSpecRealMatEx}.  To this end assume that $(X,Y) \in \cD_S \times \cD_B$ is a tuple of real matrices but that $Y$ is not matrix extreme in $\cD_B$ and that $X$ is not Arveson extreme in $\cD_S$. It then follows from Proposition \ref{prop:MatExker} that there exists some tuple $(H,Z)$ such that 
\[
\ker L_B(Y) \subset \ker \Lambda_{(I,-B)} (H,Z)
\]
and such that $(H,Z) \neq \alpha (I,Y)$ for any $\alpha \in \R$. 

Now, if $X$ is not Euclidean extreme in $\cD_S$, then $(X,Y)$ is not Euclidean extreme in $\cD_S \times \cD_B$ by Lemma \ref{lem:ExtremeOfCartesian}, hence is not matrix extreme. On the other hand, since $X$ is not Arveson extreme by assumption, if $X$ is Euclidean extreme, then using Lemma \ref{lem:exists-W} shows that there exists some $W \in SM_2(\R)^g$ such that
\[
\ker L_S(X) \subset \ker \Lambda_{(I,-S)} (H,W).
\]
Furthermore, we have that $(H,W,Z) \neq \alpha(I,X,Y)$ for any $\alpha \in \R$ since $(H,Z) \neq \alpha (I,Y)$ for any $\alpha \in \R$. It follows from Lemma \ref{lem:ExtremeOfCartesian} \eqref{it:MatExCartesian} that $(X,Y)$ is not matrix extreme in $\cD_S \times \cD_B$.

Item \ref{thm:SimplexXIntervalRealMatEx} \eqref{it:SimpleXSimplexRealMatEx} follows quickly from item \eqref{it:SimpleXSpecRealMatEx}. In particular, if $\cD_S$ and $\cD_B$ are both free simplices and $(X,Y) \in (\cD_S \times \cD_B)(2)$, then, as discussed in Remark \ref{rem:extreme-points-free-simplex}, $X$ is not matrix extreme in $\cD_S$ and $Y$ is not matrix extreme in $\cD_B$. Therefore, if $(X,Y)$ is matrix extreme in $\cD_S \times \cD_B$, then applying Item  \ref{thm:SimplexXIntervalRealMatEx} \eqref{it:SimpleXSpecRealMatEx} shows that $X$ must be Arveson extreme in $\cD_S$ and that $Y$ must be Arveson extreme in $\cD_B$. Thus, by Lemma \ref{lem:ExtremeOfCartesian} $(X,Y)$ is Arveson extreme in $\cD_S \times \cD_B$ as well. Furthermore, the tuple $(X,Y)$ must be irreducible since matrix extreme points are irreducible. 

The reverse direction of Item \ref{thm:SimplexXIntervalRealMatEx} \eqref{it:SimpleXSimplexRealMatEx} is straightforward from the fact that free extreme points are always matrix extreme.
\end{proof}

\subsection{Complex extreme points of a free simplex Cartesian product the matrix interval}

We now examine the extreme points of the Cartesian product of a free simplex with the matrix interval when working over the complexes. As it turns out, the situation here is quite different from the real setting. 

\begin{thm}
    \label{thm:SimpleXIntervalComplexMatEx} Let $\cD_S$ be a free simplex and let $\cD_C$ be the matrix interval. Additionally, let $X \in \cD_S(2)$ be a real Euclidean extreme point of $\cD_S$ that is not Arveson extreme and let $Y \in \cD_C(2)$ be an Arveson extreme point of $\cD_C$ with nonzero complex part. Then we have the following: 
    
    \begin{enumerate}
    
        \item \label{it:RealXComplexIsMatEx} $(X,Y)$ is a matrix extreme point of $\cD_S \times \cD_C$. 

    \item \label{it:RealEmbeddingNotMatEx} Let $Y_\R$ and $Y_\C$ denote the real and complex parts of $Y$. The tuple
    \[
    \left(\begin{pmatrix} X & 0 \\
    0 & X
    \end{pmatrix},
    \begin{pmatrix} Y_\R & Y_\C \\
    Y_\C^T & Y_\R
    \end{pmatrix}\right)
    \]
    is never a matrix extreme point of the real free spectrahedron $\cD_S^{\R} \times \cD_C^{\R}$.

    \end{enumerate}
\end{thm}

A few remarks are in order. First, as a consequence of Theorem \ref{thm:SimpleXIntervalComplexMatEx} \eqref{it:RealXComplexIsMatEx}, when working over the complexes there exist matrix extreme points $(X,Y) \in (\cD_S \times \cD_C) (2)$ such that $X$ is not an Arveson extreme point of $\cD_S$. For example, one can take $\cD_S$ and $X$ as in Example \ref{ex:SimplexXSquareMatNotArv} and let $Y$ be any self-adjoint $2 \times 2$ matrix with nonzero imaginary part such that $Y^2=I.$ 

One might expect that the real embedding of such a tuple gives an element of $(\cD_S \times \cD_C) (4)$ (viewed as a real free spectrahedron) that is matrix extreme but not Arveson extreme. However, Theorem \ref{thm:SimpleXIntervalComplexMatEx} \eqref{it:RealEmbeddingNotMatEx} shows that the real embedding of such a matrix extreme point is never matrix extreme. As a consequence, while such tuples could in principle achieve the maximum value of problem \eqref{eq:MinBallContainmentFixedn} when $n=2$, they need not be considered when working at level $n=4$ even when working over the reals. Currently, it is not known if there exist real matrix extreme points of $\cD_S \times \cD_C$ that are not Arveson extreme.

 Finally, Theorem \ref{thm:SimpleXIntervalComplexMatEx} illustrates a previously unknown difference between free extreme points and matrix extreme points of free spectrahedra. In particular, \cite[Theorem 2.4]{evert2024Compact} shows that if $A$ is a tuple of real matrices and $X \in \cD_A (n) \subset SM(\C)^g$ is a free extreme point of the complex free spectrahedron $\cD_A$ and if $\hat{X}$ is the standard embedding of $X$ into the reals, then $\hat{X}$ is a free extreme point of the real free spectrahedron $\cD_A$ so long as $\hat{X}$ is irreducible over $\R$. As a consequence of Theorem \ref{thm:SimpleXIntervalComplexMatEx} \eqref{it:RealEmbeddingNotMatEx}, this property does not extend to matrix extreme points.

We again collect several lemmas before proving the theorem. The first two lemmas together show that if $X \in \cD_S(2)$ is an Euclidean extreme point of the free simplex and $(H,W)$ is a solution to the matrix extreme point kernel containment for $X$, then $H$ is real.

\begin{lem}
\label{lem:SimplexEuc2Kernels}
    Let $\cD_S$ be a bounded free simplex in $g$ variables and for each $j=1,\dots,g$ write
    \[
    S_j = \mathrm{diag}(s_{1,j},\dots, s_{g+1,j}).
    \]
    If $X \in \cD_S(2)$ is an Euclidean extreme point of $\cD_S$, then either 
    \[
    \dim \ker \bigg(I_2 -\sum_{j=1}^g s_{\ell,j} X_j\bigg) \geq 1 \qquad \mathrm{for \ all} \qquad \ell=1,\dots,g+1,
    \]
    or there exists an $\ell_0 \in \{1,\dots,g+1\}$ such that
    \[
    I_2-\sum_{j=1}^g s_{\ell,j} X_j = 0 \qquad \mathrm{for \ all} \qquad \ell \in \{1,\dots,g+1\}\backslash\{\ell_0\}.
    \]
    In the second case, there exists an Euclidean extreme point $x = (x_1,\dots,x_g)$ of $\cD_A(1)$ such that
    \[
    X_j = x_j I_2 \qquad \mathrm{for \ all} \qquad j=1,\dots,g. 
    \]
\end{lem}

\begin{proof}
    Let $X \in \cD_S$ and note that this equivalent to 
    \[
    I_2 -\sum_{j=1}^g s_{\ell,j} X_j \succeq 0 \qquad \mathrm{for \ all} \qquad \ell=1,\dots,g+1.
    \]
    We now without loss of generality assume that
    \[
    \dim \ker \bigg(I_2 -\sum_{j=1}^g s_{1,j} X_j\bigg) = 0 \qquad \mathrm{and} \qquad     I_2-\sum_{j=1}^g s_{2,j} X_j \neq 0.
    \]
   Additionally, since $\cD_S$ is bounded, up to an invertible linear transformation, we may assume that 
    \[
    s_{\ell,j} = \delta_{\ell,j} \qquad \mathrm{for} \qquad \ell,j \in {1,\dots,g}
    \]
    and that $s_{g+1,j} < 0$ for all $j = 1,\dots,g$.

    With this setup, since $I-X_2 \succeq 0$ and has rank at least $1$ and since $I-X_1 \succ 0$, we can choose a rank one symmetric matrix $\beta$ such that $I-X_2 \pm s_{g+1,1} \beta \succeq 0$ and such that $I-X_1 \pm s_{g+1,2} \beta \succeq 0$. It is then straightforward to check that
    \[
    (X_1,X_2,X_3,\dots,X_g) \pm (-s_{g+1,2} \beta,s_{g+1,1} \beta,0,\dots,0) \in \cD_A,
    \]
    hence $X$ is not Euclidean extreme in $\cD_A$. 

    In the case that $g$ of the $g+1$ defining equations for the simplex evaluate to zero on $X$, it is straightforward to verify that $X$ has the form $X = (x_1 I,\dots, x_g I)$, where $x \in \cD_S (1)$ is the Euclidean extreme point of $\cD_S (1)$ that is determined by the same $g$ defining equations of the simplex vanishing on $x$. 
\end{proof}

\begin{lem}
\label{lem:SimplexHomogeneousSolutions}
    Let $\cD_S$ be a bounded free simplex and let $X = (X_1,\dots,X_g) \in \cD_S(2)$ be a real Euclidean extreme point of $\cD_S (2)$. Assume there is a $j \in \{1,\dots,g\}$ such that $X_j$ is not a constant multiple of the identity. If $(H,W) \in SM_2 (\C)^{g+1}$ is a solution to the kernel containment
    \[
    \ker L_A(S) \subset \ker \Lambda_{(I,-S)} (H,W),
    \]
    then $H$ is real.
\end{lem}
\begin{proof}
   Write $S_j = \mathrm{diag}(s_{1,j},\dots, s_{g+1,j})$ for $j=1,\dots,g$. By assumption there is an $X_j$ that is not a constant multiple of the identity, so Lemma \ref{lem:SimplexEuc2Kernels} implies that the real matrix
   \[
    I_2 -\sum_{j=1}^g s_{\ell,j} X_j
   \]
   has a nontrivial (real) null space for each $\ell = 1,\dots,g+1$. Also note that, the kernel containment
   \[
    \ker L_S (X) \subseteq \ker \Lambda_{(I,-S)} (H,W)
   \]
   is equivalent to
      \[
    \ker I_2 - \sum_{j=1}^g s_{\ell,j} X_j \subset \ker H - \sum_{j=1}^g s_{\ell,j} W_j
   \]
    for all $\ell = 1,\dots,g+1$. Since each $X_j$ is real, it follows that each $H - \sum_{j=1}^g s_{\ell,j} W_j$ is a $2 \times 2$ self-adjoint matrix with nontrivial real kernel. It follows that each $H - \sum_{j=1}^g s_{\ell,j} W_j$ must be real, hence $\Lambda_{(I,-S)} (\mathrm{Im} (H,W)) = 0$. 
   
   Now, assume towards a contradiction that $\mathrm{Im} ((H,W) )\neq 0$. Then $\Lambda_{(I,-S)} (\mathrm{Im} (H,W)) = 0$ implies that there exists a nonzero vector $v = (h,w) \in \R^{g+1}$ such that 
   \[
    \Lambda_{(I,-S)} (h,w) = 0. 
   \]
   In the case that $h \neq 0$, we obtain that $I = \Lambda_S(w/h)$ which implies that
   \[
    L_S(-\alpha w/h) = I+\alpha \Lambda_S(w/h) = (1+\alpha)I,
   \]
   is positive semidefinite for all $\alpha \geq -1$. On the other hand, if $h = 0$, then $\Lambda_S (w) = 0$, which implies $L_S(\alpha w) = I \succeq 0$ for all $\alpha \in \R$. In either case, we find that $\cD_S(1)$ is unbounded which contradicts our assumptions. 
   \end{proof}

The next lemma shows that if $Y \in SM_2 (\C)$ is an Arveson extreme point of the matrix interval with nonzero imaginary part and $(H,Z)$ is a nontrivial solution to the matrix extreme point kernel containment for $X$, then $H$ has nonzero imaginary part. 

  \begin{lem}
  \label{lem:IntervalHomogeneousSolutions}
Let $C = 1 \oplus -1$ so that $\cD_C$ is the matrix interval, and let $Y \in \cD_C(2)$ be an Euclidean extreme point of $\cD_C$ with $\mathrm{Im} (Y) \neq 0$. If $(H,Z) \in SM_2(\C)^2$ is a nonzero solution to the kernel containment 
  \[
    \ker L_C (Y) \subset \ker \Lambda_{(I,-C)} (H,Z)
  \]
  with $\mathrm{tr} (H) = 0$, then $\mathrm{Im} (H) \neq 0$. 
  \end{lem} 
  \begin{proof}
      Note that $Y \in \cD_C (2)$ is an Euclidean extreme point of $\cD_C$ if and only if $Y^2 = I$. Therefore, if $Y$ is an Euclidean extreme point of $\cD_C(2)$ and $\mathrm{Im} (Y) \neq 0$, then there must exist a unitary $U \in M_2 (\C)$ such that $ Y = U^*(1 \oplus -1) U$. It follows that $(H,Z)\in SM_2(\C)^2$ satisfies
      \[
    \ker L_C (Y) \subset \ker \Lambda_{(I,-C)} (H, Z )
      \]
      if and only if
            \[
    \ker L_C (1 \oplus -1) \subset \ker \Lambda_{(I,-C)} (U HU^* , UZU^* ).
      \]
     A direct computation shows that if 
     \[
     0 = \mathrm{tr} (H) = \mathrm{tr} (U^*H U)
     \]
     and $(U HU^* ,UZU^*) \in SM_2(\C)^2$ satisfies the second kernel containment above, then 
     \[
    (U HU^* , UZU^* ) = \left(\begin{pmatrix}
        h & 0 \\
        0 & -h
    \end{pmatrix},
    \begin{pmatrix}
        h & 0 \\
        0 & h
    \end{pmatrix}
    \right)
    = h(U Y U^*,I)
     \]
     for some nonzero $h \in \R$. Therefore, if $\mathrm{tr} (H) = 0$ and $(H,Z)$ satisfy the first kernel containment, then we must have $H = hY$ which has nonzero imaginary part by assumption. 
       \end{proof}

\subsubsection{Proof of Theorem \ref{thm:SimpleXIntervalComplexMatEx}}

\begin{proof}
 We first prove Item \ref{thm:SimpleXIntervalComplexMatEx} \eqref{it:RealXComplexIsMatEx}. To this end, suppose $X \in \cD_S (2)$ is a real Euclidean extreme point of $\cD_S$ which is not Arveson extreme in $\cD_S$ and that $Y \in \cD_C(2)$ is an Arveson extreme point of $\cD_C$ with nonzero imaginary part. Since $X$ is not Arveson extreme in $\cD_S$, there must be some $j\in \{1,\dots,g\}$ such that $X_j$ is not a constant multiple of the identity, see Remark \ref{rem:extreme-points-free-simplex}. It then follows from Lemma \ref{lem:SimplexHomogeneousSolutions} that every solution $(H,W)$ of the kernel containment 
\[
\ker L_S(X) \subset \ker \Lambda_{(I,-S)} (H,W)
\]
has $\mathrm{Im} (H) = 0$. On the other hand, as a consequence of Lemma \ref{lem:IntervalHomogeneousSolutions}, if $(H,Z)$ is a solution to 
\[
\ker L_C(Y) \subset \ker \Lambda_{(I,-C)} (H,Z)
\]
with $\mathrm{Im} (H) = 0$, then $H = \alpha I$ for some $\alpha \in \mathbb{R}$. Furthermore, since $X$ and $Y$ are both Euclidean extreme, using \cite[Corollary 2.3]{evert2018extreme} shows that the above kernel containments with $H= \alpha I$ force $W = \alpha X$ and $Z = \alpha Y$. It follows from Lemma \ref{lem:ExtremeOfCartesian} \eqref{it:MatExCartesian} that $(X,Y)$ is matrix extreme in $\cD_S \times \cD_C$.

    We now prove Item \ref{thm:SimpleXIntervalComplexMatEx} \eqref{it:RealEmbeddingNotMatEx}. We will show that there is a nontrivial solution to the kernel containment condition given in Lemma \ref{lem:ExtremeOfCartesian} \eqref{it:MatExCartesian}. To simplify notation in the proof we set
    \[
    \hat{Y}:=\begin{pmatrix} Y_\R & Y_\C \\
    Y_\C^T & Y_\R
    \end{pmatrix}
    \]
    We first claim that there exists an $H \in SM_2(\R)$ and a $Z \in SM_4(\R)$ such that $\mathrm{tr} (H) = 0$ and such that
    \[
    \ker L_C  (\hat{Y}) 
    \subseteq
    \ker \Lambda_{(I,-C)}  (H \oplus -H,Z).
    \]
    To this end, choose $H \in SM_2 (\R)$ so that the eigenvalues of $H$ are $1$ and $-1$ and so that 
    \[
    H Y_\R - Y_\R H = 0.
    \] 
    Using the fact that $H$ is a trace $0$ real symmetric $2 \times 2$ matrix and that $Y_\C$ is a real skew symmetric $2\times 2$ matrix, we have $Y_\C H + H Y_\C = 0$, from which it follows that $Y H - H \overline{Y} =0$. This then implies that
    \begin{equation}
    \label{eq:YYbarHskewCommuting}
    \begin{pmatrix}
        Y & 0 \\
        0 & \overline{Y}
    \end{pmatrix}
    \begin{pmatrix}
        0 & i H \\
        -i H^* & 0 
    \end{pmatrix} -
     \begin{pmatrix}
        0 & i H \\
        -i H^* & 0 
    \end{pmatrix}
  \begin{pmatrix}
        Y & 0 \\
        0 & \overline{Y}
    \end{pmatrix} = 0.
    \end{equation}
    
    Now, since $Y$ is an Arveson extreme point of $\cD_C$, its eigenvalues are $1$ and $-1$. Thus there is a unitary $V \in SM_4(\C)$ such that 
    \begin{equation}
    \label{eq:VunitaryDef}
    V^* (Y \oplus \overline{Y}) V = -I_2 \oplus I_2. 
    \end{equation}
    The commuting condition in equation \eqref{eq:YYbarHskewCommuting} then implies that there exists self-adjoint matrices $C_1,C_2 \in SM_2 (\C)$ such that
    \begin{equation}
    \label{eq:VBlockDiagonalizesHskew}
    V^* \begin{pmatrix}
        0 & i H \\
        -i H^* & 0 
    \end{pmatrix} V = C_1 \oplus C_2. 
    \end{equation}
    Furthermore, since the eigenvalues of $H$ are $1$ and $-1$, the eigenvalues of both $C_1$ and $C_2$ must be $1$ and $-1$, so $C_1$ and $C_2$ are nonzero. 
    
    Set $W = V (-C_1 \oplus C_2) V^*$ and observe that
    \[
    L_C (-I_2 \oplus I_2) = 2I_2 \oplus 0_4 \oplus 2I_2 \qquad \mathrm{and} \qquad \Lambda_{(I,-C)} (C_1\oplus C_2,-C_1\oplus C_2) = 2C_1 \oplus 0_4 \oplus 2C_2,
    \]
    hence
    \[
    \ker L_C (-I_2 \oplus I_2) \subseteq \ker \Lambda_{(I,-C)} (C_1\oplus C_2,-C_1\oplus C_2),
    \]
    from which it follows that
     \begin{equation}
     \label{eq:YYbarKernelContainment}
    \ker L_C (Y \oplus \overline{Y}) \subseteq \ker \Lambda_{(I,-C)} \left(\begin{pmatrix}
        0 & i H \\
        -i H^* & 0 
    \end{pmatrix} , W\right).
    \end{equation}

    Now, let $U \in SM_4 (\C)$ be the unitary
    \[
    U = \frac{\sqrt{2}}{2}\begin{pmatrix} I_2 & -i I_2 \\
    -i I_2 & I_2
    \end{pmatrix}.
    \]
    Then direct computations show
    \[
    U^* (Y \oplus \overline{Y}) U =  \hat{Y} \qquad \mathrm{and} \qquad U^* \begin{pmatrix}
        0 & i H\\
        -i H^* & 0 
    \end{pmatrix} U = H \oplus -H.
    \]
    Furthermore, it one sets $Z = U^* W U$, then equation \eqref{eq:YYbarKernelContainment} implies that
    \begin{equation}
    \label{eq:YemKerContain}
    \ker L_C (\hat{Y}) \subseteq \ker \Lambda_{(I,-C)} \left(H \oplus -H , Z\right),
    \end{equation}
     which proves the desired kernel containment. 
     
     To prove the initial claim, we must show that $Z$ is real. To this end, first observe that the kernel containment \eqref{eq:YemKerContain} implies
     \begin{equation}
    \label{eq:ZRealOnYemKer}
    \left(I \otimes (H \oplus -H)\right)  v = (C \otimes Z) v \qquad \mathrm{for \ all} \qquad v \in \ker \Lambda_{(I,-C)} (I_4,\hat{Y}) \subset \R^{8}
    \end{equation}
    Similarly, using equations \eqref{eq:VunitaryDef} and \eqref{eq:VBlockDiagonalizesHskew} together with the definition of $H$ shows that
   \[
    \ker L_C (-\hat{Y}) \subseteq \ker \Lambda_{(I,-C)} \left(H \oplus -H , -Z\right) \subset \R^{8},
    \]
    which implies that
    \begin{equation}
    \label{eq:ZRealOnMinusYemKer}
    \left(I_2 \otimes (H \oplus -H)\right) v = -(C \otimes Z) v \qquad \mathrm{for \ all} \qquad v \in \ker \Lambda_{(I_2,-B)} (I_4,-\hat{Y}) \subset \R^{8} 
    \end{equation}
    Using the fact that $\hat{Y} \in SM_4(\R)$ has eigenvalues $1$ and $-1$ each with multiplicity $2$, with have
    \[
     \R^{8} =\ker L_C (-\hat{Y}) \oplus \ker L_C (\hat{Y}).
    \]
    Therefore, since $\left(I \otimes (H \oplus -H)\right)$ is a real matrix, combining equations \eqref{eq:ZRealOnYemKer} and \eqref{eq:ZRealOnMinusYemKer} shows that $C \otimes Z$ is real, hence $Z$ is real. 

    Finally, to complete the proof, using Lemma \ref{lem:exists-W} shows that since $X$ is real Euclidean extreme point of $\cD_S(2)$ that is not Arveson extreme in $\cD_S$,  
     there exist $W_1,W_2 \in SM_2(\R)^g$ such that
    \begin{equation}
    \label{eq:EmbeddingBadMatSimplexSolution}
    \ker L_S (X \oplus X) \subseteq \ker \Lambda_{(I,-S)} (H \oplus -H,W_1 \oplus W_2).
    \end{equation}
    By construction, $H \oplus -H \neq \alpha I$ for any $\alpha \in \R$, so, using Lemma \ref{lem:ExtremeOfCartesian} \eqref{it:MatExCartesian} together with equations \eqref{eq:YemKerContain} and \eqref{eq:EmbeddingBadMatSimplexSolution} shows that $(X\oplus X,\hat{Y})$ is not matrix extreme in $\cD_S \times \cD_C$. 
\end{proof}

\section{Optimizers for some optimization problems involving Cartesian products of free simplices}\label{sec:optimizers-direct-sums}

This section is primarily dedicated to solving the $n=2$ optimization problem \eqref{eq:SumEigMaxFixedn} for the  real free spectrahedron $\cD_{A(k)} \subset SM(\R)^{k+1}$ that is defined as follows. Let $e_j(k) \in \mathbb{R}^{k}$ denote the $j$-th standard basis vector in $\R^k$ and set $E_j(k) = e_j(k) e_j(k)^T \in \mathbb{R}^{k \times k}$. The tuple $A(k)=(A_1(k),\dots,A_{k+1}(k))$ is given by
\begin{equation}
\label{eq:AkDefinition}
\begin{split}
A_j(k) &= E_j(k) \oplus \operatorname{diag}(-1,0,0) \in \R^{k+3 \times k+3} \quad \mathrm{for \ } j=1,\dots,k,  \\
 A_{k+1}(k) &= 0_{k+1} \oplus \operatorname{diag}(-1,1) \in \R^{k+3 \times k+3}.
 \end{split}
\end{equation}
In particular, the free spectrahedron $\cD_{A(k)}$ is a Cartesian product of a real free simplex in $k$ variables and the real matrix interval. Moreover, $\cD_{A(k)}$ is obtained by an invertible linear transformation of the real matrix jewel $D_{\jewel, \mathbf{k}}$ where $\mathbf{k}=(k,1)$. Notably, invertible linear transformations do not affect inclusion constants, see \cite[Lemma 3.1]{passer2018minimal}.

We warn the reader that while all bounded simplices are related by an invertible affine linear transformation, only specific pairs of simplices are related by an invertible linear transformation. As a consequence of the upcoming Theorem \ref{thm:kSimplexPlusLineOptimum}, invertible affine transformations do not preserve inclusion constants. 

\subsection{Cartesian product of line and simplex} 

\begin{thm}
\label{thm:kSimplexPlusLineOptimum} The maximum value achieved by
\begin{equation}
\label{eq:SumEigMaxFixedGenK}
\begin{split}
  \mathrm{max} \quad & \lambda_{\max} \sum_{i=1}^{k+3} \left(X_i \otimes Y_i\right) \\
    \mathrm{s.t.} \quad
    & X \in \cD^{\R}_{A(k)}(2), \quad Y\in \cD_{A(k)}^{\square, \R}(2)\, ,
    \end{split}
\end{equation}
is equal to 
\[
\gamma(k):=\gamk{k} .
\]
Furthermore, this maximum value is achieved when $X$ is given by
\[
X_1 = \begin{pmatrix}
1 & 0 \\
0 & -k
\end{pmatrix}, \quad 
X_2 = \begin{pmatrix}
    -k & 0 \\
    0 & 1
\end{pmatrix}, \quad
X_{k+1} = \begin{pmatrix}
    0 & 1 \\
    1 & 0
\end{pmatrix}
\]
and $X_j = I_2$ for all $j \neq 1,2,k+1$ and when $Y$ is given by 
\[
Y_1 = \begin{pmatrix}
\frac{1+\sqrt{1+k}}{1+k+2\sqrt{1+k}} & 0 \\
0 & -\frac{1}{1+k+2\sqrt{1+k}}
\end{pmatrix}, \ 
Y_2 = \begin{pmatrix}
    -\frac{1}{1+k+2\sqrt{1+k}} & 0 \\
    0 & \frac{1+\sqrt{1+k}}{1+k+2\sqrt{1+k}}
\end{pmatrix}, \ 
Y_{k+1} = \begin{pmatrix}
    0 & \frac{k}{k-1+\sqrt{1+k}} \\
   \frac{k}{k-1+\sqrt{1+k}} & 0
\end{pmatrix},
\]
and $Y_j = 0$ for all $j \neq 1,2,k+1$.
\end{thm}

\begin{proof} 
    By computing the classical extreme points of $\cD_{A(k)}(1)$, one finds that $\cD_{A(k)}^\square = \cD_{B(k)}$ where
\begin{equation} \label{eq:line-plus-simplex_dual-FS-1}
B_{j}(k) = I_2 \otimes \big(I_{k+1}-(k+1) E_{j}(k)\big) \qquad \mathrm{for} \qquad j=1,\dots,k
\end{equation}
and
\begin{equation} \label{eq:line-plus-simplex_dual-FS-2bis}
B_{k+1}(k) = \operatorname{diag} (1,-1) \otimes I_{k+1}.
\end{equation}
    It is then straightforward to check that the tuples $X=(X_1,X_2,\dots,X_{k+1})$ and $Y=(Y_1,Y_2,\dots,Y_{k+1})$ above satisfy $X \in \cD_{A(k)}(2)$ and $Y \in \cD_{A(k)}^\square (2)$ and
    \[
   \gamma(k) = \lambda_{\max} \left( \sum_{j=1}^k X_j \otimes Y_j \right),
    \]
    hence the maximum is at least $\gamma(k)$.
    
    Using Theorem \ref{thm:MaxEigIsMinBallInclusionConst}, to show $\gamma(k)$ is the maximum it is sufficient to show that 
    \[
    \delmat\cD_{A(k)} (2) \subset \cWmin(\gamma(k)  \cD_{A(k)}(1)).
    \]
     By definition, $\cWmin(\gamma(k) \cD_{A(k)}(1))$ is the matrix convex hull of the classical extreme points of $\gamma(k) \cD_{A(k)}(1)$. Since $\gamma(k) \cD_{A(k)}(1)$ only has finitely many extreme points, we can use \cite[Section 4.1]{helton_matricial_2013}  to reformulate the containment $X \in \cWmin(\gamma(k)  \cD_{A(k)}(1))$ as a feasibility SDP for fixed $X \in \delmat\cD_{A(k)} (2)$ (see also Proposition \ref{prop:spectrahedron-mcset-equivalent}).

     We use Theorem \ref{thm:SimplexXIntervalRealMatEx} to classify (up to unitary equivalence) the real matrix extreme points of $\cD_{A(k)}(2)$, see Proposition \ref{prop:TwoSimplexXIntervalMatEx}. Then, by studying the feasibility SDP in question, we show that it suffices to only consider a family of matrix extreme points $\{X(\theta)\}$ defined in terms of $\theta \in [0,\pi/2]$. Here 
     \[
 X(\theta) = (X_1,\dots,X_n,X_{n+1} (\theta)) \in \delmat\cD_{A(k)}(2)
    \]
   is defined by
\[
X_1 = \begin{pmatrix}
1 & 0 \\
0 & -k
\end{pmatrix}, \quad 
X_2 = \begin{pmatrix}
    -k & 0 \\
    0 & 1
\end{pmatrix}, \quad
X_{k+1}(\theta) = \begin{pmatrix}
    \cos(\theta) & \sin(\theta) \\
    \sin(\theta) & -\cos(\theta)
\end{pmatrix}
\]
 and $X_j = I_2$ for all $j =3,4,\dots,k$. The proof is completed by constructing exact solutions in terms of $k$ and $\theta$ for the feasibility SDP in question. This will be done in Sections \ref{sec:proof-k-2} and \ref{sec:proof-k-bigger-2}. 
\end{proof}

We in fact conjecture that the above result holds even if one does not restrict to level two of $\cD_{A(k)}$ and $\cD_{A(k)}^\square$. That is, we conjecture that the optimum is still $\frac{2k}{k-1+\sqrt{1+k}}$ even if one searches over all of $\cD_{A(k)}$ and $\cD_{A(k)}^\square$.

\begin{conj}
\label{conj:kSimplexPlusLineOptimum} Let $\cD_{A(k)}$ be as defined in Theorem \ref{thm:kSimplexPlusLineOptimum}. Then the maximum value achieved by
\begin{equation}
\label{eq:SumEigMaxFixedGenKbis}
\begin{split}
  \mathrm{max} \quad & \lambda_{\max} \sum \left(X_i \otimes Y_i\right) \\
    \mathrm{s.t.} \quad
    & X \in \cD_A, \quad Y \in \cD_A^\square.
    \end{split}
\end{equation}
is equal to 
\[
\gamma(k)=\frac{2k}{k-1+\sqrt{1+k}}.
\]
Furthermore, this maximum value is achieved when $X$ is defined by
\[
X_1 = \begin{pmatrix}
1 & 0 \\
0 & -k
\end{pmatrix}, \quad 
X_2 = \begin{pmatrix}
    -k & 0 \\
    0 & 1
\end{pmatrix}, \quad
X_{k+1} = \begin{pmatrix}
    0 & 1 \\
    1 & 0
\end{pmatrix}
\]
and $X_j = I_2$ for all $j \neq 1,2,k+1$ and when $Y$ is defined by 
\[
Y_1 = \begin{pmatrix}
\frac{1+\sqrt{1+k}}{1+k+2\sqrt{1+k}} & 0 \\
0 & -\frac{1}{1+k+2\sqrt{1+k}}
\end{pmatrix}, \ 
Y_2 = \begin{pmatrix}
    -\frac{1}{1+k+2\sqrt{1+k}} & 0 \\
    0 & \frac{1+\sqrt{1+k}}{1+k+2\sqrt{1+k}}
\end{pmatrix}, \ 
Y_{k+1} = \begin{pmatrix}
    0 & \frac{k}{k-1+\sqrt{1+k}} \\
   \frac{k}{k-1+\sqrt{1+k}} & 0
\end{pmatrix},
\]
and $Y_j = 0$ for all $j \neq 1,2,k+1$. 
\end{conj}

\evidence 
For the dimension free version of the problem where the $X_j$ and $Y_j$ are allowed to have any size $n$, we have run systematic computer experiments for $k=2,3,4,5$ and $n=3,4,5,6,7,8$. Theorem \ref{thm:kSimplexPlusLineOptimum} already gives a lower bound for the objective value, so it is appropriate to use the formulation from Theorem \ref{thm:MaxEigIsMinBallInclusionConst}. 

In particular, for each above pair $(k,n)$ use the Mathematica package NCSE \cite{evert2021NCSE} to randomly generate 100,000 real extreme points $X$ at level $n$ of the free spectrahedron $\cD_{A(k)}$. For each $X$ we record the smallest constant $\gamma_X \geq 1$ such that
\[
X \in \gamma_X \mathcal{W}^{\min} (\cD_{A(k)}(1)).
\]
For each $(k,n)$ we record $\gamma_{k,n}$, the maximum of the observed constants $\gamma_X$. If the conjecture holds, then one would expect to see that $\gamma_{k,n} \leq \gamma(k)$ for all such $k,n$, and this is indeed what we observe. Furthermore, we typically observe that $\gamma(k)-\gamma_{k,n}$ is small. We report the observed values of $\log_{10} (\gamma(k)-\gamma_{k,n})$  in Table \ref{tab:gamkminusgamkn}. The fact that $\gamma(k)-\gamma_{k,n}$ is small suggests our method has done a reasonably good job of exploring the possible choices for extreme points $X$. We do observe that $\gamma(k)-\gamma_{k,n}$ is typically larger when $k$ and $n$ are larger. Assuming Conjecture \ref{conj:kSimplexPlusLineOptimum} is true, this is unsurprising, as the search space becomes larger as $k$ and $n$ increase. 

The methodology used to generate extreme points of $\cD_{A(k)}$ is the same as is used in \cite{evert2023empirical}. In particular, as $\cD_{A(k)} (n)$ is itself a spectrahedron for fixed $k,n$, we can use semidefinite programming to optimize randomly generated linear functionals over $\cD_{A(k)}(n)$. As discussed in \cite{evert2023empirical}, these optimizers are typically observed to be free extreme points. The linear functionals are themselves generated by taking coefficients from a bounded uniform distribution. The choice to use a uniform distribution for coefficients instead of the normal distribution has empirically been observed to make little difference, see \cite[Section 3.1.2]{evert2023empirical}.
We additionally note that while these experiments are done over the reals, since $\cD_{A(k)}$ is closed under complex conjugation, any complex element implicitly considered complex elements of $\cD_{A(k)} (n)$ for $n=2,3,4$.

\begin{table}[]
    \centering
\begin{tabular}{c|c|c|c|c|c|c|c|c|}
 \multicolumn{2}{c}{ }  & \multicolumn{7}{c}{$n$} \\
    \cline{3-9}
  \multicolumn{2}{c|}{ } & 2  & 3 & 4 & 5 & 6 & 7 & 8 \\ \cline{2-9} 
    \multirow{4}{*}{$k$} 
     & 2 & -7.75  & -8.18 & -8.04  & -7.63 & -8.21 & -6.96 & -6.57 \\ \cline{2-9} 
     & 3 & -8.12  & -7.78 & -7.51  & -8.06 & -6.32 & -4.89 & -5.79 \\ \cline{2-9} 
     & 4 & -7.46  & -7.42 & -7.19  & -6.58 & -5.01 & -4.66 & -3.70 \\ \cline{2-9} 
     & 5 & -8.04  & -7.27 & -6.68  & -6.75 & -3.86 & -4.22 & -3.12 \\ \cline{2-9} 
\end{tabular}
    \caption{Observed values for $\log_{10} (\gamma(k)-\gamma_{k,n})$, where $\gamma_{k,n}$ is the largest numerically observed inclusion constant from optimization problem \eqref{eq:MinBallContainmentFixedn} across $100,000$ randomly generated extreme points of $\cD_{A(k)} (n)$. }
    \label{tab:gamkminusgamkn}
\end{table}

\subsubsection{How we obtained the form of the optimizers} We now discuss how the form of the optimizers $X$ and $Y$ appearing in Theorem \ref{thm:kSimplexPlusLineOptimum} was initially obtained. In the following, let $B(2)$ denote a minimal defining tuple for $\cD_{A(2)}^\square.$

To begin, we focused on the case $k=2$ and used NCSE \cite{evert2021NCSE} to randomly generate large numbers of extreme points of the free spectrahedra $\cD_{A(2)}$ and $\cD_{B(2)}$. Then, for each computed extreme point $X \in \cD_{A(2)}$ and $Y \in \cD_{B(2)}$, we computed the quantity $\lambda_{\max} \sum (X_i \otimes Y_i)$. We then by hand examined pairs that led to large values. This quickly led to the observation that taking
\[
X = \left( \begin{pmatrix}
    1 & 0 \\
    0 & -2 
\end{pmatrix},\begin{pmatrix}
    -2 & 0 \\
    0 & -1 
\end{pmatrix},
\begin{pmatrix}
    0 & 1 \\
    1 & 0 
\end{pmatrix}\right),
\]
as in the theorem tended to give larger values of $\lambda_{\max} \sum (X_i \otimes Y_i)$. 

We mention that this choice of $X$ is intuitively natural for at least two reasons. First, such a tuple is a Cartesian product of Arveson extreme points of the free simplex and matrix interval, hence is Arveson extreme in $\cD_A$. Second, given the characterization of the optimization problem in Theorem \ref{thm:MaxEigIsMinBallInclusionConst}, the optimization problem can be viewed as finding the non-commutative tuple $X \in \cD_{A(2)}$ that is ``most difficult" to embed into a commutative tuple $Z \in \gamma \cD_{A(2)}$ where $\gamma \geq 1$. As the Arveson extreme points of a free simplex are, up to unitary equivalence, tuples of diagonal matrices, taking an Arveson extreme point of the interval that is anti-diagonal intuitively makes the tuple X ``more non-commutative". 

Having a guess for the form of $X$, we turned out attention to $Y$. We first observed that, up to unitary equivalence, tuples $Y$ that were near optimal were approximately of the form
\[
Y = \left( \begin{pmatrix}
    a & 0 \\
    0 & b 
\end{pmatrix},\begin{pmatrix}
    b & 0 \\
    0 & a 
\end{pmatrix},
\begin{pmatrix}
    0 & c \\
    c & 0 
\end{pmatrix}\right),
\]
for some $a,b,c \in \R$. To find an exact tuple $Y$, we used the fact an optimal $Y$ must be an extreme point of $\cD_{B(2)}$ together with the observation that extreme points of a free spectrahedron tend to have large kernel dimension when evaluated in the corresponding linear pencil. See \cite{evert2023empirical} for a detailed discussion of this phenomenon. In particular, evaluating near optimal tuples $Y$ in the defining equation 
\[
I-Y_1+2Y_2-Y_3 \succeq 0
\]
led to the numerical observation that
\begin{equation}
\label{eq:DefiningEqFormGuess}
-Y_1+2Y_2-Y_3 =  \begin{pmatrix}
    -a+2b & -c \\
    -c & -b+2a
\end{pmatrix} = \begin{pmatrix}
    -c & -c \\
    -c & 1
\end{pmatrix},
\end{equation}
and that $c$ should be chosen so that this matrix has negative one as an eigenvalue. This leads to two possible choices for $c$, namely $c=-1+\sqrt{3}$ and $c = -1-\sqrt{3}$. However, adding appropriately chosen defining equations for $\cD_{B(2)}$ shows that if $Y \in \cD_{B(2)}$, then $I-Y_3 \succeq 0$, which is not the case if $c = -1-\sqrt{3}$. This led us to conjecture that $c =-1+
\sqrt{3}$. From here, one can solve for $a$ and $b$ in equation \eqref{eq:DefiningEqFormGuess}.

The above of course only treated the $k=2$ case of the problem. To generalize to other values of $k$, we repeated the above choice of $k$ for $k=2,3,4,5$ and $6$, obtain conjectured optimum values for each case. The conjectured form for the optimal $X$ is a straightforward generalization of the $k=2$ case, and, fortunately, we found the pattern 
\[
\gamma(k) = \frac{2k}{k-1+\sqrt{k+1}}.
\]
This then enabled us to compute a conjectured form for the optimal $Y$.

\subsubsection{Proof of Theorem \ref{thm:kSimplexPlusLineOptimum} for $k=2$} \label{sec:proof-k-2}

Given a tuple $X \in SM_2(\C^3)$ and a constant $\gamma \geq 0$, let   $\MinBallSDP{\gamma}{X} \subset SM_2(\C^6)$ denote those tuples $(C_1,C_2,\dots,C_6) \in SM_2(\C^6)$ that satisfy the following feasibility SDP
\begin{equation}
\label{eq:XMinBallSDP}
\begin{split}
&  \oplus_{i=1}^6 C_i \succeq 0,  \\
& C_1-2C_2+C_3+C_4-2C_5+C_6 =X_1, \\
& C_1+C_2-2C_3+C_4+C_5-2C_6 =X_2, \\
& C_1+C_2+C_3-C_4-C_5-C_6 =X_3, \\
& C_1+C_2+C_3+C_4+C_5+C_6 = \gamma I,
\end{split}
\end{equation}

\begin{prop}
\label{prop:HKMFeasibilitySDP}
    Let $\cD_A$ be the free simplex in two variables direct sum the interval and let $X \in SM_2(\C^3)$. Then $X \in \gamma \cWmin(\cD_A(1))$ if and only if $\MinBallSDP{\gamma}{X}$ is nonempty. 
\end{prop}
\begin{proof}
    Using Proposition \ref{prop:spectrahedron-mcset-equivalent}, we infer that $X \in \gamma \cWmin(\cD_A(1))$ if and only if $ \mathcal W^{\max}(\mathcal D_A(1)^\bullet) \subseteq \gamma \mathcal D_X$. This inclusion can be checked using an SDP \cite[Section 4.1]{helton_matricial_2013}, namely the optimization problem in \eqref{eq:inclusion-free-spectra-SDP}, as $\mathcal W^{\max}(\mathcal D_A(1)^\bullet)=\mathcal D_{A(2)}^\square$ is a free spectrahedron defined by diagonal matrices (see \eqref{eq:line-plus-simplex_dual-FS-1} and \eqref{eq:line-plus-simplex_dual-FS-2bis}). For the situation at hand, we recover the SDP in \eqref{eq:XMinBallSDP}. Thus, $\MinBallSDP{\gamma}{X}$ is nonempty if and only if $\mathcal W^{\max}(\mathcal D_A(1)^\bullet) \subseteq \gamma \mathcal D_X$.
\end{proof}

\begin{prop}
\label{prop:TwoSimplexXIntervalMatEx}
    Up to unitary equivalence, every real matrix extreme point of $\cD_{A(2)} (2)$ has the form
    \[
X(\theta) = (X_1,X_2,X_3 (\theta)) = \left(\begin{pmatrix}
1 & 0 \\
0 & -2
\end{pmatrix},\begin{pmatrix}
    -2 & 0 \\
    0 & 1
\end{pmatrix},
\begin{pmatrix}
    \cos(\theta) & \sin(\theta) \\
    \sin(\theta) & -\cos(\theta)
\end{pmatrix}\right),
\]
or
\[
Y(\theta) = (Y_1,Y_2,Y_3 (\theta)) = \left(\begin{pmatrix}
1 & 0 \\
0 & 1
\end{pmatrix},
\begin{pmatrix}
    -2 & 0 \\
    0 & 1
\end{pmatrix},
\begin{pmatrix}
    \cos(\theta) & \sin(\theta) \\
    \sin(\theta) & -\cos(\theta)
\end{pmatrix}\right),
\]
or
\[
Z(\theta) = (Z_1,Z_2,Z_3 (\theta)) = \left(\begin{pmatrix}
-2 & 0 \\
0 & 1
\end{pmatrix},
\begin{pmatrix}
    1 & 0 \\
    0 & 1
\end{pmatrix},
\begin{pmatrix}
    \cos(\theta) & \sin(\theta) \\
    \sin(\theta) & -\cos(\theta)
\end{pmatrix}\right).
\]
for some $\theta \in [0,\pi]$. Moreover, if $\MinBallSDP{\gamma}{X(\theta)}$ is nonempty for all $\theta \in [0,\pi/2]$ and some $\gamma \geq 0$, then $\MinBallSDP{\gamma}{W}$ is nonempty for all real matrix extreme points $W$ of $\cD_A$. 
\end{prop}

\begin{proof}
    We first prove that every matrix extreme point of $\cD_A$ is unitarily equivalent to one of the three above forms. From Theorem \ref{thm:SimplexXIntervalRealMatEx}, we know that any matrix extreme point of $X = (X_1, X_2,X_3) \in \cD_A$ satisfies that $(X_1,X_2)$ is an Arveson extreme point of the simplex and that $X_3$ is an Arveson extreme point of the interval. Since all Arveson extreme points of a simplex are, up to unitary equivalence, a direct sum of extreme points of $\cD_A (1)$, it follows that, up to unitary equivalence, we can take $(X_1,X_2)$ of the form 
    \[
    (X_1,X_2) = \left(\begin{pmatrix}
        a_1 & 0 \\
        0 & a_2
    \end{pmatrix},\begin{pmatrix}
        b_1 & 0 \\
        0 & b_2
    \end{pmatrix}\right)
    \]
    where $(a_j,b_j) \in \{(-2,1),(1,-2),(1,1)\}$ for $j=1,2$. Also, since $X$ is irreducible, we have $(a_1,b_1) \neq (a_2,b_2)$. This gives six possible choices for $(X_1,X_2)$.
    
   Next, since $X_3$ is an Arveson extreme point of the real matrix interval, we have that $X_3$ is a real valued self-adjoint unitary. Additionally, by conjugating $X$ by the unitary $1 \oplus -1$ if necessary, we can assume the off-diagonal entries of $X$ are positive. Furthermore, the determinant of $X_3$ must be $-1$, otherwise $X_3$ would be a constant multiple of the identity and $(X_1,X_2,X_3)$ would be reducible. It is straightforward to check that such a matrix must has the form
    \[
    X_3 = \begin{pmatrix}
    \cos(\theta) & \sin(\theta) \\
    \sin(\theta) & -\cos(\theta)
\end{pmatrix}, \qquad \qquad \mathrm{for \ } \theta \in [0,\pi]
    \]
    
    It remains to reduce from six choices of $(X_1,X_2)$ to the three. Note that for $\phi \in [0,\pi]$, one has
    \[
   \begin{pmatrix}
       0 & 1 \\
       1 & 0
   \end{pmatrix}  \left(\begin{pmatrix}
-2 & 0 \\
0 & 1
\end{pmatrix},\begin{pmatrix}
    1 & 0 \\
    0 & -2
\end{pmatrix},
\begin{pmatrix}
    \cos(\phi) & \sin(\phi) \\
    \sin(\phi) & -\cos(\phi)
\end{pmatrix}\right)
    \begin{pmatrix}
       0 & 1 \\
       1 & 0
   \end{pmatrix} = X(\pi -\phi). 
   \]
   Setting $\theta = \pi-\phi \in [0,\pi]$ shows that matrix extreme points of this form are unitarily equivalent to $X(\theta)$. Similar arguments can be made for the remaining possible choices of $(X_1,X_2)$.

   We now prove the second part of the proposition. From the first part, it is sufficient to show that if $\MinBallSDP{\gamma}{X(\theta)}$ is nonempty for all $\theta \in [0,\pi/2]$ then $\MinBallSDP{\gamma}{X(\phi)}$ and $\MinBallSDP{\gamma}{Y(\phi)}$ and $\MinBallSDP{\gamma}{Z(\phi)}$ are each nonempty for all $\phi \in [0,\pi]$. To this end, first assume that $\MinBallSDP{\gamma}{X(\phi)}$ is nonempty for some $\phi \in [0,\pi]$. That is, there exists a tuple $C=(C_1,C_2,C_3,C_4,C_5,C_6) \in SM_2 (\C^6)$ that is a solution to equation \eqref{eq:XMinBallSDP}, with $(X_1,X_2,X_3)$ in the equation equal to $X(\phi)$. Then by adding the negatives of the first and second equation of the feasibility SDP \eqref{eq:XMinBallSDPkSimplex}, one finds that
   \[
    -2C_1+C_2+C_3-2C_4+C_5+C_6 = -X_1 -X_2 = I_2.
   \]
   It follows that $(C_2,C_1,C_3,C_5,C_4,C_6) \in \MinBallSDP{\gamma}{Y(\phi)}$ and that $(C_2,C_3,C_1,C_5,C_6,C_4) \in \MinBallSDP{\gamma}{Z(\phi)}$.

   Next let $\phi \in [\pi/2,\pi]$. The using arguments similar to the above, we find that $X(\phi)$ is unitarily equivalent to $X'(\theta):=(X_1,X_2,-X_3(\theta))$ where $\theta = \pi-\phi \in [0,\pi/2]$. Thus $\MinBallSDP{\gamma}{X(\phi)}$ is nonempty if and only if $\MinBallSDP{\gamma}{X'(\theta)}$ is. Furthermore, it is straightforward to check that $(C_1,C_2,C_3,C_4,C_5,C_6) \in \MinBallSDP{\gamma}{X(\theta)}$ if and only if $(C_4,C_5,C_6,C_1,C_2,C_3) \in \MinBallSDP{\gamma}{X'(\theta)}$, which completes the proof. 
\end{proof}

\begin{remark}
    In fact, if $X$ is assumed to be a complex valued Arveson extreme point of $\cD_A(2)$, then $X$ is again unitarily equivalent to one of $X(\theta)$ or $ Y(\theta)$ or $Z(\theta)$ from Proposition \ref{prop:TwoSimplexXIntervalMatEx}. To see this, note that $(X_1,X_2)$ still must be Arveson extreme in the free simplex, so up to unitary equivalence we have the same six choices for $(X_1,X_2)$. The only barrier is that $X_3$ may have complex off-diagonal entries. These can be rotated to be real by conjugating $X$ by a unitary of the form $1 \oplus e^{i \theta}$ for some appropriately chosen $\theta$. 

    Thus, the barrier to extending Theorem \ref{thm:kSimplexPlusLineOptimum} to the complexes are those matrix extreme points in $\cD_{A(k)}$ which are not a Cartesian product of matrix extreme points. 
\end{remark}

\begin{thm} 
\label{thm:SDPSolutions}
    Set $\gamma:= \frac{4}{1+\sqrt{3}}$. Then $\MinBallSDP{\gamma}{X(\theta)}$ is nonempty for all $\theta \in [0,\pi/2]$. 
\end{thm}
\begin{proof}
    The proof is accomplished by explicitly constructing an element of $\MinBallSDP{\gamma}{X(\theta)}$ for arbitrary $\theta \in [0,\pi/2]$. It can be found in Appendix \ref{sec:Proof-of-SDPSolutions}.
\end{proof}

\subsubsection{Proof of Theorem \ref{thm:kSimplexPlusLineOptimum} for $k>2$} \label{sec:proof-k-bigger-2}

We now show that, for the level-fixed version of the optimization problem over the Cartesian product of the $k$-variable simplex and a line described in equation \eqref{eq:SumEigMaxFixedGenK} where $X$ and $Y$ are both tuples of $2 \times 2$ matrices, we can reduce to considering a simplex in two variables.

\begin{prop}
\label{prop:kVarsSimplexToTwoVar}
Fix $k \geq 3$ and let $S \subset SM(\C)^3$ be the tuple
\[
S_1 = \operatorname{diag}(1,0,-1/(k-1),0,0), \quad S_2= \operatorname{diag}(0,1,-1/(k-1),0,0), \quad S_3= \operatorname{diag}(0,0,0,-1,1),
\]
and let $A(k)$ be as defined in equation \eqref{eq:AkDefinition}. Then for $\gamma \in \R$ one has
\[
\left( \begin{pmatrix}
1 & 0 \\
0 & -k
\end{pmatrix},\begin{pmatrix}
-k & 0 \\
0 & 1
\end{pmatrix},\begin{pmatrix}
    \cos(\theta) & \sin(\theta) \\
    \sin(\theta) & -\cos(\theta)
\end{pmatrix}\right) \subset \gamma\cWmin(\cD_S(1)),\]
if and only if
\[
\cD_{A(k)}(2) \subset \gamma\cWmin(\cD_{A(k)}(1)).
\]
\end{prop}
\begin{proof}
    Following similar arguments to those used in the proof of Proposition \ref{prop:TwoSimplexXIntervalMatEx} one can argue that, to show $\cD_{A(k)}(2) \subset \gamma\cWmin(\cD_{A(k)}(1))$ it is sufficient to show that
    \[
    X(\theta) = (X_1,\dots,X_n,X_{n+1} (\theta)) \in \gamma\cWmin(\cD_{A(k)}(1))
    \]
    for all $\theta \in [0,\pi/2]$ where 
\[
X_1 = \begin{pmatrix}
1 & 0 \\
0 & -k
\end{pmatrix}, \quad 
X_2 = \begin{pmatrix}
    -k & 0 \\
    0 & 1
\end{pmatrix}, \quad
X_{k+1}(\theta) = \begin{pmatrix}
    \cos(\theta) & \sin(\theta) \\
    \sin(\theta) & -\cos(\theta)
\end{pmatrix}
\]
and $X_j = I_2$ for all $j =3,4,\dots,k$. 

In particular, using Theorem \ref{thm:SimplexXIntervalRealMatEx} together with the fact that the extreme points of $\cD_A (1)$ are of the form $(x_1,\dots,x_k)$ where either $x_j = 1$ for all $j=1,\dots,k$ or where $x_\ell = -k$ for exactly one $\ell \in \{1,\dots,k\}$ and $x_j = 1$ for all $j \neq \ell$ shows that, up to unitary equivalence, there are only finitely many families of extreme points to consider. Each such family is parameterized by $X_{k+1}(\theta)$ of the above form with $\theta \in [0,\pi].$ As in the proof of the second part of Proposition \ref{prop:TwoSimplexXIntervalMatEx}, the fact that out of these finitely many families one  only needs to consider the tuples $X(\theta)$ can be proved by showing that solutions to the relevant SDP for any of one these finitely many families correspond to solutions to the SDP for $X(\theta)$. 

Continuing to follow the argument of Proposition \ref{prop:TwoSimplexXIntervalMatEx}, one can now show that  $\delmat\cD_{A(k)}(2) \subseteq \cWmin(\gamma \cD_{A(k)}(1))$ if and only if the SDP
\begin{equation}
\label{eq:XMinBallSDPkSimplex}
\begin{split}
&  C_1 \oplus \dots \oplus C_{k+1} \oplus D_1 \oplus \dots \oplus D_{k+1} \succeq 0,  \\
& C_1-kC_2+C_3+C_4+\dots+C_{k+1}+D_1-kD_2+D_3+D_4+\dots+D_{k+1} =X_1, \\
&  C_1+C_2-kC_3+C_4+\dots+C_{k+1}+D_1+D_2-kD_3+D_4+\dots+D_{k+1} =X_2, \\
& C_1+\dots+C_{k+1}-D_1-\dots-D_{k+1} =X_{k+1} (\theta), \\
& C_1+\dots+C_{k+1}+D_1+\dots+D_{k+1} = \gamma I, \\
& -k C_\ell +  \sum_{j \neq \ell} C_j-k D_\ell +  \sum_{j \neq \ell} D_j  = I \qquad \mathrm{for} \qquad \ell=4,\dots,k,
\end{split}
\end{equation}
is feasible for all $\theta \in [0,\pi/2]$. Similarly, one can show that $(X_1,X_2,X_{k+1} (\theta)) \subset \cWmin(\gamma \cD_S(1))$ if and only if
\begin{equation}
\label{eq:XMinBallSDPtwoSimplexScaled}
\begin{split}
&  C'_1 \oplus C'_2 \oplus C'_3 \oplus D'_1 \oplus D'_2 \oplus D'_3 \succeq 0,  \\
& C'_1-kC'_2+C'_3+D'_1-kD'_2+D'_3 =X_1, \\
&  C'_1+C'_2-kC'_3+D'_1+D'_2-kD'_3 =X_2, \\
& C'_1+C'_2+C'_3-D'_1-D'_2-D'_3 =X_{k+1} (\theta), \\
& C'_1+C'_2+C'_3+D'_1+D'_2+D'_3  = \gamma I,
\end{split}
\end{equation}
is feasible for all $\theta \in [0,\pi/2]$.

Suppose $(C_1,\dots,C_{k+1},D_1,\dots,D_{k+1})$ is a feasible point of SDP \eqref{eq:XMinBallSDPkSimplex}. It is then immediate that
\[
(C'_1,C'_2,C'_3,D'_1,D'_2,D'_3)  := (C_1+C_4+\dots+C_{k+1},C_2,C_3,D_1+D_4+\dots+D_{k+1},D_2,D_3)
\]
is feasible for SDP \eqref{eq:XMinBallSDPtwoSimplexScaled}. On the other hand, if $(C'_1,C'_2,C'_3,D'_1,D'_2,D'_3)$ is feasible for SDP \eqref{eq:XMinBallSDPtwoSimplexScaled}, setting
\[
C_2 = C'_2 \qquad C_3=C'_3 \qquad C_j = C'_1/(k-1) \mathrm{\ for \ } j \neq 2,3,
\]
and similarly
\[
D_2 = D'_2 \qquad D_3=D'_3 \qquad D_j = D'_1/(k-1) \mathrm{\ for \ } j \neq 2,3,
\]
gives that $(C_1,\dots,C_{k+1},D_1,\dots,D_{k+1})$ is a feasible point of SDP \eqref{eq:XMinBallSDPkSimplex}. In particular, with this choice, the only equation of \eqref{eq:XMinBallSDPkSimplex} that is nontrivial to check is the final equality. This can be verified by computing
\[
\frac{1}{1-k} (X_1+X_2) = I
\]
and simplifying the expressions for $X_1$ and $X_2$ given in \eqref{eq:XMinBallSDPkSimplex}. More precisely, for a fixed $\ell \geq 4$ we have
\[
\begin{split}
I=\frac{1}{1-k} (X_1+X_2) &= C_2 + C_3 + \sum_{j \neq 2,3} \frac{2}{1-k} C_j+ D_2 + D_3 + \sum_{j \neq 2,3} \frac{2}{1-k} D_j \\
& = C_2 + C_3 + \sum_{j \neq 2,3} \left(C_j-\frac{1-k-2}{1-k} C_j \right) +D_2 + D_3 + \sum_{j \neq 2,3} \left(D_j-\frac{1-k-2}{1-k} D_j \right)\\
& = C_2 + C_3 + \sum_{j \neq 2,3} \left(C_j-\frac{k+1}{k-1} C_\ell \right) + D_2 + D_3 + \sum_{j \neq 2,3} \left(D_j-\frac{k+1}{k-1} D_\ell \right)\\
& =   C_\ell-(k-1)\frac{k+1}{k-1} C_\ell+ \sum_{j \neq \ell} C_j + D_\ell-(k-1)\frac{k+1}{k-1}  D_\ell + \sum_{j \neq \ell} D_j \\
& = -k C_\ell +  \sum_{j \neq \ell} C_j-k D_\ell +  \sum_{j \neq \ell} D_j,
\end{split}
\]
from which the result follows. 
\end{proof}

\begin{remark}
    Fix $k \geq 3$ and let $\cD_S$ and $\cD_{A(k)}$ and $X_2$ and $X(\theta)$ be as in the statement of Proposition \ref{prop:kVarsSimplexToTwoVar}. Then the tuples $ (I,X_2,X(\theta))$ and $(I,X_2,I,\dots,I,X(\theta))$ are Arveson extreme points of $\cD_S$ and $\cD_{A(k)}$, respectively. However, we warn the reader that the optimum value achieved by the SDPs
\[
\begin{split}
& \max \gamma \ \ \mathrm{s.t.} \\
&  C_1 \oplus \dots \oplus C_{k+1} \oplus D_1 \oplus \dots \oplus D_{k+1} \succeq 0,  \\
& C_1-kC_2+C_3+C_4+\dots+C_{k+1}+D_1-kD_2+D_3+D_4+\dots+D_{k+1} = I, \\
&  C_1+C_2-kC_3+C_4+\dots+C_{k+1}+D_1+D_2-kD_3+D_4+\dots+D_{k+1} =X_2, \\
& C_1+\dots+C_{k+1}-D_1-\dots-D_{k+1} =X(\theta), \\
& C_1+\dots+C_{k+1}+D_1+\dots+D_{k+1} = \gamma I, \\
& -k C_\ell +  \sum_{j \neq \ell} C_j-k D_\ell +  \sum_{j \neq \ell} D_j  = I \qquad \mathrm{for} \qquad \ell=4,\dots,k,
\end{split}
\]
and
\[
\begin{split}
& \max \gamma \ \ \mathrm{s.t.} \\
&  C'_1 \oplus C'_2 \oplus C'_3 \oplus D'_1 \oplus D'_2 \oplus D'_3 \succeq 0,  \\
& C'_1-kC'_2+C'_3+D'_1-kD'_2+D'_3 = I, \\
&  C'_1+C'_2-kC'_3+D'_1+D'_2-kD'_3 =X_2, \\
& C'_1+C'_2+C'_3-D'_1-D'_2-D'_3 =X (\theta), \\
& C'_1+C'_2+C'_3+D'_1+D'_2+D'_3  = \gamma I,
\end{split}
\]
do not necessarily coincide. E.g., when $k=3$ our results show that the optimum achieved by the first SDP for $\theta = \pi/2$ is 3/2. Numerically we find that the optimum value achieved by the second SDP ?? for $\theta = \pi/2$ ?? is approximately 1.456. Thus the choice of extreme point plays an important role in the reduction of the $k$ variable case to the two variable case. 
\end{remark}

\begin{thm}
\label{thm:kSDPFeasible}
As in Theorem \ref{thm:kSimplexPlusLineOptimum}, define $\gamma(k)$ by
\[
\gamma(k)=\gamk{k}.
\]
Then the semidefinite program in \eqref{eq:XMinBallSDPtwoSimplexScaled},
    \begin{equation*}
\begin{split}
&  C_1 \oplus C_2 \oplus C_3 \oplus D_1 \oplus D_2 \oplus D_3 \succeq 0,  \\
& C_1-kC_2+C_3+D_1-kD_2+D_3 =X_1, \\
&  C_1+C_2-kC_3+D_1+D_2-kD_3 =X_2, \\
& C_1+C_2+C_3-D_1-D_2-D_3 =X_{3} (\theta), \\
& C_1+C_2+C_3+D_1+D_2+D_3  = \gamma(k) I,
\end{split}
\end{equation*}
is feasible for all $\theta \in [0,\pi/2]$.
\end{thm}

\begin{proof}
We have already shown that the SDP is feasible when $k=2$, so it is sufficient to consider $k \geq 3$. In this case, set
\[
\alpha(k):=\frac{1}{2 k+4 \sqrt{k+1}+2}
\]
and set
\begin{equation*}
\begin{split}
C_1(\theta) &= \alpha(k) \left(
\begin{array}{cc}
 k-1-2 \cos (\theta) & (k-1) \sin (\theta) \\
 (k-1) \sin (\theta) & k-1+2 \cos (\theta) \\
\end{array}
\right) \\ 
C_2(\theta) &= \alpha(k) \left(
\begin{array}{cc}
 1+\cos (\theta) & \left(\sqrt{k+1}+1\right) \sin (\theta) \\
 \left(\sqrt{k+1}+1\right) \sin (\theta) & \left(k+2 \sqrt{k+1}+2\right) (1-\cos (\theta)) \\
\end{array}
\right) \\
C_3(\theta) &= \alpha(k) 
\left(
\begin{array}{cc}
 \left(k+2 \sqrt{k+1}+2\right) (1+\cos (\theta)) & \left(\sqrt{k+1}+1\right) \sin (\theta) \\
 \left(\sqrt{k+1}+1\right) \sin (\theta) & 1-\cos (\theta) \\
\end{array}
\right)
\\
D_1(\theta) &= \alpha(k)
\left(
\begin{array}{cc}
 k-1+2 \cos (\theta) & -(k-1) \sin (\theta) \\
 -(k-1) \sin (\theta) & k-1-2 \cos (\theta) \\
\end{array}
\right)\\
D_2(\theta) &= \alpha(k)
\left(
\begin{array}{cc}
 1-\cos (\theta) & -\left(\sqrt{k+1}+1\right) \sin (\theta) \\
 -\left(\sqrt{k+1}+1\right) \sin (\theta) & \left(k+2 \sqrt{k+1}+2\right) (1+\cos (\theta)) \\
\end{array}
\right) \\
D_3(\theta) &= \alpha(k)
\left(
\begin{array}{cc}
 \left(k+2 \sqrt{k+1}+2\right) (1-\cos (\theta)) & -\left(\sqrt{k+1}+1\right) \sin
   (\theta) \\
 -\left(\sqrt{k+1}+1\right) \sin (\theta) & 1+\cos (\theta) \\
\end{array}
\right)
\end{split}
\end{equation*}
Then 
\begin{equation*}
\begin{split}
\tr(C_1 (\theta)) = \tr(D_1(\theta)) &= \alpha(k)(2k-2), \\
\tr(C_2 (\theta)) = \tr(D_3(\theta)) &= \alpha(k)\big((1+\cos(\theta)) + (2+k+2\sqrt{k+1})(1-\cos(\theta))\big) \\
\tr(C_3 (\theta)) = \tr(D_2(\theta)) &= \alpha(k)\big((1-\cos(\theta)) + (2+k+2\sqrt{k+1})(1+\cos(\theta))\big) \\
\end{split}
\end{equation*}
and
\begin{equation*}
\begin{split}
\det(C_1(\theta))=\det(D_1(\theta)) = \alpha(k)^2 (k-3)(k+1)\cos(\theta)^2\\
\det(C_2(\theta))=\det(C_3(\theta))=\det(D_2(\theta))=\det(D_3(\theta)) = 0.
\end{split}
\end{equation*}
Since the trace and determinant of all of these symmetric $2 \times 2$ matrices are each nonnegative for all $k \geq 3$ and all $\theta \in [0,2\pi]$, we conclude that all six of these matrices are positive semidefinite when $k \geq 3$. It is a straightforward computation to verify that the remaining equalities hold. 
\end{proof}

\begin{remark}
Note that the matrices used in the proof of Theorem \ref{thm:kSDPFeasible} do not give a feasible point of the SDP when $k=2$. This is due to the $(k-3)$ factor that appears in the determinant of $C_1$ and $D_1$. In particular, these matrices are not PSD when $k=2$ (unless $\theta = \pi/2$). 

We also point out that there are numerous symmetries that appear in the matrices. In particular, the tuple
$(C_1(\theta),C_2(\theta),C_3(\theta),D_1(\theta),D_2(\theta),D_3(\theta))$ has the form
\[
\left(\begin{pmatrix}
    a_1 & a_2 \\
    a_2 & a_3 
\end{pmatrix},
\begin{pmatrix}
    b_1 & b_2 \\
    b_2 & b_3 
\end{pmatrix},
\begin{pmatrix}
    c_1 & b_2 \\
    b_2 & c_3 
\end{pmatrix},
\begin{pmatrix}
    a_3 & -a_2 \\
    -a_2 & a_1 
\end{pmatrix},
\begin{pmatrix}
    c_3 & -b_2 \\
    -b_2 & c_1 
\end{pmatrix},
\begin{pmatrix}
    b_3 & -b_2 \\
    -b_2 & b_1
\end{pmatrix}
\right)
\]
for appropriately chosen values of the $a_i,b_i$ and $c_i$ (which depend on $k$ and $\theta$). 

Lastly we mention that for $k\geq 3$, the matrices above are in fact a feasible point of the SDP in question for all $\theta \in \R$, not just for $\theta \in [0,\pi/2]$. However, the solution we construct in the appendix for the $k=2$ case only works when $\theta \in [0,\pi/2]$. 
\end{remark}

\subsection{Cartesian product of lines}

\begin{prop}\label{prop:extreme-points-four-lines}
     Let $\cD_A$ be the matrix cube in four variables. We consider the optimization problem
    \begin{align*}
        \textrm{maximize} \quad & \quad \lambda_{\max}\left(\sum_{i \in [4]} X_i \otimes Y_i\right) \\
        \textrm{such~that} \quad & \quad X \in \mathcal D_A(2), Y \in \mathcal D_A^\square(2) \,. 
    \end{align*}
    Then, the optimal value is at least $\sqrt{13}/2$.
\end{prop}
\begin{proof}
    Let us consider the matrices 
    \[
X_1 = \begin{pmatrix}
    1 & 0\\ 0 & -1
\end{pmatrix} \quad
X_2 = \begin{pmatrix}
    0 & 1\\ 1 & 0
\end{pmatrix}\quad
X_3 = \begin{pmatrix}
    0 & e^{2\pi i/3}\\ e^{-2\pi i/3} & 0
\end{pmatrix}
\quad
X_4 = \begin{pmatrix}
    0 & e^{-2\pi i/3}\\ e^{2\pi i/3} & 0
\end{pmatrix}
\]
and
\begin{align*}
Y_1 &= \begin{pmatrix}
    \frac{-2}{\sqrt{13}} & 0\\ 0 & \frac{2}{\sqrt{13}} 
\end{pmatrix} \qquad \qquad \qquad \quad \ \
Y_2 = \begin{pmatrix}
    0 & \frac{-3}{2\sqrt{13}} \\ \frac{-3}{2\sqrt{13}}  & 0
\end{pmatrix}\\
Y_3 &= \begin{pmatrix}
    0 & \frac{-3}{2\sqrt{13}}e^{2\pi i/3}\\ \frac{-3}{2\sqrt{13}}e^{-2\pi i/3}\ & 0
\end{pmatrix}
\quad
Y_4 = \begin{pmatrix}
    0 & \frac{-3}{2\sqrt{13}}e^{-2\pi i/3}\\ \frac{-3}{2\sqrt{13}}e^{2\pi i/3} & 0
\end{pmatrix} \,.
\end{align*}
One can verify that indeed $X \in \mathcal D_A(2)$, $Y \in \mathcal D_A^\square(2)$ and 
\begin{equation*}
    \lambda_{\max}\left(\sum_{i \in [4]} X_i \otimes Y_i\right) = \frac{\sqrt{13}}{2} \,. \qedhere
\end{equation*} 
\end{proof}
To give an intuition about how we found these points, we remark that for $X = (X_1,\dots,X_k)$ in the matrix cube, we know that $X$ is irreducible Euclidean extreme $\iff$ $X$ is matrix extreme $\iff$ $X$ is free extreme $\iff$ $X$ is both irreducible and satisfies that $X_i^2 = I$ for all $i=1,\dots,k$. This can be seen from Remark \ref{rem:extreme-points-free-simplex} and Lemma \ref{lem:ExtremeOfCartesian}.

We conjecture that the value $\sqrt{13}/2$ is actual optimal:
\begin{conj}\label{conj:extreme-points-four-lines}
      Let $\cD_A$ be the matrix cube in four variables. We consider the optimization problem
    \begin{align*}
        \textrm{maximize} \quad & \quad \lambda_{\max}\left(\sum_{i \in [4]} X_i \otimes Y_i\right) \\
        \textrm{such~that} \quad & \quad X \in \mathcal D_A(2), Y \in \mathcal D_A^\square(2) \,. 
    \end{align*}
    Then, the optimal value is $\sqrt{13}/2$.
\end{conj}
To support this conjecture, we did a search using a net over the set of free extreme points of the matrix cube at level $2$. This means we could easily generate large collections of evenly spread out extreme points of the matrix cube to compare against.

\section{Noise robustness of measurement incompatibility as a polynomial optimization problem} \label{sec:opti-problems-from-qit}

\subsection{Optimization problem for any number of measurements and outcomes} \label{sec:opti-prob-for-any-number}
We will now use the formulation of measurement incompatibility in terms of free spectrahedra reviewed in Section \ref{sec:incompatibility} to find an optimization problem in non-commutative variables that we can tackle subsequently with methods from non-commutative polynomial optimization.

 \begin{thm}\label{thm:lambda-max-formulation}
     Let $g$, $d \in \mathbb N$ and $k_x \in \mathbb N$ for all $x \in [g]$. Let us consider the optimization problem
     \begin{align*}
         \mathrm{maximize} \quad & \quad \lambda_{\max}\left(\sum_{\substack{i \in [k_x-1]\\x \in [g]}}A_{i|x} \otimes X_{i|x}\right) & \\
         \mathrm{subject~to} \quad & \quad -\frac{k_x}{2} A_{i|x} \preceq I_d & \forall i \in [k_x -1], \forall x \in [g] \, , \\
         \quad & \quad \frac{k_x}{2} \sum_{i \in [k_x -1]} A_{i|x} \preceq I_d & \forall x \in [g] \, , \\
         \quad & \quad \sum_{x \in [g]} \sum_{i \in [k_x - 1]} \left(- \frac{2}{k_x} + 2 \delta_{i, j_x}\right) X_{i|x} \preceq I_d &\forall j \in [k_1] \times \ldots \times [k_g] \, ,\\
         \quad & \quad A_{i|x}, X_{i|x} \in SM_d(\mathbb C) & \forall i \in [k_x-1], \forall x \in [g] \, .
     \end{align*}
     Thus, there are $\sum_{x \in [g]} k_x$ many constraints on the $A_{i|x}$ and $\prod_{x \in [g]}k_x$ many constraints on the $X_{i|x}$. Then, the solution to this optimization problem is $1/s_\C(d,g,(k_1, \ldots, k_g))$.
 \end{thm}

 \begin{proof}
     We want to make use of Proposition \ref{prop:incompatibility-spectrahedra}. From \cite[Definition 4.1 and Proposition 5.1]{bluhm2020compatibility}, we recall that $\mathcal D_{\jewel, \mathbf{k}}(1) \subseteq \mathcal D_A(1)$ is equivalent to $\mathcal D_{\jewel, k_x}(1) \subseteq \mathcal D_{(A_{i|x})_{i \in [k_x-1]}}(1)$ for all $x \in [g]$. This containment can be checked on the extreme points of $\mathcal D_{\jewel, k_x}(1)$. Using the expression for these extreme points in \cite[Lemma 4.3]{bluhm2020compatibility}, we obtain the constraints on $\{A_{i|x}\}_{i \in [k_x-1], x \in [g]}$. By the first point of Proposition \ref{prop:incompatibility-spectrahedra} and setting $A_{i|x} = 2E_{i|x} - \frac{2}{k_i}I_d$, we find that we are optimizing over all POVMs $(E_{i|x})_{i \in [k_x]}$, $x \in [g]$. The constraints on $\{X_{i|x}\}_{i \in [k_x-1], x \in [g]}$ are equivalent to $X \in \mathcal D_{\jewel, k_x}(d)$. Finally, it is readily verified that 
     \begin{equation*}
         \frac{1}{s}=\lambda_{\max}\left(\sum_{\substack{i \in [k_x-1]\\x \in [g]}}A_{i|x} \otimes X_{i|x}\right)
     \end{equation*}
     for $s \neq 0$ is equivalent to finding largest $s$ such that $X \in \mathcal D_{s\cdot A}(d) = \frac{1}{s} \mathcal D_{ A}(d)$. Hence, the optimization problem gives the largest $s$ such that
         $\mathcal D_{\jewel, \mathbf{k}} \subseteq \mathcal D_{s\cdot A}$ for all $A \in SM_d(\mathbb C)^g$ such that $\mathcal D_{\jewel, \mathbf{k}}(1) \subseteq \mathcal D_A(1)$. In order to make the connection to measurement incompatibility, note that using $A_{i|x} = 2E_{i|x} - \frac{2}{k_i}I_d$, we find that $sA_{i|x} = 2E_{i|x}(s) - \frac{2}{k_i}I_d$. Using the second point of Proposition \ref{prop:incompatibility-spectrahedra}, the assertion then follows.
 \end{proof}
If we do not fix the dimension $d$ in the optimization problem, we obtain $s^{\mathbb C}_{\mathrm{di}}(g,(k_1, \ldots, k_g))^{-1}$, where
\begin{equation*}
    s^{\mathbb F}_{\mathrm{di}}(g,(k_1, \ldots, k_g)) := \inf_{d \in \mathbb N}  s_{\mathbb F}(d,g,(k_1, \ldots, k_g)) \, .
\end{equation*}
For this dimension-independent problem, we can use the NPA hierarchy to compute its optimal value. We will do this in Section \ref{sec:line_simplex_num} for the example of one measurement with $2$ outcomes and one measurement with $k$ outcomes.

At the cost of introducing an additional variable, we can replace the maximal eigenvalue in the objective function by a trace:
 \begin{thm} \label{thm:trace-formulation}
     Let $g$, $d \in \mathbb N$ and $k_x \in \mathbb N$ for all $x \in [g]$. Let us consider the optimization problem
     \begin{align*}
         \mathrm{maximize} \quad & \quad \frac{1}{d} \sum_{\substack{i \in [k_x-1]\\x \in [g]}}\tr\left(A_{i|x} X_{i|x}\right) & \\
         \mathrm{subject~to} \quad & \quad -\frac{k_x}{2} A_{i|x} \preceq N & \forall i \in [k_x -1], \forall x \in [g] \, , \\
         \quad & \quad \frac{k_x}{2} \sum_{i \in [k_x -1]} A_{i|x} \preceq N & \forall x \in [g] \, , \\
         \quad & \quad \sum_{x \in [g]} \sum_{i \in [k_x - 1]} \left(- \frac{2}{k_x} + 2 \delta_{i, j_x}\right) X_{i|x} \preceq N&\forall j \in [k_1] \times \ldots \times [k_g] \, ,\\
         \quad & \quad N \succeq 0  \, ,& \\
         \quad & \quad \frac{1}{d} \tr(N) = 1 \, , & \\
         \quad & \quad A_{i|x}, X_{i|x} \in SM_d(\mathbb C) & \forall i \in [k_x-1], \forall x \in [g] \, , \\
         \quad & \quad N \in  SM_d(\mathbb C)\, . & 
     \end{align*}
    Then, the solution to this optimization problem is $1/s_\C(d,g,(k_1, \ldots, k_g))$.
 \end{thm}
 \begin{proof}
     Clearly,
\begin{equation*}
    \lambda I_{d^2} - \sum_{\substack{i \in [k_x-1]\\x \in [g]}}A_{i|x} \otimes X_{i|x} \succeq 0 \iff \psi^\ast  \left(\lambda I -\sum_{\substack{i \in [k_x-1]\\x \in [g]}}A_{i|x} \otimes X_{i|x}\right) \psi \geq 0 \quad \forall \psi \in \mathbb C^{d^2},~\|\psi\| = 1.
\end{equation*}
Let 
\begin{equation*}
    \Omega = \frac{1}{\sqrt{d}}\sum_{i = 1}^d e_i \otimes e_i,
\end{equation*}
where $e_i$ is the standard basis of $\mathbb C^d$. Then, for any  $\psi \in \mathbb C^{d^2}$, $\|\psi\| = 1$, there exists a matrix $V \in  M_d(\mathbb C)$ with $\|V\|_2 = \sqrt{d}$ such that 
\begin{equation*}
    \psi = (V \otimes I_d) \Omega,
\end{equation*}
which follows from the Schmidt decomposition. Moreover, for any $a \in \mathbb R^k$, it holds that if
\begin{equation*}
    \sum_{x \in [g]} \sum_{i \in [k_x-1]} a_{i|x} A_{i|x} \preceq I_d \, ,
\end{equation*}
then 
\begin{equation*}
    V^\ast \left( \sum_{x \in [g]} \sum_{i \in [k_x-1]} a_{i|x} A_{i|x} \right) V \preceq V^\ast V.
\end{equation*}
Thus, we can define $\widetilde A_{i|x} := V^\ast A_{i|x} V $ for all $i \in [k_x]$, $x \in [g]$ and $ N := V^\ast V$. 
Moreover, it holds that 
\begin{equation*}
    \Omega^\ast \left(\sum_{\substack{i \in [k_x-1]\\x \in [g]}} \widetilde A_{i|x} \otimes X_{i|x}\right) \Omega = \frac{1}{d}\sum_{\substack{i \in [k_x-1]\\x \in [g]}} \tr[\widetilde A_{i|x}^T  X_{i|x}]\, .
\end{equation*}
Thus, the solution of the optimization problem in Theorem \ref{thm:trace-formulation} is larger or equal the one in Theorem \ref{thm:lambda-max-formulation}. Conversely, as $N \succeq 0$, we can write $N+\epsilon I$ as $N+\epsilon I=V^* V$ for some invertible $V \in M_d(\mathbb C)$ and some $\epsilon >0$. Then, we can reverse the above steps. We arrive at the conclusion that if $\lambda^\ast_1$ is the optimal value of the optimization problem in Theorem \ref{thm:trace-formulation} and $\lambda^\ast_2$ the optimal value of the optimization problem in Theorem \ref{thm:lambda-max-formulation}, then
\begin{equation*}
    \lambda^\ast_1 \leq (1+\epsilon)\lambda^\ast_2.
\end{equation*} 
As $\epsilon > 0$ was arbitrary, the assertion follows.
 \end{proof}

 Finally, we can give a version that is less good for numerical optimization but useful for mathematical analysis because it is of the form of the problems studied in Section \ref{sec:extreme-points}.

\begin{thm}\label{thm:dual-free-polytopes}
     Let $g$, $d \in \mathbb N$ and $k_x \in \mathbb N$ for all $x \in [g]$. Let us consider the optimization problem
     \begin{align*}
         \mathrm{maximize} \quad & \quad \lambda_{\max}\left(\sum_{\substack{i \in [k_x-1]\\x \in [g]}}A_{i|x} \otimes X_{i|x}\right) & \\
         \mathrm{subject~to} \quad & \quad A \in \mathcal D_{\jewel, \mathbf{k}}^{\square,\mathbb F}(d) \,,\\ 
         \quad & \quad X \in \mathcal D^{\mathbb F}_{\jewel, \mathbf{k}}(d) \,,
     \end{align*}
    Then, the solution to this optimization problem is $1/s_{\mathbb F}(d,g,(k_1, \ldots, k_g))$. 
 \end{thm}
\begin{proof}
    We proceed similarly to the proof of Theorem  \ref{thm:lambda-max-formulation}. By Proposition \ref{prop:incompatibility-spectrahedra} and Corollary \ref{cor:incompatibility-spectrahedra-real}, for $A$ such that $A_{i|x} \in SM_d(\mathbb F)$ for all $i \in [k_x-1]$, $x \in [g]$, $A$ corresponds to a tuple of compatible measurements (using the transformation $A_{i|x} = 2E_{i|x} - \frac{2}{k_i}I_d$) if and only if 
    \begin{equation*}
        \lambda_{\max}\left(\sum_{\substack{i \in [k_x-1]\\x \in [g]}}A_{i|x} \otimes X_{i|x}\right) \leq 1
    \end{equation*}
    for all $X \in \mathcal D^{\mathbb F}_{\jewel, \mathbf{k}}$. Using \cite[Lemma 2.3]{helton2019dilations} for $\mathbb F = \R$ and \cite[Corollary 4.6]{bluhm2018joint} for $\mathbb F = \C$, we can restrict to $X \in \mathcal D^{\mathbb F}_{\jewel, \mathbf{k}}(d)$. To optimize over all $A$ that correspond to tuples of valid measurements, we combine Proposition \ref{prop:incompatibility-spectrahedra} and Corollary \ref{cor:incompatibility-spectrahedra-real} with Proposition \ref{prop:spectrahedron-mcset-equivalent} to find that $A$ corresponds to a tuples of valid measurements described by matrices of $\mathbb F$ if and only if $A \in \mathcal D_{\jewel, \mathbf{k}}^{\square,\mathbb F}(d)$. As for $A_{i|x} = 2E_{i|x} - \frac{2}{k_i}I_d$, it holds that $sA_{i|x} = 2E_{i|x}(s) - \frac{2}{k_i}I_d$, we find that the optimal value is indeed $1/s_{\mathbb F}(d,g,(k_1, \ldots, k_g))$. 
\end{proof}

\begin{remark}
    Considering the form of $D_{\jewel, \mathbf{k}}^\square$, it is easy to see that $D_{\jewel, \mathbf{k}}^\square$ is a Cartesian product product of free simplices in $k_i-1$ variables given in a particular parametrization. See \cite{bluhm2020compatibility} for details.
\end{remark}

\subsection{One dichotomic and one arbitrary measurement}

In this section, we will look more closely at the situation in which we have one measurement with $2$ outcomes and one measurement with $k$ outcomes. In this case, we have the following bounds on the minimum compatibility degree from Proposition \ref{prop:bounds-incompat-degree}:
\begin{equation*}
    \frac{1}{2} < \frac{2+kd}{2(1+kd)} < \frac{1}{2} \left(1+\frac{1}{\sqrt{2k}+1}\right) \leq s_{\C}(d,2,(2,k)) \leq \frac{1}{\sqrt{2}} \,.
\end{equation*}
Thus, in particular
\begin{equation*}
   0.64 \approx \frac{1}{10}(4+\sqrt{6}) \leq  s_{\C}(2,2,(2,3)) \leq \frac{1}{\sqrt{2}} \approx 0.71 \,.
\end{equation*}
Using the results from Section \ref{sec:extreme-points}, we can in fact compute $s_{\R}(d,2,(2,k))$ exactly, i.e., the minimum compatibility degree if we restrict real measurements. 
\begin{thm} \label{thm:results-for-2-plus-k} Let $k \in \mathbb N$. Then, 
    \begin{equation*}
        s_{\R}(2,2,(2,k+1)) = \frac{k-1+\sqrt{1+k}}{2k}=\frac{1}{2}\left(1+\frac{1}{1+\sqrt{k+1}}\right) \, .
    \end{equation*}
A pair of maximally incompatible measurements is
\begin{align*}
    E &= \left( \frac{1}{2}\begin{pmatrix}
        1 & 1 \\ 1 & 1
    \end{pmatrix},  \frac{1}{2}\begin{pmatrix}
        1 & -1 \\ -1 & 1
    \end{pmatrix} \right)\,, \\
    F &= \left( \begin{pmatrix}
        0 & 0 \\ 0 & 1
    \end{pmatrix}, \begin{pmatrix}
        1 & 0 \\ 0 & 0
    \end{pmatrix}, \begin{pmatrix}
        0 & 0 \\ 0 & 0
    \end{pmatrix}, \ldots, \begin{pmatrix}
        0 & 0 \\ 0 & 0
    \end{pmatrix} \right),
\end{align*}
i.e, a measurement in the Hadamard basis and one in the computational basis, padded with outcomes that have zero probability to occur.
\end{thm}
\begin{proof}
    It is enough to realize that the optimization problem in Theorem \ref{thm:dual-free-polytopes} is the same as in Theorem \ref{thm:kSimplexPlusLineOptimum}, up to an invertible linear transformation of the polytope, which does not affect the optimal value \cite[Lemma 3.1]{passer2018minimal}. The result then follows from Theorem \ref{thm:kSimplexPlusLineOptimum} and an application of the invertible linear transformation to the $X_i$ which achieve the maximum value.
\end{proof}
In the same way, Conjecture \ref{conj:kSimplexPlusLineOptimum} would imply that the theorem still holds if we replace the $d=2$ with any dimension $d \geq 2$ and if we allow complex measurements as well.
\begin{conj}\label{conj:conj-for-2-plus-k}
   Let $k \in \mathbb N$. Then,
    \begin{equation*}
        s_\C(d,2,(2,k+1)) = \frac{k-1+\sqrt{1+k}}{2k}=\frac{1}{2}\left(1+\frac{1}{1+\sqrt{k+1}}\right) \, .
    \end{equation*}
A pair of maximally incompatible measurements is
\begin{align*}
    E &= \left( \frac{1}{2}\begin{pmatrix}
        1 & 1 \\ 1 & 1
    \end{pmatrix},  \frac{1}{2}\begin{pmatrix}
        1 & -1 \\ -1 & 1
    \end{pmatrix} \right)\,, \\
    F &= \left( \begin{pmatrix}
        0 & 0 \\ 0 & 1
    \end{pmatrix}, \begin{pmatrix}
        1 & 0 \\ 0 & 0
    \end{pmatrix}, \begin{pmatrix}
        0 & 0 \\ 0 & 0
    \end{pmatrix}, \ldots, \begin{pmatrix}
        0 & 0 \\ 0 & 0
    \end{pmatrix} \right) \,.
\end{align*} 
\end{conj}
\begin{remark}
It is interesting to note that the number $\frac{1}{2}\left(1+\frac{1}{1+\sqrt{k+1}}\right)$ is also the minimal noise level to make two mutually unbiased bases (MUBs) in dimension $k+1$ compatible \cite{Designolle2019}. Whether there is a deeper reason for this or whether it is merely a coincidence remains at present unclear, since the situation we consider seems unrelated to the one of $2$ MUBs.
\end{remark}

\subsection{Four dichotomic qubit measurements} 

In this section, we will consider the minimum compatibility degree of $4$ dichotomic qubit measurements. From Proposition \ref{prop:bounds-incompat-degree}, we know that 
\begin{equation*}
    \frac{1}{2} \leq s_\C(2,4) \leq  1/\sqrt{3} \approx 0.58.
\end{equation*}
Moreover, there are stronger numerical bounds: In \cite[Table I]{bavaresco2017most}, the authors find that $s_\C(2,4) \lessapprox 0.5547$ (they have more upper bounds for other values of $g$), which agrees with the numerical bounds two of the present authors later found in \cite{bluhm2022incompatibility}. In \cite[Table I]{bavaresco2017most}, the authors also give a lower bound for projective measurements of $0.5437$. However, this bound would only translate into a lower bound on $s_\C(2,4)$ if we could prove that the most incompatible measurements are projective, which is not known for the noise model we use (see also \cite{Designolle2019}). Note however that for the cases $s_\C(2,2)$ and $s_\C(2,3)$ where we know the minimum compatibility degree to be $1/\sqrt{g}$, projective measurements are indeed the most incompatible ones. The upper bounds in \cite{bavaresco2017most} were found using a see-saw algorithm (also known as alternate directions) and the lower bounds by outer polytope approximation. The authors in \cite{bavaresco2017most} consider a slightly different noise model, which coincides with ours for rank $1$ projective measurements.

For this situation, the optimization problem in Theorem \ref{thm:lambda-max-formulation} has the form
\begin{align}
    \mathrm{maximize} \quad & \lambda_{\max}\left(\sum_{i = 1}^4 A_i \otimes X_i\right) \nonumber\\
    \mathrm{subject~to} \quad
    & A_1 \preceq I_2, \quad-A_1 \preceq I_2,\nonumber \\
    & A_2 \preceq I_2, \quad-A_2 \preceq I_2,\nonumber \\
    & A_3 \preceq I_2, \quad-A_3 \preceq I_2,\label{eq:lambda-max-four-qubits} \\
    & A_4 \preceq I_2, \quad-A_4 \preceq I_2,\nonumber \\
    &\sum_{i=1}^4 \epsilon_i X_i \preceq I_2, \qquad \forall \epsilon \in \{\pm 1\}^4\nonumber\\
    & B_i,~X_i \in \mathcal SM_2(\mathbb C) \qquad \forall i \in \{1,2,3, 4\}\nonumber
\end{align}
Its optimal value is $1/s_\C(2,4)$.
\begin{lem}\label{lem:squared-equal-identity}
    In the above optimization problem, we can restrict to $A_i^2=I_2$ for all $i \in [4]$.
\end{lem}
\begin{proof}
Theorem \ref{thm:MaxEigIsMinBallInclusionConst} tells us that it is enough to consider matrix extreme points of $\mathcal W^{\max}([-1,1]^4)$, as the constraints correspond to $(A_1, \ldots, A_4) \in \mathcal W^{\max}([-1,1]^4)=\bigtimes_{i \in [4]} \mathcal W^{\max}([-1,1])$ (see also Theorem \ref{thm:dual-free-polytopes}). We can then combine Lemma \ref{lem:ExtremeOfCartesian} with Remark \ref{rem:extreme-points-free-simplex}, as the matrix and free extreme points coincide for the matrix cube (see \cite[Proposition 7.1]{evert2018extreme}, using the fact that matrix extreme points are irreducible).
\end{proof}

\begin{thm} \label{thm:results-for-four-qubits}
The following upper bound on the minimum compatibility degree of four dichotomic qubit measurements holds:
\begin{equation*}
    s_\C(2,4) \leq \frac{2}{\sqrt{13}} \approx 0.5547\,.
\end{equation*}
The upper bound is achieved for the measurements
\begin{align}
    E &= \left(\begin{pmatrix}
        1 & 0 \\ 0 & 0
    \end{pmatrix}, \begin{pmatrix}
        0 & 0 \\ 0 & 1
    \end{pmatrix}\right) \nonumber \\
    F &= \left(\frac{1}{2} \begin{pmatrix}
        1 & 1 \\ 1 & 1
    \end{pmatrix}, \frac{1}{2} \begin{pmatrix}
        1 & -1 \\ -1 & 1
    \end{pmatrix}\right) \label{eq:4qubitsexample} \\
    \quad   G &=  \left(\frac{1}{2}\begin{pmatrix}
    1 & e^{2\pi i/3}\\ e^{-2\pi i/3} & 1
\end{pmatrix},  \frac{1}{2}\begin{pmatrix}
    1 & e^{-\pi i/3}\\ e^{\pi i/3} & 1
\end{pmatrix}\right) \nonumber\\
    \quad   H &= \left( \frac{1}{2}\begin{pmatrix}
    1 & e^{-2\pi i/3}\\ e^{2\pi i/3} & 1
\end{pmatrix}, \frac{1}{2}\begin{pmatrix}
    1 & e^{\pi i/3}\\ e^{-\pi i/3} & 1
\end{pmatrix}\right) \, .\nonumber
\end{align}
\end{thm}

\begin{proof}
    This follows directly from Proposition \ref{prop:extreme-points-four-lines} and an application of the invertible linear transformation to the $X_i$ which achieve the maximum value.
\end{proof}

If one represents the above measurements in terms of their Bloch vectors, we get three equiangular lines in a plane and a fourth line perpendicular to them, see Figure \ref{fig:plot-4-meas-Bloch}. If Conjecture \ref{conj:extreme-points-four-lines} was true, we would have the following:
\begin{conj}
\label{conj:fourlines}
    We conjecture that 
    \begin{equation*}
    s_\C(2,4) = \frac{2}{\sqrt{13}}\,.
\end{equation*}
\end{conj}
That would mean that the value computed in \cite{bavaresco2017most} is optimal up to numerical precision. 

\begin{figure}
    \centering
    \includegraphics[width=0.25\linewidth]{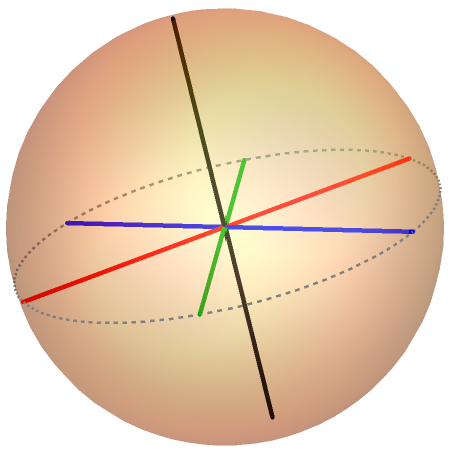}
    \caption{The conjectured most incompatible 4-tuple of dichotomic qubit measurements in Bloch representation. It consists of 3 equiangular antipodal segments on a plane and a fourth one perpendicular to the first three.}
    \label{fig:plot-4-meas-Bloch}
\end{figure}

\medskip

We will now simplify the optimization problem \eqref{eq:lambda-max-four-qubits}. If the maximum is attained on some pure state $\psi \in \mathbb C^d \otimes \mathbb C^d$, considering its partial trace $\rho$, we can write (considering $g$ dichotomic measurements in dimension $d$ instead of $4$ qubit measurements)
\begin{align*}
    \lambda_{\max}\left(\sum_{i = 1}^g A_i \otimes X_i\right) &=  \psi^\ast \left(\sum_{i = 1}^g A_i \otimes X_i \right) \psi  =  d\, \Omega^\ast\left(\rho^{1/2} \otimes I_d\left(\sum_{i = 1}^g A_i \otimes X_i\right)\rho^{1/2} \otimes I_d\right) \Omega \\
    &= \sum_{i=1}^g \tr(\rho^{1/2}A_i\rho^{1/2}X_i^\top).
\end{align*}
Here, $\Omega=1/\sqrt{d}\sum_{i \in [d]}e_i \otimes e_i$ and $e_i$ are the standard unit basis vectors. Optimizing over each $A_i$ gives
$$\max_{\|A_i\|_{\infty} \leq 1} \tr(\rho^{1/2}A_i\rho^{1/2}X_i^\top) = \|\rho^{1/2}X_i^\top\rho^{1/2} \|_1,$$
which is achieved for $A_i$ commuting with $\rho^{1/2}X_i^\top\rho^{1/2}$ and having as eigenvalues the signs of the corresponding eigenvalues of $\rho^{1/2}X_i^\top\rho^{1/2}$. This explains in particular the commutation relation $[A_i, X_i] = 0$ noticed in the case where the state achieving $\lambda_{\max}$ was the maximally entangled state $\Omega$, which corresponds to $\rho = I/d$. We have thus shown that

\begin{align*}
   1/s_\C(d,g)= \mathrm{maximize} \quad & \sum_{i = 1}^g \|\rho^{1/2}X_i\rho^{1/2}\|_1 \\
    \mathrm{subject~to} \quad
    &\sum_{i=1}^g \epsilon_i X_i \preceq I_d \qquad \forall \epsilon \in \{\pm 1\}^g\\
    &\rho \succeq 0, \quad \tr \rho =1\\
    & X_i \in SM_d(\mathbb C) \qquad \forall i \in [g]\\
    & \rho \in SM_d(\mathbb C).
\end{align*}

Since the Schatten $1$-norm $\|\cdot\|_1$ and the constraints on $X_i$ are unitarily invariant, we can diagonalize $\rho$ and further simplify our problem: 

\begin{align*}
   1/s_\C(d,g)= \mathrm{maximize} \quad & \sum_{i = 1}^g \|\operatorname{diag}(\sqrt{p}) X_i\operatorname{diag}(\sqrt{p})\|_1 \\
    \mathrm{subject~to} \quad
    &\sum_{i=1}^g \epsilon_i X_i \preceq I_d \qquad \forall \epsilon \in \{\pm 1\}^g\\
    &p_i \geq 0, \quad \sum_{i=1}^g p_i =1\\
    & X_i \in SM_d(\mathbb C) \qquad \forall i \in [g].
\end{align*}

\begin{remark}
    Numerical results indicate that we cannot assume in general that the optimal $p$ is flat (i.e.~$p_i = 1/d$). For example, if $g=2$ and $d=3$, numerics indicate that $1/s_\C(2,3) = \sqrt 2$, with $p=(1/2, 1/2, 0)$ and 
    $$X_1 = \frac{1}{\sqrt{2}} \sigma_X \oplus 0 \text{ and } X_2 = \frac{1}{\sqrt{2}} \sigma_Z \oplus 0.$$
\end{remark}

\begin{prop} In the case of qubits ($d=2$), we further have

\begin{align} 
   1/s_\C(2,g)= \mathrm{maximize} \quad & \sum_{i = 1}^g \|x^{(i)}\|  \nonumber \\
    \mathrm{subject~to} \quad
    &\Big\|\sum_{i=1}^g \epsilon_i x^{(i)} \Big \| \leq 1 \qquad \forall \epsilon \in \{\pm 1\}^g  \label{eq:3D-opt} \\
    & x^{(i)} \in \mathbb R^3\qquad \forall i \in [g]. \nonumber
\end{align}

Equivalently, 
    \begin{align*}
   1/s_\C(2,g)= \mathrm{maximize} \quad & \frac 1 2\sum_{i = 1}^g \|X_i\|_1\\ 
    \mathrm{subject~to} \quad
    &\sum_{i=1}^g \epsilon_i X_i \preceq I_2 \qquad \forall \epsilon \in \{\pm 1\}^g\\
    & \tr[X_i]=0  \qquad \forall i \in [g]\\
    & X_i \in SM_2(\mathbb C) \qquad \forall i \in [g].
\end{align*}
\end{prop}
\begin{proof}
The first reformulation follows from \cite[Theorem 8.5]{bluhm2022incompatibility}. Indeed, note that the value $s_\C(2,g)$ is the largest element of the form $(s,s,\ldots,s)$ of the set of (generalized) inclusion constants $\Pi(g,\mathcal M_2^{\textrm{sa}},\mathrm{PSD}_2)$ from \cite[Section 7]{bluhm2022incompatibility}. In \cite[Theorem 8.5]{bluhm2022incompatibility} it is shown that, for centrally symmetric state spaces, this inclusion constant set is equal to a ``traceless version'' $\Pi'$; this is indeed the case for qubits, where the state space is the unit ball of $\R^3$. Hence, using the definition of the set $\Pi'$ and the formulas for the injective (resp.~projective) tensor norms from \cite[Section 8]{bluhm2022incompatibility}, we recover the formulation with $g$ vectors from $\R^3$ as variables. 

The matrix formulation follows by expressing the matrices $X_i$ in the Pauli basis: 
$$X_i  = x^{(i)}_1 \sigma_X + x^{(i)}_2 \sigma_Y + x^{(i)}_3 \sigma_Z,$$
using the vectors $x^{(1)}, \ldots, x^{(g)} \in \mathbb R^3$. 
\end{proof}

Note that the optimization problem with $g$ vectors in $\R^3$ can be rephrased in terms of the Gram matrix $G \in M_g(\mathbb R)$ of the vectors $x^{(i)}$
$$G_{ij} = \langle x^{(i)},x^{(j)} \rangle \qquad i,j \in [g]$$
as follows (note the rank condition on the matrix $G$):
\begin{align*}
   1/s_\C(2,g)= \mathrm{maximize} \quad & \sum_{i = 1}^g \sqrt{G_{ii}}  \\
    \mathrm{subject~to} \quad
    &\epsilon^\ast G \epsilon \leq 1 \qquad \forall \epsilon \in \{\pm 1\}^g\\
    & G \geq 0\\
    & \rank G \leq 3\\
    & G \in SM_g(\R).
\end{align*}

\medskip

If we would rather optimize a trace as in Theorem \ref{thm:trace-formulation}, this becomes
\begin{align*}
    \mathrm{maximize} \quad & \frac{1}{2} \sum_{i=1}^4\operatorname{tr}[A_i X_i]\\
    \mathrm{subject~to} \quad
    & N \succeq 0, \quad \operatorname{tr}[N] = 2, \\
    & A_1 \preceq N, \quad-A_1 \preceq N, \\
    & A_2 \preceq N, \quad-A_2 \preceq N, \\
    & A_3 \preceq N, \quad-A_3 \preceq N, \\
    & A_4 \preceq N, \quad-A_4 \preceq N, \\
    &\sum_{i=1}^4 \epsilon_i X_i \preceq I_2, \qquad \forall \epsilon \in \{\pm 1\}^4\\
    & N,~A_i,~X_i \in \mathcal SM_2(\mathbb C) \qquad \forall i \in \{1,2,3, 4\}
\end{align*}

Instead of solving this directly, we can expand the variables in a basis of Hermitian operators. A convenient choice are the Pauli matrices together with the identity matrix. We obtain
\begin{align*}
    N &= n_0 I_2 + n_1 \sigma_X + n_2 \sigma_Y + n_3 \sigma_Z \,,\\
    A_i &= a^{(i)}_0 I_2 + a^{(i)}_1 \sigma_X + a^{(i)}_2 \sigma_Y + a^{(i)}_3 \sigma_Z \,, \\
    X_i &= x^{(i)}_0 I_2 + x^{(i)}_1 \sigma_X + x^{(i)}_2 \sigma_Y + x^{(i)}_3 \sigma_Z \,.
\end{align*}
Then, we get a polynomial optimization problem in commuting variables:
\begin{align}
    \mathrm{maximize} \quad & \sum_{i=1}^4 \sum_{j = 0}^3 a^{(i)}_j x^{(i)}_j \nonumber\\
    \mathrm{subject~to} \quad
    &1 - (n_1)^2 - (n_2)^2 - (n_3)^2 \geq 0, \nonumber\\
    &(1-a^{(i)}_0)^2 - (n_1 - a^{(i)}_1)^2 - (n_2 - a^{(i)}_2)^2 - (n_3 -  a^{(i)}_3)^2 \geq 0, \qquad \forall i \in \{1,2,3,4\}\nonumber\\
    &(1+a^{(i)}_0)^2 - (n_1 + a^{(i)}_1)^2 - (n_2 + a^{(i)}_2)^2 - (n_3 +  a^{(i)}_3)^2 \geq 0, \qquad \forall i \in \{1,2,3,4\} \label{eq:four-qubits-Bloch}\\
    & 1 \geq a_0^{(i)} \geq -1, \qquad \forall i \in \{1,2,3,4\}\nonumber\\
    &\left(1-\sum_{i=1}^4 \epsilon_i x^{(i)}_0\right)^2 \geq \sum_{j=1}^3 \left(\sum_{i=1}^4 \epsilon_i x^{(i)}_j\right)^2, \qquad \forall \epsilon \in \{\pm 1\}^4\nonumber\\
    & 1 \geq \sum_{i=1}^4 \epsilon_i x^{(i)}_0, \qquad \forall \epsilon \in \{\pm 1\}^4\nonumber\\
    & a^{(i)}_j,~x^{(i)}_j \in \mathbb R  \qquad \forall i, j \nonumber
\end{align} 

This form can be further simplified if we know more about the optimizers. Starting from the optimization problem in equation \eqref{eq:3D-opt}, we can use the fact that $\|x\| = \sup_{\|a\| \leq 1} \langle a, x \rangle$ to find a simpler polynomial optimization problem in commuting variables:
\begin{align}
    \mathrm{maximize} \quad & \sum_{i=1}^g \sum_{j = 1}^3 a^{(i)}_j x^{(i)}_j \nonumber\\
    \mathrm{subject~to} \quad
    &1 - \sum_{j=1}^3 (a^{(i)}_j)^2 \geq 0, \qquad \forall i \in \{1, \ldots ,g\}\nonumber\\
    &1-\sum_{j=1}^3 \left(\sum_{i=1}^g \epsilon_i x^{(i)}_j\right)^2 \geq 0, \qquad \forall \epsilon \in \{\pm 1\}^g \label{eq:g-qubits-Bloch-simpler}\\
    & a^{(i)}_j,~x^{(i)}_j \in \mathbb R  \qquad \forall i \in \{1, \ldots, g\}, j \in \{1, 2, 3\} \nonumber
\end{align}
This formulation will be useful for numerical experiments, as we can compute bounds with the Lasserre hierarchy (see Section \ref{sec:four-qubits-numerics}).

\section{Concluding numerical experiments and perspectives} 
\label{sec:conclusion}
In this section, we will present numerical results that support our Conjectures \ref{conj:conj-for-2-plus-k} and \ref{conj:fourlines} , obtained vie techniques to solve the polynomial optimization problems in the previous section.

\subsection{Numerics for one dichotomic and one arbitrary measurement} \label{sec:line_simplex_num}
Using Theorem \ref{thm:lambda-max-formulation}, we can compute $1/s_\C(d,s,(2,k))$ with the following optimization problem:
\begin{align}
    \mathrm{maximize} \quad & \lambda_{\max}\left(A_{1|1} \otimes X_{1|1}+\sum_{i = 1}^{k-1} A_{i|2} \otimes X_{i|2}\right) \nonumber\\
    \mathrm{subject~to} \quad
    & A_{1|1} \preceq I_d, \quad-A_{1|1} \preceq I_d, \nonumber\\
    & -\frac{k}{2} A_{i|2} \leq I_d,\quad \forall i \in [k-1] \nonumber\\
    & \frac{k}{2} \sum_{j = 1}^{k-1} A_{j|2} \preceq I_d, \label{eq:2pluskoptimax}\\
    &\pm X_{1|1} - \frac{2}{k} \sum_{j = 1}^{k-1}  X_{j|2} + 2 X_{i|2} \preceq I_d, \quad \forall i \in [k-1]\nonumber\\
    &\pm X_{1|1} - \frac{2}{k} \sum_{j = 1}^{k-1}  X_{j|2} \preceq I_d, \nonumber\\
    & A_{1|1}, X_{1|1}, A_{i|2},~X_{i|2} \in M_d(\mathbb C) \qquad \forall i \in [k-1].\nonumber
\end{align}
Unlike the problem for $4$ dichotomic measurements, this problem is still interesting if we optimize over $d$ as well, thus computing $s^\mathbb C_{\mathrm{di}}(g,(2,k))$. 
The advantage of not fixing the dimension is that we can compute lower bounds of $s^\mathbb C_{\mathrm{di}}(g,(2,k))$ using the NPA hierarchy of semidefinite programs to solve the following non-commutative eigenvalue maximization problem:
\begin{align}
    \mathrm{maximize}_{\mathcal{H}} \quad & \lambda_{\max}\left(A_{1|1}  X_{1|1}+\sum_{i = 1}^{k-1} A_{i|2} X_{i|2}\right) \nonumber\\
    \mathrm{subject~to} \quad
    & A_{1|1} \preceq I, \quad-A_{1|1} \preceq I, \nonumber\\
    & -\frac{k}{2} A_{i|2} \leq I,\quad \forall i \in [k-1] \nonumber\\
    & \frac{k}{2} \sum_{j = 1}^{k-1} A_{j|2} \preceq I, \label{eq:2pluskoptimaxqc}\\
    &\pm X_{1|1} - \frac{2}{k} \sum_{j = 1}^{k-1}  X_{j|2} + 2 X_{i|2} \preceq I, \quad \forall i \in [k-1]\nonumber\\
    &\pm X_{1|1} - \frac{2}{k} \sum_{j = 1}^{k-1}  X_{j|2} \preceq I, \nonumber\\
    & A_{1|1}, X_{1|1}, A_{i|2},~X_{i|2} \in \mathcal{B}(\mathcal H)  \qquad \forall i \in [k-1],\nonumber
\end{align}
where $I$ is the identity element of $\mathcal{B}(\mathcal H)$.

Optimal values of all studied semidefinite programs were computed on one of the CALMIP\footnote{\url{https://www.calmip.univ-toulouse.fr/}} servers with 2 Intel Xeon Gold 6140 CPUs @ 2.30 GHz, with 72 logical processors (18 cores per socket, 2 threads per core) and 800GB of RAM. 
All the instances have been modeled thanks to the Julia library \texttt{NCTSSOS}\footnote{\url{https://github.com/wangjie212/NCTSSOS}} \cite{wang2021exploiting}, and their corresponding semidefinite relaxations have been solved with Mosek \cite{andersen2000mosek}. 
The results reported in Table \ref{tab:2pluskknown} can be reproduced with online scripts\footnote{\url{https://wangjie212.github.io/NCTSSOS/dev/inclusion/}}. 
We now provide more detailed explanation for these results. 

For $k=3$, we could rely on the theory of extreme points to add the following equality constraints to the above problem \eqref{eq:2pluskoptimaxqc}: $A_{1|1}^2=I, [A_{1|2}, A_{2|2}]=0, (A_{i|2}+2/3)(A_{i|2}-4/3)=0$, for all $i=1,2$. 
At level 4 of the NPA hierarchy, we found a numerical bound of $0.683012 \lesssim (1+\sqrt{3})/4$ and we could successfully extract a tuple of finite-dimensional optimizers (of size $12$) thanks to the linear algebra procedure implementing the Gelfand-Naimark-Segal (GNS) construction; see \cite[Algorithm~1.2]{burgdorf2016optimization} for more details. 
The extracted solution reads as follows:
\begin{align*}
X_{1|1} & = (\sqrt{3}-1) \begin{pmatrix}
\begin{pmatrix}
0 & 1   \\
1 & 0
\end{pmatrix} \otimes I_2 &  \\
 & I_2 \otimes \begin{pmatrix}
0 & 1   \\
1 & 0
\end{pmatrix} \otimes \operatorname{diag}(1, -1)
\end{pmatrix},\\
X_{1|2} & = \frac{\sqrt{3}}{2}  \begin{pmatrix}
-  \operatorname{diag}(1, -1) \otimes I_2 &  \\
 & \operatorname{diag}(1, -1) \otimes \operatorname{diag}(\sqrt{3}-2, \sqrt{3}- 1) \otimes I_2
\end{pmatrix},\\
X_{2|2} & = \frac{\sqrt{3}}{2} \begin{pmatrix}
\operatorname{diag}(2-\sqrt{3}, 1-\sqrt{3}) \otimes I_2 &  &  \\    
 & - \operatorname{diag}(1, -1) \otimes I_2 &  \\
 &  & \operatorname{diag}(1-\sqrt{3}, \sqrt{3}-2) \otimes I_2
\end{pmatrix},\\
A_{1|1} & = - \begin{pmatrix}
 I_2 \otimes \begin{pmatrix}
0 & 1   \\
1 & 0
\end{pmatrix} & \\
 &  I_2 \otimes \operatorname{diag}(1, -1) \otimes \begin{pmatrix}
0 & 1   \\
1 & 0
\end{pmatrix}
\end{pmatrix}, \\
A_{1|2} & = - \frac{2}{3} \begin{pmatrix}
I_2 \otimes \operatorname{diag}(1, -2) & & \\
& I_4 & \\
& & I_2 \otimes \operatorname{diag}(-2, 1)
\end{pmatrix}, \\
A_{2|2} & = - \frac{2}{3} \begin{pmatrix}
I_4 & \\
 & I_4 \otimes \operatorname{diag}(1, -2)
\end{pmatrix}. 
\end{align*}

For $k=4$, we added the following equality constraints to the above problem \eqref{eq:2pluskoptimaxqc}: $A_{1|1}^2=I, [A_{1|2}, A_{2|2}]= [A_{2|2}, A_{3|2}]= [A_{1|2}, A_{3|2}] = 0, (A_{i|2}+1/2)(A_{i|2}-3/2)=0$, for all $i=1,2,3$. 
At levels 3 and 4 of the NPA hierarchy, we found a numerical bound of $0.666666 \lesssim 2/3$ but unfortunately we could not extract a tuple of optimizers. 

For $k=5$, we similarly found a numerical bound of $0.654508 \lesssim (3+\sqrt{5})/8$ at level 3 of the NPA hierarchy without being able to perform extraction. 

These numerical results support Conjecture \ref{conj:conj-for-2-plus-k}.

\begin{table}[]
    \centering
\begin{tabular}{|c|c|c|c|}
\hline
    $k$& lower bound & upper bound & comment  \\ \hline
    2 & $1/\sqrt{2}$ & $1/\sqrt{2}$ & follows, e.g., from \cite{bluhm2018joint} \\ \hline
    3 & $0.683012$ & $(1+\sqrt{3})/4$ & NPA level 4 \\ \hline
    4 & $0.666666$ & $2/3$ & NPA level 4\\ \hline
    5 & $ 0.654508$ & $(3+\sqrt{5})/8$ & NPA level 3\\ \hline
\end{tabular}
    \caption{Upper and lower bounds for $s^\C_{\mathrm{di}}(2,(2,k))$. The upper bounds are the ones from Theorem \ref{thm:results-for-2-plus-k}. The lower bounds come from applying the NPA hierarchy to \eqref{eq:2pluskoptimaxqc}, thus they hold for any $d$. 
    }
    \label{tab:2pluskknown}
\end{table}

\subsection{Numerics for four dichotomic qubit measurements} \label{sec:four-qubits-numerics}

For four dichotomic qubit measurements, we consider the optimization problem \eqref{eq:g-qubits-Bloch-simpler}. 
To obtain lower bounds on $s_\C(4,2)$, we can run the Lasserre hierarchy recalled in Section \ref{sec:lasserre}. 
The results below can be reproduced with online scripts\footnote{\url{https://wangjie212.github.io/TSSOS/dev/inclusion/}}. 

The upper bound on $1/s_\C(2,4)$ of $1.8029 \gtrsim \sqrt{13}/2 \simeq 1.8028$ is obtained with the \texttt{TSSOS}\footnote{\url{https://github.com/wangjie212/TSSOS}}\cite{wang2021tssos}   solver at the 3rd relaxation order after 2 rounds of exploiting term sparsity and $10^5$ iterations with the first-order SDP solver COSMO \cite{garstka2021cosmo}. 
Even if it is currently unknown how to extract minimizers while using term sparsity, this upper bound provides us with a strong numerical evidence that supports Conjecture \ref{conj:fourlines}. 

\subsection{Discussion and open questions} 
In this work, we have studied extreme points of Cartesian products of free simplices to compute inclusion constants for free spectrahedra relevant to computing the noise robustness of measurement incompatibility in quantum information theory. 

For one dichotomic and one measurement with $k+1$ outcomes, we have proved that
\begin{equation*}
        s_{\R}(2,2,(2,k+1)) = \frac{k+\sqrt{1+k}}{2k}=\frac{1}{2}\left(1+\frac{1}{1+\sqrt{k+1}}\right) \, .
    \end{equation*}
An important open question is to prove Conjecture \ref{conj:conj-for-2-plus-k}, stating that 
\begin{equation*}
    s_{\R}(2,2,(2,k+1))= s_{\C}(2,2,(2,k+1))
\end{equation*}
for all $k \in \mathbb N$. The conjecture is supported by numerical evidence we have gathered using techniques from non-commutative polynomial optimization.

 To make progress towards this conjecture, one way would be to modify the SDP hierarchies that we have been using such that we can find a flat extension at reasonably low level. This would give us a certificate that we have indeed found the optimal value for the problem instead of merely computing a lower bound. 

 To prove the conjecture analytically, we would need a better understanding of the extreme points of Cartesian products of complex free simplices. Theorem \ref{thm:SimpleXIntervalComplexMatEx} indicates however that this is a difficult problem. Moreover, it would be very interesting to understand the extreme points of direct sums of free simplices, which seems to require new techniques.

 For the case of four dichotomic qubit measurements we have proved that
 \begin{equation*}
     s_\C(2,4) \leq \frac{2}{\sqrt{13}}
 \end{equation*}
 and conjectured that equality holds (Conjecture \ref{conj:fourlines}). 
 To prove the conjecture, we would likely need both a deeper understanding of the optimizers in the associated optimization problem, enabling us to construct a flat extension at a reasonably low level, and a method to generate sums of squares certificates for the optimal bounds.
One promising avenue for further investigation is to leverage the recently developed high-precision SDP solver \cite{leijenhorst2024solving}, which includes a rounding feature to produce sums of squares decompositions with coefficients in a field extension of the rationals.
\medskip\\
\textbf{Acknowledgments}. 
\noindent
 A.B. was supported by the French National Research Agency in the framework of the “France 2030” program (ANR-11-LABX-0025-01) for the LabEx PERSYVAL and by the ANR project PraQPV, grant number ANR-24-CE47-3023. I.N.~was supported by the ANR projects \href{https://esquisses.math.cnrs.fr/}{ESQuisses} grant number ANR-20-CE47-0014-01 and \href{https://www.ceremade.dauphine.fr/dokuwiki/anr-tagada:start}{TAGADA} grant number ANR-25-CE40-5672. This work was supported by the European Union’s HORIZON–MSCA-2023-DN-JD programme under under the Horizon Europe (HORIZON) Marie Skłodowska-Curie Actions, grant agreement 101120296 (TENORS), the project COMPUTE, funded within the QuantERA II Programme that has received funding from the EU's H2020 research and innovation programme under the GA No 101017733 {\normalsize\euflag}.
This work was partially performed using HPC resources from CALMIP (Grant 2023-P23035). 

\bibliographystyle{alpha}
\bibliography{lit}

\newcommand{\etalchar}[1]{$^{#1}$}
\begin{thebibliography}{DDOSS17}

\bibitem[AA00]{andersen2000mosek}
Erling~D. Andersen and Knud~D. Andersen.
\newblock The {MOSEK} interior point optimizer for linear programming: an
  implementation of the homogeneous algorithm.
\newblock In {\em High Performance Optimization}, volume~33 of {\em Applied
  Optimization}, pages 197--232. Springer, 2000.

\bibitem[AKG{\etalchar{+}}23]{araujo2311first}
Mateus Ara{\'u}jo, Igor Klep, Andrew J.~P. Garner, Tamás V{\'e}rtesi, and
  Miguel Navascu{\'e}s.
\newblock First-order optimality conditions for non-commutative optimization
  problems.
\newblock {\em arXiv preprint arXiv:2311.18707}, 2023.

\bibitem[BCP{\etalchar{+}}14]{Brunner2014}
Nicolas {Brunner}, Daniel {Cavalcanti}, Stefano {Pironio}, Valerio {Scarani},
  and Stephanie {Wehner}.
\newblock {Bell nonlocality}.
\newblock {\em Reviews of Modern Physics}, 86:419--478, 2014.

\bibitem[BJN22]{bluhm2022incompatibility}
Andreas Bluhm, Anna Jen{\v{c}}ov{\'a}, and Ion Nechita.
\newblock Incompatibility in general probabilistic theories, generalized
  spectrahedra, and tensor norms.
\newblock {\em Communications in Mathematical Physics}, 393(3):1125--1198,
  2022.

\bibitem[BKP16]{burgdorf2016optimization}
Sabine Burgdorf, Igor Klep, and Janez Povh.
\newblock {\em Optimization of polynomials in non-commuting variables}.
\newblock SpringerBriefs in Mathematics. Springer, 2016.

\bibitem[BLN25]{bluhm2025random}
Andreas Bluhm, C{\'e}cilia Lancien, and Ion Nechita.
\newblock Random measurements are almost maximally incompatible.
\newblock {\em arXiv preprint arXiv:2507.20600}, 2025.

\bibitem[BM25a]{blecher2025real}
David~P. Blecher and Caleb McClure.
\newblock Real noncommutative convexity {I}.
\newblock {\em arXiv preprint arXiv:2506.13512}, 2025.

\bibitem[BM25b]{blecher2025realII}
David~P. Blecher and Caleb McClure.
\newblock Real noncommutative convexity {II}: {E}xtremality and {NC} convex
  functions.
\newblock {\em arXiv preprint arXiv:2507.22833}, 2025.

\bibitem[BN18]{bluhm2018joint}
Andreas Bluhm and Ion Nechita.
\newblock Joint measurability of quantum effects and the matrix diamond.
\newblock {\em Journal of Mathematical Physics}, 59(11):112202, 2018.

\bibitem[BN20]{bluhm2020compatibility}
Andreas Bluhm and Ion Nechita.
\newblock Compatibility of quantum measurements and inclusion constants for the
  matrix jewel.
\newblock {\em SIAM Journal on Applied Algebra and Geometry}, 4(2):255--296,
  2020.

\bibitem[BN22a]{bluhm2022steering}
Andreas Bluhm and Ion Nechita.
\newblock Maximal violation of steering inequalities and the matrix cube.
\newblock {\em {Quantum}}, 6:656, February 2022.

\bibitem[BN22b]{bluhm2022tensor}
Andreas Bluhm and Ion Nechita.
\newblock A tensor norm approach to quantum compatibility.
\newblock {\em Journal of Mathematical Physics}, 63(6):062201, 2022.

\bibitem[Boh28]{Bohr1928}
Niels Bohr.
\newblock The quantum postulate and the recent development of atomic theory.
\newblock {\em Nature}, 121(3050):580--590, 1928.

\bibitem[BQG{\etalchar{+}}17]{bavaresco2017most}
Jessica Bavaresco, Marco~T{\'u}lio Quintino, Leonardo Guerini, Thiago~O.
  Maciel, Daniel Cavalcanti, and Marcelo~Terra Cunha.
\newblock Most incompatible measurements for robust steering tests.
\newblock {\em Physical Review A}, 96(2):022110, 2017.

\bibitem[BTN02]{ben-tal2002tractable}
Aharon Ben-Tal and Arkadi Nemirovski.
\newblock On tractable approximations of uncertain linear matrix inequalities
  affected by interval uncertainty.
\newblock {\em SIAM Journal on Optimization}, 12(3):811--833, 2002.

\bibitem[DDOSS17]{davidson2016dilations}
Kenneth~R. Davidson, Adam Dor-On, Orr~Moshe Shalit, and Baruch Solel.
\newblock Dilations, inclusions of matrix convex sets, and completely positive
  maps.
\newblock {\em International Mathematics Research Notices},
  2017(13):4069--4130, 2017.

\bibitem[DFK19]{Designolle2019}
S{\'{e}}bastien Designolle, M{\'{a}}t{\'{e}} Farkas, and J{{e}}drzej Kaniewski.
\newblock Incompatibility robustness of quantum measurements: a unified
  framework.
\newblock {\em New Journal of Physics}, 21(11):113053, 2019.

\bibitem[DK25]{Davidson2019Noncommutative}
Kenneth~R. Davidson and Matthew Kennedy.
\newblock Noncommutative {Choquet} theory: A survey.
\newblock In {\em Operator Theory}. Springer, 2025.

\bibitem[EEd{\etalchar{+}}25]{evert2021NCSE}
Eric {Evert}, Aidan {Epperly}, Mauricio {de Oliveira}, John {Yin}, and
  J.~William {Helton}.
\newblock {NCSE} 3.1.2: {A}n {NCA}lgebra package for optimization over free
  spectrahedra, 2025.
\newblock
  \url{https://github.com/NCAlgebra/UserNCNotebooks/tree/master/NCSpectrahedronExtreme}.

\bibitem[EEHK24]{epperly2024matex}
Aidan Epperly, Eric Evert, J.~William Helton, and Igor Klep.
\newblock Matrix extreme points and free extreme points of free spectrahedra.
\newblock {\em Optimization Methods and Software}, 39(6):1263--1308, 2024.

\bibitem[EFHY23]{evert2023empirical}
Eric Evert, Yi~Fu, J.~William Helton, and John Yin.
\newblock Empirical properties of optima in free semidefinite programs.
\newblock {\em Experimental Mathematics}, 32(3):477--501, 2023.

\bibitem[EH19]{evert2019span}
Eric Evert and J.~William Helton.
\newblock {A}rveson extreme points span free spectrahedra.
\newblock {\em Mathematische Annalen}, 375:629--653, 2019.

\bibitem[EHKM18]{evert2018extreme}
Eric Evert, J.~William Helton, Igor Klep, and Scott McCullough.
\newblock Extreme points of matrix convex sets, free spectrahedra, and dilation
  theory.
\newblock {\em The Journal of Geometric Analysis}, 28:1373--1408, 2018.

\bibitem[EP25]{evert2025matrix}
Eric Evert and Benjamin Passer.
\newblock Matrix convex sets over the {E}uclidean ball and polar duals of real
  free spectrahedra.
\newblock {\em arXiv preprint arXiv:2507.20325}, 2025.

\bibitem[EP{\v{S}}24]{evert2024extreme}
Eric Evert, Benjamin Passer, and Tea {\v{S}}trekelj.
\newblock Extreme points of matrix convex sets and their spanning properties.
\newblock {\em arXiv preprint arXiv:2405.07924}, 2024.

\bibitem[Eve21]{evert2021quadrilaterals}
Eric Evert.
\newblock The {A}rveson boundary of a free quadrilateral is given by a
  noncommutative variety.
\newblock {\em Operators and Matrices}, 15:1351--1378, 2021.

\bibitem[Eve25]{evert2024Compact}
Eric Evert.
\newblock Free extreme points span generalized free spectrahedra given by
  compact coefficients.
\newblock {\em Journal of Mathematical Analysis and Applications},
  545(1):129170, 2025.

\bibitem[EW97]{Effros1997}
Edward~G. Effros and Soren Winkler.
\newblock Matrix convexity: Operator analogues of the bipolar and {Hahn-Banach}
  theorems.
\newblock {\em Journal of Functional Analysis}, 144(1):117 -- 152, 1997.

\bibitem[Far00]{farenick2000extremal}
Douglas~R. Farenick.
\newblock Extremal matrix states on operator systems.
\newblock {\em Journal of the London Mathematical Society}, 61:885--892, 2000.

\bibitem[Fin82]{Fine1982}
Arthur Fine.
\newblock {Hidden variables, joint probability, and the Bell inequalities}.
\newblock {\em Physical Review Letters}, 48(5):291--295, 1982.

\bibitem[GCG21]{garstka2021cosmo}
Michael Garstka, Mark Cannon, and Paul Goulart.
\newblock {COSMO}: A conic operator splitting method for convex conic problems.
\newblock {\em Journal of Optimization Theory and Applications},
  190(3):779--810, 2021.

\bibitem[Hei27]{Heisenberg1927}
Werner Heisenberg.
\newblock {{\"U}ber den anschaulichen Inhalt der quantentheoretischen Kinematik
  und Mechanik}.
\newblock {\em Zeitschrift f{\"u}r Physik}, 43(3):172--198, 1927.

\bibitem[HKM12]{convexPsatz}
J.~William Helton, Igor Klep, and Scott McCullough.
\newblock The convex {Positivstellensatz} in a free algebra.
\newblock {\em Advances in Mathematics}, 231(1):516--534, 2012.

\bibitem[HKM13]{helton_matricial_2013}
J.~William Helton, Igor Klep, and Scott McCullough.
\newblock The matricial relaxation of a linear matrix inequality.
\newblock {\em Mathematical Programming}, 138(1-2):401--445, 2013.

\bibitem[HKM17]{helton2017tracial}
J.~William Helton, Igor Klep, and Scott McCullough.
\newblock The tracial {H}ahn-{B}anach theorem, polar duals, matrix convex sets,
  and projections of free spectrahedra.
\newblock {\em Journal of the European Mathematical Society}, 19:1845--1897,
  2017.

\bibitem[HKMS19]{helton2019dilations}
J.~William Helton, Igor Klep, Scott McCullough, and Markus Schweighofer.
\newblock Dilations, linear matrix inequalities, the matrix cube problem and
  beta distributions.
\newblock {\em Memoirs of the American Mathematical Society}, 257(1232), 2019.

\bibitem[HL21]{hartz2021dilation}
Michael Hartz and Martino Lupini.
\newblock Dilation theory in finite dimensions and matrix convexity.
\newblock {\em Israel Journal of Mathematics}, 245:39--73, 2021.

\bibitem[HM04]{HM04}
J.~William Helton and Scott~A. McCullough.
\newblock A positivstellensatz for non-commutative polynomials.
\newblock {\em Transactions of the American Mathematical Society},
  356(9):3721--3737, 2004.

\bibitem[HM12]{HMannals}
J.~William Helton and Scott McCullough.
\newblock Every convex free basic semi-algebraic set has an {LMI}
  representation.
\newblock {\em Annals of Mathematics}, 176(2):979--1013, 2012.

\bibitem[HMZ16]{Heinosaari2016}
Teiko Heinosaari, Takayuki Miyadera, and M\'ario Ziman.
\newblock An invitation to quantum incompatibility.
\newblock {\em Journal of Physics A: Mathematical and Theoretical},
  49(12):123001, 2016.

\bibitem[HSTZ14]{Heinosaari2014}
Teiko Heinosaari, Jussi Schultz, Alessandro Toigo, and M\'ario Ziman.
\newblock Maximally incompatible quantum observables.
\newblock {\em Physical Review A}, 378(24-25):1695--1699, 2014.

\bibitem[HZ11]{Heinosaari2011}
Teiko Heinosaari and M{\'a}rio Ziman.
\newblock {\em The Mathematical Language of Quantum Theory}.
\newblock Cambridge University Press, 2011.

\bibitem[Kle14]{kleski2014boundary}
Craig Kleski.
\newblock Boundary representations and pure completely positive maps.
\newblock {\em Journal of Operator Theory}, 71:45--62, 2014.

\bibitem[Kri19]{kriel2019intro}
Tom-Lukas Kriel.
\newblock An introduction to matrix convex sets and free spectrahedra.
\newblock {\em Complex Analysis and Operator Theory}, 13:3251--3335, 2019.

\bibitem[K{\v{S}}22]{klep2022facial}
Igor Klep and Tea {\v{S}}trekelj.
\newblock Facial structure of matrix convex sets.
\newblock {\em Journal of Functional Analysis}, 283:109601, 2022.

\bibitem[KTT13]{kellner2013containment}
Kai Kellner, Thorsten Theobald, and Christian Trabandt.
\newblock Containment problems for polytopes and spectrahedra.
\newblock {\em SIAM Journal on Optimization}, 23(2):1000--1020, 2013.

\bibitem[Las01]{lasserre2001global}
Jean~B. Lasserre.
\newblock Global optimization with polynomials and the problem of moments.
\newblock {\em SIAM Journal on Optimization}, 11(3):796--817, 2001.

\bibitem[LdL24]{leijenhorst2024solving}
Nando Leijenhorst and David de~Laat.
\newblock Solving clustered low-rank semidefinite programs arising from
  polynomial optimization.
\newblock {\em Mathematical Programming Computation}, 16(3):503--534, 2024.

\bibitem[Nem07]{Nemirovski2007}
Arkadi Nemirovski.
\newblock Advances in convex optimization: conic programming.
\newblock In {\em Proceedings of the International Congress of Mathematicians
  Madrid 2006}, page 413–444. EMS Press, 2007.

\bibitem[PAB{\etalchar{+}}09]{pironio2009device}
Stefano Pironio, Antonio Ac{\'\i}n, Nicolas Brunner, Nicolas Gisin, Serge
  Massar, and Valerio Scarani.
\newblock Device-independent quantum key distribution secure against collective
  attacks.
\newblock {\em New Journal of Physics}, 11(4):045021, 2009.

\bibitem[Pau03]{Paulsen2002}
Vern Paulsen.
\newblock {\em Completely Bounded Maps and Operator Algebras}, volume~78 of
  {\em Cambridge Studies in Advanced Mathematics}.
\newblock Cambridge University Press, 2003.

\bibitem[PNA10]{pironio2010convergent}
Stefano Pironio, Miguel Navascu{\'e}s, and Antonio Acin.
\newblock Convergent relaxations of polynomial optimization problems with
  noncommuting variables.
\newblock {\em SIAM Journal on Optimization}, 20(5):2157--2180, 2010.

\bibitem[PSS18]{passer2018minimal}
Benjamin Passer, Orr~Moshe Shalit, and Baruch Solel.
\newblock Minimal and maximal matrix convex sets.
\newblock {\em Journal of Functional Analysis}, 274:3197--3253, 2018.

\bibitem[RG95]{ramana1995geometric}
Motakuri Ramana and A.~J. Goldman.
\newblock Some geometric results in semidefinite programming.
\newblock {\em Journal of Global Optimization}, 7:33--50, 1995.

\bibitem[WM21]{wang2021exploiting}
Jie Wang and Victor Magron.
\newblock Exploiting term sparsity in noncommutative polynomial optimization.
\newblock {\em Computational Optimization and Applications}, 80(2):483--521,
  2021.

\bibitem[WML21]{wang2021tssos}
Jie Wang, Victor Magron, and Jean-Bernard Lasserre.
\newblock {TSSOS: A moment-SOS hierarchy that exploits term sparsity}.
\newblock {\em SIAM Journal on Optimization}, 31(1):30--58, 2021.

\bibitem[WSV00]{wolkowicz2012handbook}
Henry Wolkowicz, Romesh Saigal, and Lieven Vandenberghe.
\newblock {\em Handbook of semidefinite programming: Theory, Algorithms, and
  Applications}, volume~27 of {\em International Series in Operations Research
  \& Management Science}.
\newblock Springer, 2000.

\bibitem[WW99]{webster1999krienmilman}
Corran Webster and Soren Winkler.
\newblock The {K}rein-{M}ilman theorem in operator convexity.
\newblock {\em Transactions of the American Mathematical Society},
  351:307--322, 1999.

\bibitem[Zal17]{zalar2017operator}
Alja{\v z} Zalar.
\newblock Operator {P}ositivstellens\"{a}tze for noncommutative polynomials
  positive on matrix convex sets.
\newblock {\em Journal of Mathematical Analysis and Applications},
  445(1):32--80, 2017.

\end{thebibliography}

\appendix

\section{Proof of Theorem \ref{thm:SDPSolutions}} \label{sec:Proof-of-SDPSolutions}
For convenience, we restate Theorem \ref{thm:SDPSolutions} before proving it.
\begin{thm} 
    Set $\gamma:= \frac{4}{1+\sqrt{3}}$. Then $\MinBallSDP{\gamma}{X(\theta)}$ is nonempty for all $\theta \in [0,\pi/2]$. 
\end{thm}

\begin{proof}
    The proof is accomplished by explicitly constructing an element of $\MinBallSDP{\gamma}{X(\theta)}$ for arbitrary $\theta \in [0,\pi/2]$. To begin, note that by solving for $C_3,C_4,C_5,C_6$ in equation \eqref{eq:XMinBallSDP}, we find
\begin{equation}
\label{eq:CiEquations}
\begin{array}{lll}
& C_3 = - C_1 -  C_2  + X_3(\theta)/2 +\gamma I/2, \\
& C_4 = -C_1 + X_1/3 +X_2/3 + \gamma I/3, \\
& C_5 = -C_2 - X_1/3 +\gamma I/3, \\
& C_6 = C_1 + C_2 - X_2/3-X_3(\theta)/2- \gamma I /6.
\end{array}
\end{equation}
Our goal now is to choose $C_1$ and $C_2$ so that the expressions for $C_3,C_4,C_5,C_6$ given by equation \eqref{eq:CiEquations} are all positive semidefinite. 

To simplify the computations that occur in the proof, we now separately consider constructions for $\theta \in [0,\pi/8]$ and $\theta \in [\pi/8,\pi/2]$. We begin with the interval $[\pi/8,\pi/2]$, as this is the more computationally intensive region to analyze. In this case, choose 
\[
C_1 = \left(\frac{1}{\sqrt{3}}-\frac{1}{2}\right) \begin{pmatrix}
    1 & 1 \\
    1 & 1
\end{pmatrix} \qquad \qquad \mathrm{and} \qquad \qquad C_2 = \alpha \begin{pmatrix} 1 & \beta \\
\beta & \beta^2
\end{pmatrix},
\]
where
\[
\alpha = \frac{2}{6+4\sqrt{3} +\sqrt{3} \beta^2},
\]
and where $\beta \in \mathbb{R}$ is a parameter depending on $\theta$ that will be chosen later. Then $C_1 \succeq 0$. Additionally, since $\det(C_2) = 0$ and $\mathrm{tr} (C_2) = \alpha(1+\beta^2)$, we have $C_2 \succeq 0$ since $\alpha > 0$. Furthermore, from equation \eqref{eq:CiEquations}, we obtain
\[
C_4 = \frac{1}{6} \begin{pmatrix}
    2\sqrt{3}- 3 & 3 -2 \sqrt{3}\\
    3 - 2 \sqrt{3} & 2\sqrt{3}-3
\end{pmatrix} \succeq 0. 
\]
Additionally we have
\[
C_5 = \frac{1}{3\alpha}\left(
\begin{array}{cc}
 \frac{3}{2} \left(2-\sqrt{3}\right) \beta^2 & -3 \beta \\
 -3 \beta & 6 \left(2+\sqrt{3}\right) \\
\end{array}
\right),
\]
from which we obtain $\det(C_5) = 0$ and 
\[
\mathrm{tr} (C_5) = \frac{\left(2-\sqrt{3}\right) \beta^2 + 6 \left(2+\sqrt{3}\right)}{3\alpha} > 0,
\]
since $\alpha > 0$. We conclude that $C_5 \succeq 0$. 

We now choose $\beta$ so that the determinant of
\[
C_3 = \frac12\left(
\begin{array}{cc}
 \cos (\theta)-2 \alpha+\frac{4}{\sqrt{3}}-1 & \sin (\theta)-2 \alpha \beta-\frac{2}{\sqrt{3}}+1 \\
 \sin (\theta)-2 \alpha \beta-\frac{2}{\sqrt{3}}+1 & -\cos (\theta)-2 \alpha \beta^2+\frac{4}{\sqrt{3}}-1 \\
\end{array}
\right)
\]
is equal to zero. The choices for $\beta$ that give $\det (C_3) = 0$ are 
\[
\beta_+ := \frac{\zeta_1 + \zeta_2}{\zeta_3} \qquad \mathrm{or} \qquad
\beta_- := \frac{\zeta_1 - \zeta_2}{\zeta_3}
\]
where
\begin{align*}
\zeta_1 & = 12 \left(2+\sqrt{3}\right) \sin (\theta)-4 \sqrt{3} \\
\zeta_2 & = \sqrt{6} \sqrt{8 \eta_1 \sin (\theta)+6 \eta_1 \sin (2
   \theta)+6 \eta_2 \cos (\theta)+6 \left(2+\sqrt{3}\right) \cos (2 \theta)+181
   \sqrt{3}+318} \\
   \zeta_3 & = -6 \sin (\theta)+12 \left(2+\sqrt{3}\right) \cos (\theta)+14 \sqrt{3}+21
\end{align*}
and where
\[
\eta_1 = 12+7 \sqrt{3} \qquad \mathrm{and} \qquad \eta_2 = 54+31 \sqrt{3}.
\]
With these choices of $\beta$, noting that 
\[
\mathrm{tr}(C_3) = -\alpha \left(\beta^2+1\right)+\frac{4}{\sqrt{3}}-1 = \frac{ \left(\sqrt{3}-1\right) \beta^2+4 \left(5+3 \sqrt{3}\right)}{\left(3+\sqrt{3}\right)
   \beta^2+2 \left(9+5 \sqrt{3}\right)} \geq 0,
\]
we have that $C_3 \succeq 0$, since $\det(C_3) = 0$ and $\mathrm{tr}(C_3) \geq 0$. 

At this point we have shown that $C_3$, $C_4$, and $C_5$ are all positive semidefinite when $\beta= \beta_+$ or $\beta=\beta_-$. Now, set
\[
C_{6,+} = \frac12 \left(
\begin{array}{cc}
 -\cos (\theta)+2 \alpha+1 & -\sin (\theta)+2 \alpha \beta_+ +\frac{2}{\sqrt{3}}-1 \\
 -\sin (\theta)+2 \alpha \beta_+ +\frac{2}{\sqrt{3}}-1 & \cos (\theta)+2 \alpha \beta_+^2-1 \\
\end{array}
\right)
\]
and
\[
C_{6,-} =\frac12 \left(
\begin{array}{cc}
 -\cos (\theta)+2 \alpha+1 & -\sin (\theta)+2 \alpha \beta_- +\frac{2}{\sqrt{3}}-1 \\
 -\sin (\theta)+2 \alpha \beta_- +\frac{2}{\sqrt{3}}-1 & \cos (\theta)+2 \alpha \beta_-^2-1 \\
\end{array}
\right)
\]
We will argue that for any choice of $\theta \in [\pi/8,\pi/2]$, at least one of $C_{6,+}$ and $C_{6,-}$ is positive semidefinite. Since
\[
\mathrm{tr}(C_{6,-}) =  \alpha(1+\beta_{-}^2) \geq 0 \qquad \mathrm{and} \qquad   \mathrm{tr}(C_{6,+}) =  \alpha(1+\beta_{+}^2) \geq 0,
\]
it is sufficient to show that either
\[
\det(C_{6,-}) \geq 0 \qquad \mathrm{or} \qquad \det(C_{6,+}) \geq 0.
\]
To prove this we instead prove that
\[
\frac{24\det(C_{6,-}) \det(C_{6,+})}{\alpha^2} \leq 0
\]
for all $\theta \in [\pi/8,\pi/2]$. Using Mathematica to simplify, we have
\begin{equation}
\label{eq:C6detProd}
\frac{24\det(C_{6,-}) \det(C_{6,+})}{\alpha^2} = \frac{-54 (1-\sin (\theta)) h(\theta)}{\left(-6 \sin (\theta)+12 \left(2+\sqrt{3}\right) \cos (\theta)+14 \sqrt{3}+21\right)^4}
\end{equation}
where
\[
\begin{array}{rcl}
h(\theta) & = & 3 \left(1659159+957244 \sqrt{3}\right) \sin (\theta)+24 \left(108048+62413 \sqrt{3}\right) \sin
   (2 \theta) \\
   & &-18 \left(84547+48802 \sqrt{3}\right) \sin (3 \theta)-36 \left(48096+27769
   \sqrt{3}\right) \sin (4 \theta) \\
   & &-81 \left(5307+3064 \sqrt{3}\right) \sin (5 \theta)+48
   \left(11401+6598 \sqrt{3}\right) \cos (\theta)\\
   & &+54 \left(5341+3100 \sqrt{3}\right) \cos (2
   \theta)-36 \left(7538+4359 \sqrt{3}\right) \cos (3\theta ) \\
   & &-108 \left(3469+2003 \sqrt{3}\right) \cos
   (4 \theta)-324 \left(362+209 \sqrt{3}\right) \cos (5 \theta)+62 \left(3963+2266 \sqrt{3}\right).
\end{array}
\]

Note that the denominator in equation \eqref{eq:C6detProd} is always positive and that the coefficient $-54(1-\sin(\theta))$ is always negative on $[\pi/8,\pi/2]$, so it is sufficient to show that $h(\theta) \geq 0$ on $[\pi/8,\pi/2]$. 
To accomplish this, we further break the interval into two parts, $[\pi/8,\pi/4]$ and $[\pi/4,\pi/2]$. For $\theta \in [\pi/4,\pi/2]$ the following inequalities hold:
\begin{align*}
\sin(\theta) \geq \frac{1}{\sqrt{2}}, \qquad & \sin(2\theta) \geq 0, \qquad  &-\sin(3\theta) \geq \frac{-1}{\sqrt{2}},\qquad  &-\sin(4\theta) \geq 0, \qquad &-\sin(5t) \geq -1 \\
\cos(\theta)  \geq 0, \qquad & \cos(2\theta) \geq -1, \qquad &-\cos(3\theta) \geq 0, \qquad &-\cos(4\theta) \geq -1, \qquad & -\cos(5\theta) \geq -1.
\end{align*}
It follows that for $t \in [\pi/4,\pi/2]$ we have
\[
\begin{array}{rcl}
h(\theta) & \geq & 3 \left(1659159+957244 \sqrt{3}\right) \frac{1}{\sqrt{2}}+24 \left(108048+62413 \sqrt{3}\right) 0\\
   & &-18 \left(84547+48802 \sqrt{3}\right) \frac{1}{\sqrt{2}}-36 \left(48096+27769
   \sqrt{3}\right) 0\\
   & &-81 \left(5307+3064 \sqrt{3}\right)+48
   \left(11401+6598 \sqrt{3}\right) 0\\
   & &+54 \left(5341+3100 \sqrt{3}\right) (-1)-36 \left(7538+4359 \sqrt{3}\right) 0 \\
   & &-108 \left(3469+2003 \sqrt{3}\right) -324 \left(362+209 \sqrt{3}\right) +62 \left(3963+2266 \sqrt{3}\right),
\end{array}
\]
which simplifies to
\[
h(\theta) \geq -964515+\frac{3455631}{\sqrt{2}}-559132 \sqrt{3}+996648 \sqrt{6} > 0.
\]

Similarly, for $\theta \in [\pi/8,\pi/4]$, we have
\begin{align*}
\sin(\theta) \geq \sin(\frac{\pi}{8}), \quad & \sin(2\theta) \geq \frac{1}{\sqrt{2}}, \quad  &-\sin(3\theta) \geq -1,\quad  &-\sin(4\theta) \geq -1, \quad &-\sin(5t) \geq \cos(\frac{\pi}{8}) \\
\cos(\theta)  \geq \frac{1}{\sqrt{2}}, \quad & \cos(2\theta) \geq 0, \quad &-\cos(3\theta) \geq -\sin(\frac{\pi}{8}), \quad &-\cos(4\theta) \geq 0, \quad & -\cos(5\theta) \geq \sin(\frac{\pi}{8}),
\end{align*}
which in turn gives
\begin{align*}
h(\theta) \geq & 4 \left(-751899+392550 \sqrt{2}-434407 \sqrt{3}+226827 \sqrt{6}\right)\\& +3
   \left(1607799+927508 \sqrt{3}\right) \sin \left(\frac{\pi }{8}\right)
   -81 \left(5307+3064
   \sqrt{3}\right) \cos \left(\frac{\pi }{8}\right)  > 0.
\end{align*}
We conclude that the tuple $(C_1,C_2,\dots,C_6)$ constructed above is an element of $\MinBallSDP{\gamma}{X(\theta)}$ for all $\theta \in [\pi/8,\pi/2]$, so $\MinBallSDP{\gamma}{X(\theta)}$ is nonempty for $\theta$ in this interval.

For $\theta \in [0,\pi/8]$, we can use a simpler choice of $C_1$ and $C_2$. Here it suffices to take 
\[
C_1 = \left(\frac{1}{\sqrt{3}}-\frac{1}{2}\right) \begin{pmatrix}
    1 & 1 \\
    1 & 1
\end{pmatrix} \qquad \qquad \mathrm{and} \qquad \qquad C_2 = \begin{pmatrix} \frac{1}{10} & 0 \\
0 & \frac{7}{50}
\end{pmatrix}.
\]
Since $C_1$ is the same as before, and since $C_4$ does not depend on $C_2$, we immediately have that $C_1,C_2$ and $C_4$ are all positive semidefinite. Additionally, we have
\[
C_5 = \left(
\begin{array}{cc}
 \frac{4}{1+\sqrt{3}}-\frac{13}{10} & 0 \\
 0 & \frac{79}{50}+\frac{4}{1+\sqrt{3}} \\
\end{array}
\right) \succeq 0,
\]
so we only need to check that $C_3$ and $C_6$ are positive semidefinite. 

To accomplish this we show that the trace and determinate of both these matrices are positive. Observe that
\[
\mathrm{tr} (C_3) = -\frac{62}{25}+\frac{8}{\sqrt{3}} > 0  \qquad \mathrm{and} \qquad \mathrm{tr}
(C_6) = 6/25 > 0.
\]
Next, we have 
\[
\det(C_3) = \frac{2}{375} \left(125 \left(2 \sqrt{3}-3\right) \sin (\theta)-15 \cos (\theta)-370
   \sqrt{3}+663\right).
\]
Since $-\cos(\theta)$ and $\sin(\theta)$ are both increasing on $[0,\pi/8]$, this function attains its minimum in $[0,\pi/8]$ at $\theta =0$ and has a minimum value of 
\[
\frac{2}{375} \left(648-370 \sqrt{3}\right) >0. 
\]
We conclude $C_3$ is positive semidefinite on $[0,\pi/8]$. Finally, we have
\[
\det(C_6) = \left(\frac{1}{\sqrt{3}}-\frac{1}{2}\right) \sin (\theta)+\frac{12 \cos
   (\theta)}{25}+\frac{1}{\sqrt{3}}-\frac{787}{750}.
\]
The second derivative with respect to $\theta$ of $\det(C_6)$ is
\[
-\left(\frac{1}{\sqrt{3}}-\frac{1}{2}\right) \sin (\theta)-\frac{12 \cos (\theta)}{25},
\]
which is strictly negative on $[0,\pi/8]$. It follows that $\det(C_6)$ attains its minimum either at $\theta = 0$ or $t = \pi/8$. Evaluating at these points, we find that the minimum value $\det(C_6)$ takes on $[0,\pi/8]$ is 
\[
\frac{1}{\sqrt{3}}-\frac{427}{750} > 0. 
\]
We conclude that $C_6$ is positive semidefinite for all $\theta \in [0,\pi/8]$, which completes the proof. 
\end{proof}

\begin{remark}
    In fact, numerically we observe that $C_{6,-}$ in the proof above is always PSD on the interval $[\pi/8,\pi/2]$. Additionally, while one can show that $C_{6,-}$ is not always PSD on the full interval $[0,\pi/2]$, we numerically observe that $\det(C_{6,-}) \det(C_{6,+}) \leq 0$ on $[0,\pi/2]$, which would imply that for each $\theta \in [0,\pi/2]$, one of $C_{6,-}$ and $C_{6,+}$ is PSD. Our restriction to particular intervals was primarily done to simplify computations.

    We further note if $h(\theta)$ were positive for all $\theta$, then one could in principal compute a Fej\'er-Riesz factorization of $h$ to show it is positive. However, one can compute that 
\[
h\left(\frac{3 \pi}{2}\right) = -32 (202713 + 117038 \sqrt{3}) < 0. 
\]

\end{remark}

\subsection{Alternative point that is conjectured to be feasible}

We conjecture that if one instead takes 
\[
C_1 = \left(\frac{1}{\sqrt{3}}-\frac{1}{2}\right) \begin{pmatrix}
    1 & 1 \\
    1 & 1
\end{pmatrix} 
\qquad \mathrm{and} \qquad  C_2=\frac{1}{12} \left(
\begin{array}{cc}
3 \cos (\theta)+8 \sqrt{3}-9-\beta & 
   3 \sin (\theta)-2 \sqrt{3}+3 \\
 3 \sin (\theta)-2 \sqrt{3}+3 &  -3 \cos (\theta)+8 \sqrt{3}-3 - \beta \\
\end{array}
\right)
\]
where
\[
\beta = \sqrt{3} \sqrt{\left(6-4 \sqrt{3}\right) \sin (\theta)+6 \cos (\theta)-4 \sqrt{3}+13},
\]
then the resulting point is feasible for $\theta \in [0,\pi/2]$, where $C_3,C_4,C_5$, and $C_6$ are defined as in equation \eqref{eq:CiEquations}. As before one can attempt to argue that the traces and determinants of the $C_i$ are positive for $\theta \in [0,\pi/2]$ for each $i=1,\dots,6$. The challenge we encounter with this approach is proving that 
\[
\det(C_2) = \frac{1}{144} \left(\beta\left(\beta-16 \sqrt{3}+12\right)+6 \left(2 \sqrt{3}-3\right) \sin (\theta)+18 \cos (\theta)-84 \sqrt{3}+189\right)
\]
is positive for $\theta \in [0,\pi/2].$

\end{document}